\pdfoutput=1
\RequirePackage{ifpdf}
\ifpdf 
\documentclass[pdftex]{sigma}
\else
\documentclass{sigma}
\fi

\numberwithin{equation}{section}
\newtheorem{thm}{Theorem}[section]
\newtheorem{cor}[thm]{Corollary}
\newtheorem{lem}[thm]{Lemma}
\newtheorem{prop}[thm]{Proposition}
\theoremstyle{definition}
\newtheorem{defn}[thm]{Definition}
\newtheorem{rem}[thm]{Remark}

\usepackage{tikz}

\begin{document}


\newcommand{\arXivNumber}{1601.06179}

\renewcommand{\PaperNumber}{110}

\FirstPageHeading

\ShortArticleName{Commutation Relations and Discrete Garnier Systems}

\ArticleName{Commutation Relations and Discrete Garnier Systems}

\Author{Christopher M.~ORMEROD~$^\dag$ and Eric M.~RAINS~$^\ddag$}
\AuthorNameForHeading{C.M.~Ormerod and E.M.~Rains}
\Address{$^\dag$~University of Maine, Department of Mathemaitcs {\rm \&} Statistics,\\
\hphantom{$^\dag$}~5752 Neville Hall, Room 322, Orono, ME 04469, USA}
\EmailD{\href{mailto:christopher.ormerod@gmail.com}{christopher.ormerod@gmail.com}}
\URLaddressD{\url{http://math.umaine.edu/~ormerod/}}
\Address{$^\ddag$~California Institute of Technology, Mathematics 253-37, Pasadena, CA 91125, USA}
\EmailD{\href{mailto:rains@caltech.edu}{rains@caltech.edu}}
\URLaddressD{\url{http://www.math.caltech.edu/people/rains.html}}

\ArticleDates{Received March 30, 2016, in f\/inal form October 30, 2016; Published online November 08, 2016}

\Abstract{We present four classes of nonlinear systems which may be considered discrete analogues of the Garnier system. These systems arise as discrete isomonodromic deformations of systems of linear dif\/ference equations in which the associated Lax matrices are presented in a factored form. A system of discrete isomonodromic deformations is completely determined by commutation relations between the factors. We also reparameterize these systems in terms of the image and kernel vectors at singular points to obtain a separate birational form. A distinguishing feature of this study is the presence of a symmetry condition on the associated linear problems that only appears as a necessary feature of the Lax pairs for the least degenerate discrete Painlev\'e equations.}


\Classification{39A10; 39A13; 37K15}

{\small \tableofcontents}

\section{Introduction}

Associated with any system of linear dif\/ferential equations is a linear representation of the fundamental group of a sphere punctured at the poles of the system, called the monodromy representation. An isomonodromic deformation is the way in which the system's coef\/f\/icients change while preserving this monodromy representation \cite{Jimbo:Monodromy1}. It is known that all the Painlev\'e equations arise as isomonodromic deformations of second-order dif\/ferential equations~\cite{Jimbo:Monodromy2}. The Garnier system arises as an isomonodromic deformation of a second-order Fuchsian scalar dif\/ferential equation with~$m$ apparent singularities and $m+3$ poles \cite{Garnier}. When we f\/ix three poles, we have $m$ remaining poles that are considered time variables~\cite{Jimbo:Monodromy2}. The simplest nontrivial case, where $m=1$, corresponds to the sixth Painlev\'e equation~\cite{Fuchs2, Fuchs1}.

The focus of this study is a collection of systems that may be regarded as discrete analogues of the Garnier system. We regard these to be nonlinear integrable systems arising as discrete isomonodromic deformations~\cite{Gramani:Isomonodromic}. Our starting point is a regular system of dif\/ference equations of the form
\begin{gather}\label{linearsigma}
\sigma Y(x) = A(x) Y(x),
\end{gather}
where $A(x)$ is a $2\times 2$ matrix polynomial whose determinant is of degree $N$ in $x$, which is called a~spectral variable, and where $\sigma = \sigma_h \colon f(x) \to f(x+h)$ or $\sigma = \sigma_q \colon f(x) \to f(qx)$. These operators are def\/ined in terms of two constants $h, q \in \mathbb{C}$ subject to the constraints $\Re h > 0$ and $0 < | q | < 1$. The goal of this work is to specify a parameterization of these matrices by giving a~factorization,
\begin{gather}\label{prodform}
A(x) = L_1(x) \cdots L_N(x),
\end{gather}
which will be conducive to f\/inding the discrete isomonodromic deformations of~\eqref{linearsigma}.

A discrete isomonodromic deformation is a transformation induced by an auxiliary system of dif\/ference equations, which may be written in matrix form as
\begin{gather}\label{Atildeev}
\tilde{Y}(x) = R(x) Y(x).
\end{gather}
The transformed matrix, $\tilde{Y}(x)$, satisifes a new equation of the form \eqref{linearsigma}, given by
\begin{gather}\label{transformedlinearsigma}
\sigma \tilde{Y}(x) = \tilde{A}(x)\tilde{Y}(x),
\end{gather}
where consistency in the calculation of $\sigma \tilde{Y}(x)$ imposes the relation
\begin{gather}\label{comp}
\tilde{A}(x)R(x) = \sigma R(x) A(x),
\end{gather}
which is compatible with \eqref{linearsigma}. Comparing the left and right-hand sides of \eqref{comp} def\/ines a~rational map between the entries of $A(x)$ and $\tilde{A}(x)$. The two operators appearing in~\eqref{linearsigma} and~\eqref{Atildeev} def\/ine a Lax pair for the resulting map. Compatibility conditions of the form~\eqref{Atildeev} give rise to discrete isomonodromic deformations in the sense of Papageorgiou et al.~\cite{Gramani:Isomonodromic}. It was shown later by Jimbo and Sakai that compatibility relations of the form~\eqref{Atildeev}, as a map between linear systems, also preserves a connection matrix~\cite{Sakai:qP6}. This connection matrix, introduced by Birkhof\/f~\cite{Birkhoff,Birkhoffallied}, is considered to be a discrete analogue of a~monodromy matrix. It is known that various discrete Painlev\'e equations, QRT maps and general classes of integrable mappings that characterize reductions of partial dif\/ference equations arise in this way~\cite{Ormerod2014a,OvdKQ:reductions, Gramani:Isomonodromic}.

Discrete isomonodromic deformations share more in common with Schlesinger transformations than isomonodromic deformations. That is, given an $A(x)$ we have a collection of transformations of the form~\eqref{Atildeev}. In a similar manner to Schlesinger transformations, the system of transformations governed by~\eqref{Atildeev} and~\eqref{comp} has the structure of a f\/initely generated lattice~\cite{Ormerodlattice, Ormerod2011b}. Our discrete Garnier systems are systems of elementary transformations generating an action of~$\mathbb{Z}^d$ for some dimension,~$d$. One of the consequences of~\eqref{prodform}, for the particular choice of~$L_i$ we propose, is that the resulting analogues of elementary Schlesinger transformations are simply expressed in terms of commutation relations between factors. This factorization, and their commutation relations, are also features of the work of Kajiwara et al.~\cite{KajiwaraSym}.

An additional novel feature of our work is the presence of {\em symmetric} Lax pairs, in which solutions of~\eqref{linearsigma} satisfy an extra symmetry constraint. We may take this into consideration by letting $Y(x)$ satisfy relations involving two operators
\begin{subequations}\label{lineartau}
\begin{gather}
\label{tau1}\tau_1 Y(x) = A(x) Y(x),\\
\label{tau2}\tau_2 Y(x) = Y(x),
\end{gather}
\end{subequations}
where $\tau_1 \colon f(x) \to f(-h-x)$ or $\tau_1 \colon f(x) \to f(1/qx)$ while $\tau_2 \colon f(x) \to f(-x)$ or $\tau_2 \colon f(x) \to f(1/x)$. The composition of~$\tau_1$ and~$\tau_2$ recovers~\eqref{linearsigma} while~\eqref{tau2} gives a constraint on the entries of~$A(x)$. The operators~$\tau_1$ and~$\tau_2$ generate a copy of the inf\/inite dihedral group.

The presence of this additional symmetry is a structure that plays an important role in hypergeometric and basic hypergeometric orthogonal polynomials, biorthogonal functions and related special functions. This constraint naturally manifests itself in the known Lax pairs for the elliptic Painlev\'e equation~\cite{rains:isomonodromy, Yamada:E8Lax} via their parameterization in terms of theta functions, but this property has not manifested itself in any obvious way in the known Lax pairs for more degenerate Painlev\'e equations.

There are a number of technical issues in presenting the discrete isomonodromic deformations of such systems: Firstly, the classical theory of Birkhof\/f (see~\cite{Birkhoff, Birkhoffallied}) is no longer suf\/f\/icient to guarantee the existence of solutions. For this we appeal to the work of Praagman~\cite{Praagman:Solutions}. Secondly, the theorems that prescribe discrete isomonodromic deformations do not necessarily preserve the required symmetry. What makes f\/inding the isomonodromic deformations of the symmetric cases tractable is that~$A(x)$ can be shown to admit a~factorization
\begin{gather}\label{AviaB}
A(x) = ( \tau_1 B(x))^{-1} B(x),
\end{gather}
for some rational matrix $B(x)$. By insisting that $B(x)$ takes the same factored form, namely
\begin{gather}\label{prodform2}
B(x) = L_1(x) \cdots L_{N'}(x),
\end{gather}
we \looseness=-1 are able to describe the discrete isomonodromic deformations of these systems in terms of the same commutation relations as the non-symmetric case. Thirdly, since the classical fundamental solutions of Birkhof\/f do not necessarily exist, it is not clear that the analogue of monodromy involving Birkhof\/f's connection matrix (see~\cite{Sakai:qP6}) is appropriate. To address this, we give a short account of how discrete isomonodromic deformations preserve the associated Galois group of~\eqref{linearsigma}.

This gives us four classes of system; two dif\/ference Garnier systems whose associated linear problems are of the form of \eqref{linearsigma}, and two symmetric dif\/ference Garnier systems, whose associated linear problems are of the form of \eqref{lineartau}. We reparameterize these systems in terms of the image and kernel vectors at the singular points. This provides a correspondence between one of our systems and Sakai's $q$-Garnier system~\cite{Sakai:Garnier}. We also consider specializations whose evolution coincides with discrete Painlev\'e equations of type $q$-$\mathrm{P}\big(A_k^{(1)}\big)$ for $k = 0,1,2,3$ and $d$-$\mathrm{P}\big(A_k^{(1)}\big)$ for $k=0,1,2$. The convention we use is that we list the type of the Af\/f\/ine root system associated with the surface of initial conditions~\cite{Sakai:Rational}. This means that the systems we treat appear as the top cases of the disrete Painlev\'e equations and $q$-Painlev\'e equations.

It should be recognised that the phrase symmetric and asymmetric discrete Painlev\'e equations has been applied to equations arising as deautonomized symmetric and asymmetric QRT maps respectively~\cite{Kruskal:AsymmetricdPs}. The way in which the word symmetric is used in the context of this article is that the associated linear problem possesses an additional symmetry. The ideas of having a~symmetric system of dif\/ference equations and having a symmetric QRT mapping or its deautonomization are very dif\/ferent and should not be confused.

The product form for $A(x)$ in \eqref{linearsigma} arises naturally in recent work on reductions of partial dif\/ference equations \cite{Ormerod:qP6, Ormerod2014a}. We present a way in which these systems characterize certain periodic and twisted reductions of the lattice Korteweg--de Vries (KdV) equation and the lattice Schwarzian KdV equation~\cite{Ormerod2014a}. A corollary of this work is that Sakai's $q$-Garnier system arises as a twisted reduction of the lattice Schwarzian KdV equation, as do any specializations. This work also gives an explicit expression for the evolution in terms of known Yang--Baxter maps.

The plan of the paper is as follows. In Section~\ref{sec:isomonodromy} we give an overview of the theory of linear systems of dif\/ference equations where we formalize the way in which we consider our systems to be isomonodromic. In Section~\ref{sec:Garnier} we provide evolution equations for the discrete Garnier systems in terms of viables naturally associated with~\eqref{prodform} and~\eqref{prodform2}, whereas Section~\ref{sec:reparam} gives the same evolution equations in terms of variables associated with the image and kernel at each value of~$x$ in which~$A(x)$ is singular. Section~\ref{sec:special} gives a number of cases in which the evolution of the discrete Garnier systems coincide with known case of discrete Painlev\'e equations. Section \ref{sec:reductions} shows how both cases of the non-symmetric Garnier systems, and their special cases, arise as reductions of discrete potential KdV equation and the discrete Schwarzian KdV equation.

\section{Linear systems of dif\/ference equations}\label{sec:isomonodromy}

This section aims to provide the relevant theorems concerning systems of linear dif\/ference equations. This includes a recapitulation of the classical results of Birkhof\/f on linear systems of dif\/ference and $q$-dif\/ference equations~\cite{Birkhoff, Birkhoffallied}. While the work of Birkhof\/f gives fundamental solutions to systems of dif\/ference equations of the form~\eqref{linearsigma}, they are not suf\/f\/icient to ensure solutions of systems of the form of~\eqref{lineartau}.

Secondly, given a system of the form of~\eqref{linearsigma} or~\eqref{lineartau}, we wish to specify the type of transformations we expect. The set of transformations has the structure of a f\/inite-dimensional lattice. Characterizating these transformations follows the work of Borodin~\cite{Borodin:connection}, who developed this theory in application to gap probabilities of random matrices \cite{Borodin2003}. We extend this to $q$-dif\/ference equations~\cite{Ormerodlattice}.

A secondary issue concerns what structures are being preserved by discrete isomonodromic deformations. The celebrated work of Jimbo and Sakai \cite{Sakai:qP6} argues that \eqref{Atildeev} preserves the connection matrix, when it exists. A fundamental object that is preserved under transformations of the form of \eqref{Atildeev} is the structure of the dif\/ference module~\cite{VanderPut2003}. This provides a more robust def\/inition of what it means to be a discrete isomonodromic deformation. In particular, this holds for discrete isomonodromic deformations of~\eqref{lineartau}, or any system in which the existence of a connection matrix may not be assumed.

\subsection[Systems of linear $h$-dif\/ference equations]{Systems of linear $\boldsymbol{h}$-dif\/ference equations}

We start with \eqref{linearsigma} where $\sigma= \sigma_h$, which we write as
\begin{gather}\label{lineardiff}
Y(x+h) = A(x)Y(x),
\end{gather}
where $A(x)$ is a rational $M \times M$ matrix that is invertible almost everywhere and $\Re h > 0$ as above. We may reduce to the case in which $A(x)$ is polynomial by multiplying $Y(x)$ by gamma functions. This means that the form of $A(x)$ can generally be taken to be
\begin{gather}\label{polydiff}
A(x) = A_0 + A_1 x + \dots + A_nx^n,
\end{gather}
where $A_n \neq 0$. Furthermore, if $A_n$ is invertible and semisimple then, by applying constant gauge transformations, we can assume that $A_n$ is diagonal. By the same argument, if $A_n = I$ and~$A_{n-1}$ is semisimple, we may assume that~$A_{n-1}$ is diagonal. Under these assumptions, it is useful to describe an asymptotic form of formal solutions which is the subject of the following theorem due to Birkhof\/f~\cite{Birkhoff}.

\begin{lem}
Let $A_n = \operatorname{diag}(\rho_1, \ldots, \rho_M)$, where the $\rho_i$ are pairwise distinct, or if $A_n = I$ and $A_{n-1} = \operatorname{diag}(r_1, \ldots, r_M)$ subject to the non resonnancy constraint
\begin{gather*}
r_i - r_j \notin \mathbb{Z}\setminus\{0\},
\end{gather*}
then there exists a unique formal matrix solution of the form
\begin{gather}\label{seriessolhdiff}
\hat{Y}(x) = x^{\frac{nx}{h}} e^{-{nx}{h}}\left(I + \frac{\mathcal{Y}_1}{x}+ \frac{\mathcal{Y}_2}{x^2} +\cdots \right)\operatorname{diag}\left( \rho_1^{x/h} \left(\frac{x}{h}\right)^{d_1}, \ldots, \rho_M^{x/h} \left(\frac{x}{h}\right)^{d_n}\right),
\end{gather}
where $\{d_i\}$ is some set of constants.
\end{lem}

Given a solution of \eqref{lineardiff} that is convergent when $\Re x \gg 0$ or $\Re x \ll 0$, and since $\Re h > 0$, we may use \eqref{lineardiff} to extend the solution by
\begin{gather*}
\hat{Y}(x) = A(x-h)A(x-2h) \cdots A(x-kh) \hat{Y}(x-kh),\\
\hat{Y}(x) = A(x)^{-1}A(x+h)^{-1} \cdots A(x+(k-1)h)^{-1} \hat{Y}(x+kh).
\end{gather*}
This extension introduces possible singularities at translates by integer multiples of $h$ of the points where $\det A(x) = 0$. The values of~$x$ in which $\det A(x) = 0$ play an important role in the theory of discrete isomonodromy, hence, it is useful to parameterize the determinant by
\begin{gather*}
\det A(x) = \rho_1 \cdots \rho_M(x-a_1)(x-a_2) \cdots (x-a_{Mn}).
\end{gather*}

\begin{thm}[see \cite{Borodin:connection}]\label{thm:diffexistence}
Assume that $A_0 = \operatorname{diag}(\rho_1, \ldots, \rho_M)$, with
\begin{gather*}
\rho_i \neq 0, \qquad \rho_i/\rho_j \notin \mathbb{R} \qquad \textrm{for all} \ \ i \neq j,
\end{gather*}
then there exists unique solutions of~\eqref{lineardiff}, $Y_l(x)$ and $Y_r(x)$, such that
\begin{enumerate}\itemsep=0pt
\item[$1.$] The functions $Y_l(x)$ and $Y_r(x)$ are analytic throughout the complex plane except at translates to the left and right by integer multiples of $h$ of the poles of $A(x)$ and $A(x-h)^{-1}$ respectively.
\item[$2.$] In any left or right half-plane, $Y_l(x)$ and $Y_r(x)$ are asymptotically represented by \eqref{seriessolhdiff}.
\end{enumerate}
\end{thm}

Both $Y_l(x)$ and $Y_r(x)$ form a basis for the solutions of \eqref{lineardiff}, which are both non-degenerate in the sense that in the limit as $x \to \pm \infty$, they possess a non-zero determinant, hence, are invertible almost everywhere. The notion that any two non-degenerate solutions of the same dif\/ference equation should be related leads us to the concept of a connection matrix. It should be clear as $Y_l(x)$ and $Y_r(x)$ are both solutions of \eqref{lineardiff}, the connection matrix, def\/ined by
\begin{gather}
P(x) = (Y_{l}(x))^{-1}Y_{r}(x),
\end{gather}
is periodic in $x$ with period $h$. The connection matrix and the monodromy matrices for systems of linear dif\/ferential equations play very similar roles \cite{Borodin:connection}.

Given a linear system of dif\/ference equations, it is useful to talk about the Riemann--Hilbert or monodromy map, which sends the Fuchsian system of dif\/ferential equations to a set of mo\-nodromy matrices~\cite{Boalch2005}. The monodromy matrices depend on the coef\/f\/icients of a given Fuchsian system. The collection of variables that specify the monodromy matrices are called the characteristic constants. The next theorem def\/ines the characteristic constants for systems of dif\/ference equations.

\begin{thm}
Under the general assumptions of Theorem~{\rm \ref{thm:diffexistence}} the entries of connection matrix take the general form
\begin{gather*}
(P(x))_{i,j} = \begin{cases}
p_{i,i}\big( e^{2\pi ix/h} \big) +e^{2\pi i (d_k + x/h)} & \text{if} \ i = j,\\
e^{2\pi i \lambda_{i,j}x/h} p_{i,j}\big(e^{2\pi ix/h}\big) , &
\end{cases}
\end{gather*}
where each $p_{i,j}(x)$ is a polynomial of degree $n-1$ with $p_{i,i}(0) = 1$ and $\lambda_{i,j}$ denotes the least integer as great as the real part of $(\log(\rho_i) - \log(\rho_j))/2\pi i$.
\end{thm}

Perhaps the simplest nontrivial example of such a connection matrix arises from solutions of the one dimension case, which can be broken down into linear factors of the form
\begin{gather*}
y(x+h) = \left(1+\frac{d}{x}\right)y(x).
\end{gather*}

In this way we associate a set of constants to each system of linear dif\/ference equations, by giving a map
\begin{gather*}
(A_0, \ldots, A_n) \mapsto (\{d_k\}; \{p_{i,j}(x)\} ).
\end{gather*}
This gives us $M(Mn+1)$ constants in total, which is also the number of entries in the coef\/f\/icient matrices.

\begin{thm}[Birkhof\/f~\cite{Birkhoff}]
Assume there exists two polynomials $A(x) = A_0 + A_1 x + \cdots + A_nx^n$ and $\tilde{A}(x)=\tilde{A}_0 + \tilde{A}_1 x + \cdots + \tilde{A}_nx^n$ such that $A_n = \tilde{A}_n = \operatorname{diag}(\rho_1, \ldots, \rho_M)$ with the same sets of characteristic constants, then there exists a rational matrix $R(x)$, such that $\tilde{A}(x)$ and $A(x)$ are related by~\eqref{comp} $($where $\sigma = \sigma_h)$. The fundamental solutions of \eqref{linearsigma}, $Y_l(x)$ and $Y_r(x)$, are related to the fundamental solutions of~\eqref{transformedlinearsigma}, denoted~$\tilde{Y}_l(x)$ and $\tilde{Y}_r(x)$, by~\eqref{Atildeev} respectively.
\end{thm}

If we f\/ix $A_n = \operatorname{diag}(\rho_1, \ldots, \rho_M)$, we may denote the algebraic variety of all $n$-tuples, $(A_0, \ldots$, $A_{n-1})$ such that
\begin{gather*}
\det\big( A_n x^n + A_{n-1} x^{n-1} + \cdots + A_0\big) = \prod_{j=1}^{M} \rho_j \prod_{k = 1}^{Mn}(x-a_k),
\end{gather*}
 by $\mathcal{M}_h(a_1, \ldots, a_{Mn}; d_1, \ldots, d_M;\rho_1, \ldots, \rho_M)$. The space of discrete isomonodromic deformations is characterized by the following theorem of Borodin~\cite{Borodin:connection}.

\begin{thm}\label{hisolattice}
For any $\epsilon_1, \ldots, \epsilon_{Mn} \in \mathbb{Z}$, $\delta_1, \ldots, \delta_M \in \mathbb{Z}$ such that
\begin{gather*}
\sum_{i=1}^{Mn} \epsilon_i + \sum_{j=1}^{M} \delta_j = 0,
\end{gather*}
there exists a non-empty Zariski open subset, $\mathcal{A} \subset \mathcal{M}_h(a_1, \ldots, a_{Mn}; d_1, \ldots, d_M;\rho_1, \ldots, \rho_M)$, such that for any $(A_0,\ldots, A_{n-1}) \in \mathcal{A}$ there exists a unique rational matrix, $R(x)$ and a matrix, $\tilde{A}(x)$ related by \eqref{comp} such that
\begin{gather*}
\tilde{A} = \tilde{A}_0 + \tilde{A}_1x + \cdots + \tilde{A}_{n-1}x^{n-1} + A_nx^n, \\
(\tilde{A}_0, \ldots, \tilde{A}_{n-1}) \in \mathcal{M}_h(a_1 + \epsilon_1, \ldots,a_{Mn} + \epsilon_{Mn}; d_1 + \delta_1, \ldots, d_M + \delta_M;\rho_1, \ldots, \rho_M),
\end{gather*}
and where $\tilde{Y}_{\pm \infty}(x)$ are related to $Y_{\pm \infty}(x)$ by~\eqref{Atildeev}.
\end{thm}

Fixing some translation, we obtain that \eqref{lineardiff} and \eqref{Atildeev} are a Lax pair for a birational map of algebraic varieties
\begin{gather*}
\phi \colon \ \mathcal{M}_h (a_1, \ldots, a_{Mn}; d_1, \ldots, d_M;\rho_1, \ldots, \rho_M) \\
\hphantom{\phi \colon}{} \ {}\to \mathcal{M}_h(a_1 + \epsilon_1, \ldots,a_{Mn} + \epsilon_{Mn}; d_1 + \delta_1, \ldots, d_M + \delta_M;\rho_1, \ldots, \rho_M),
\end{gather*}
which we wish to identify as some integrable system. For reasons of simplif\/ication, if $A_n = I$ we will also assume that $A_{n-1}$ is semisimple, in which case we we may apply a constant gauge transformation so that $A_{n-1}$ is also diagonal. Hence, if $A_n = I$, we will impose the condition that $A_{n-1}$ is diagonal for any tuple in $\mathcal{M}_h(a_1, \ldots, a_{Mn}; d_1, \ldots, d_M;1, \ldots, 1)$.

\subsection[Classical $q$-dif\/ference results]{Classical $\boldsymbol{q}$-dif\/ference results}

We write \eqref{linearsigma} where $\sigma= \sigma_q$ as
\begin{gather}\label{linearqdiff}
Y(qx) = A(x)Y(x),
\end{gather}
where $A(x)$ is a rational $M\times M$ matrix that is invertible almost everywhere and $0 < |q| < 1$ as above.

The functions required to express the solution of any scalar linear f\/irst-order system of $q$-dif\/ference equations are not as commonly used as the gamma function. Hence, before discussing some of the particular existence theorems, let us introduce some standard functions, all of which may be found in \cite{GasperRahman}. We def\/ine the $q$-Pochhammer symbol by
\begin{gather*}
( a;q )_{\infty} = \prod_{n=0}^{\infty} \big(1-aq^{n}\big),\qquad
( a;q )_m = \frac{(a;q)_{\infty}}{(aq^{m};q)_{\infty}}.
\end{gather*}
The important property of $(x;q)_{\infty}$ is that
\begin{gather*}
(qx;q)_{\infty} = \frac{(x;q)_{\infty}}{1-x}.
\end{gather*}
We also have the Jacobi theta function,
\begin{gather*}
\theta_q(x) = \sum_{n=-\infty}^{\infty} q^{\frac{n(n-1)}{2}} x^{n},
\end{gather*}
which is analytic over $\mathbb{C}^{*}$ and satisf\/ies
\begin{gather*}
\theta_q(qx) = qx\theta_q(x),\qquad
\theta_q(x) = (q;q)_{\infty}(-xq;q)_{\infty}(-q/x;q)_{\infty}.
\end{gather*}
The last expression is known as the Jacobi triple product identity~\cite{GasperRahman}. The function $\theta_q(x)$ has simple roots on~$-q^{\mathbb{Z}}$. We def\/ine the $q$-character to be
\begin{gather*}
e_{q,c}(x) = \frac{\theta_q(x)}{\theta_q (x/c)},
\end{gather*}
which satisf\/ies $e_{q,c}(qx) = ce_{q,c}(x)$ and has simple zeroes at $x=q^{\mathbb{Z}}$ and simple poles at $x=cq^{\mathbb{Z}}$. In the special case in which $c = q^n$ then $e_{q,q^n}(x)$ is proportional to~$x^n$. Lastly, we have the $q$-logarithm
\begin{gather*}
l_q(x) = x\frac{\theta_q'(x)}{\theta_q(x)},
\end{gather*}
which satisf\/ies
\begin{gather*}
\sigma_q l_q(x) = l_q(x)+1,
\end{gather*}
and is meromorphic over $\mathbb{C}^{*}$ with simple poles on $q^{\mathbb{Z}}$.

We may use the functions above to solve any scalar $q$-dif\/ference equation, hence transform~\eqref{linearqdiff} in which~$A(x)$ is rational to a~case in which~$A(x)$ is polynomial, given by~\eqref{polydiff}. If $A_0$ and $A_n$ are semisimple and invertible, then by using constant gauge transformations we can assume that one of them, say~$A_n$, is diagonal. Under these circumstances, we may specify two formal solutions.

\begin{lem}[Birkhof\/f \cite{Birkhoffallied}]\label{lemq}
Suppose $A_n = \operatorname{diag}(\kappa_1, \ldots, \kappa_M)$ and $A_0$ is semisimple with non-zero eigenvalues, $\theta_1, \ldots, \theta_M$, such that the non resonnancy conditions
\begin{gather*}
\frac{\kappa_i}{\kappa_j}, \ \frac{\theta_i}{\theta_j} \neq q, \ q^2, \ \ldots
\end{gather*}
are satisfied, then there exists two formal matrix solutions
\begin{subequations}
\begin{gather}\label{basicqdif}
Y_0(x) = \left(\mathcal{Y}_0 + \mathcal{Y}_1x + \cdots \right) \operatorname{diag}(e_{q,\theta_1}, \ldots, e_{q,\theta_M}),\\
Y_{\infty}(x) = \left(I + \frac{\mathcal{Y}_{-1}}{x} + \cdots \right) \operatorname{diag}(e_{q,\kappa_1}, \ldots, e_{q,\kappa_M}) \theta_q(x/q)^{-n},
\end{gather}
\end{subequations}
where $\mathcal{Y}_0$ diagonalizes $A_0$.
\end{lem}

By using $\theta_q$ as our building block for the multiplicative factors appearing on the right in~\eqref{basicqdif}, this formulation is slightly dif\/ferent from the original formulation of Birkhof\/f~\cite{Birkhoffallied}. These functions have nicer properties with respect to the Galois theory of dif\/ference equations~\cite{Sauloy, VanderPut2003}. We should mention that the above form can be generalized to the case in which some of the eigenvalues are $0$ by using the so-called Birkhof\/f--Guenther form~\cite{BirkhoddAdamsSum}. Formal solutions def\/ined in terms of the Birkhof\/f--Guenther form do not necessarily def\/ine convergent solutions. This issue of convergence gives rise to a $q$-analogue of the Stokes phenomenon for systems of linear dif\/ferential equations~\cite{FlaschkaNewell}. Regardless of the convergence, these solutions may be used to derive deformations of the form~\eqref{Atildeev}, as shown in~\cite{Ormerodlattice}.

We are interested in solutions def\/ined in open neighborhoods of $x= 0$ and $x= \infty$, which may be extended by
\begin{gather*}
Y_{\infty}(x) = A(x/q)A\big(x/q^2\big) \cdots A\big(x/q^k\big) Y_{\infty}\big(x/q^k\big),\\
Y_0(x) = A(x)^{-1}A(qx)^{-1} \cdots A\big(q^{k-1} x\big)^{-1} Y_0\big(q^kx\big).
\end{gather*}
The resulting solutions are singular at $q$-power multiples of the values of~$x$ where $\det A(x) = 0$. For this reason, it is once again convenient to f\/ix where~$A(x)$ is not invertible. If~$A_n$ is semisimple with eigenvalues $\kappa_1, \ldots, \kappa_M$ ($\kappa_i \neq 0$) then we parameterize the determinant as
\begin{gather}\label{qdetA}
\det A(x) = \kappa_1 \cdots \kappa_M (x-a_1) \cdots (x-a_{Mn}).
\end{gather}
The series part of the solution around $x= 0$, which we denote $\hat{Y}_0(x)$, satisf\/ies
\begin{gather*}
\hat{Y}_0(qx) A_0 = A(x)\hat{Y}_0(x),
\end{gather*}
whereas the series part of the solution around $x=\infty$, denoted $\hat{Y}_{\infty}(x)$, satisf\/ies a similar equation. By a succinct argument featured in van der Put and Singer~\cite[Section~12.2.1]{VanderPut2003} we have solutions, $\hat{Y}_0(x)$ and $\hat{Y}_{\infty}(x)$, that are convergent in neighborhoods of $x=0$ and $x=\infty$ respectively.

\begin{prop}The series part of the solutions, $\hat{Y}_0(x)$ and $\hat{Y}_{\infty}(x)$, specified by~\eqref{basicqdif} are holomorphic and invertible at $x=0$ and $x = \infty$ respectively. Moreover, $\hat{Y}_{\infty}(x)$ and $\hat{Y}_0(x)^{-1}$ have no poles, while $Y_{\infty}(x)^{-1}$ and $\hat{Y}_0(x)$ have possible poles at $q^{k+1}a_i$ and $q^{-k}a_i$ respectively for $i = 1,\ldots, mn$ and $k \in \mathbb{N}$.
\end{prop}

While we do not have an explicit presentation of the connection matrix, it is generally known to be expressible in terms of elliptic theta functions. In particular, the entries of the connection matrix lie in the f\/ield of meromorphic functions on an elliptic curve, i.e., $\mathbb{C}^{*}/\langle q \rangle$.

Let us specify the required lattice actions in a similar way to the $h$-dif\/ference case. We denote the algebraic variety of all $n$-tuples of $M\times M$ matrices, $(A_1, \ldots, A_n)$ such that $A_n$ has eigenvalues $\kappa_1,\ldots, \kappa_M$ and $A_0 = \operatorname{diag}(\theta_1, \ldots, \theta_M)$ with determinant specif\/ied by \eqref{qdetA} by $\mathcal{M}_q(a_1, \ldots, a_{Mn}; \kappa_1, \ldots, \kappa_M; \theta_1, \ldots, \theta_M)$. The natural constraint obtained by evaluating \eqref{qdetA} at $x=0$, is that
\begin{gather}\label{qgenprodcon}
\prod_{j=1}^{m} \kappa_j \prod_{k = 1}^{Mn} a_k (-1)^{Mn} =\prod_{j=1}^{M} \theta_j.
\end{gather}

\begin{thm}\label{qlattice}
For any $\epsilon_1, \ldots, \epsilon_{Mn} \in \mathbb{Z}$ and $\delta_1, \ldots, \delta_M \in \mathbb{Z}$ such that
\begin{gather*}
\sum_{j=1}^{Mn} \epsilon_{j} + \sum_{i=1}^{M} \delta_i = 0,
\end{gather*}
there exists a non-empty Zariski open subset, $\mathcal{A} \subset \mathcal{M}_q(a_1, \ldots, a_{Mn}; \kappa_1, \ldots, \kappa_M; \theta_1, \ldots, \theta_M)$, such that for any $(A_1,\ldots, A_{n-1}) \in \mathcal{A}$ there exists a rational matrix, $R(x)$, and a matrix, $\tilde{A}(x)$ related by \eqref{comp} with
\begin{gather*}
\tilde{A}(x) = A_0 + \tilde{A}_1x + \cdots + \tilde{A}_{n-1}x^{n-1} + \tilde{A}_nx^n, \\
(\tilde{A}_1, \ldots, \tilde{A}_{n-1}) \in \mathcal{M}_q\big(a_1q^{\epsilon_1}, \ldots,a_{mn} q^{\epsilon_{mn}}; \kappa_1q^{\delta_1}, \ldots, \kappa_Mq^{\delta_M}; \theta_1, \ldots, \theta_M\big),
\end{gather*}
and where $\tilde{Y}_{0}(x)$ and $\tilde{Y}_{\infty}(x)$ are related to $Y_{0}(x)$ and $Y_{\infty}(x)$ by~\eqref{Atildeev}.
\end{thm}

\begin{proof}
It is suf\/f\/icient to specify an atomic operation that performs the following invertible operation
\begin{gather*}
e_{1,1} \colon \ \kappa_1 \to \kappa_1/q, \qquad a_1 \to qa_1,
\end{gather*}
which when composed with actions that permutes $a_1,\ldots, a_{Mn}$ and $\kappa_1, \ldots, \kappa_n$ give us all transformations we require. A matrix that does this is found by using a constant gauge transformation to change the basis so that the vectors
\begin{gather*}
\ker(A(a_1)), \ \ker(A_m - \kappa_2), \ \ldots, \ \ker(A_m - \kappa_M)
\end{gather*}
are the new coordinate vectors. We then perform a gauge transformation of the form \eqref{comp} whose ef\/fect is dividing the f\/irst column by~$(1-x/a_1)$ and multiplying the f\/irst row by~$(1- x/a_1q)$. Reverting back to a basis in which~$A_0$ is the constant coef\/f\/icient matrix using another constant matrix gives the required matrix. It should be clear from the determinant that $a_1 \to q a_1$, while looking at $\tilde{A}(x)$ asymptotically around $x=\infty$ it is clear $\kappa_1 \to \kappa_1/q$. Since all these steps were invertible, the inverse atomic operation is also rational, hence, we obtain all possible transformations this way.
\end{proof}

Systems of linear $q$-dif\/ference equations can also be treated as discrete connections, where the matrix presentations of these systems of linear $q$-dif\/ference equations arise as trivializations of linear maps between the f\/ibres of a vector bundle. In this framework, the theorem above may also be deduced by purely geometric means, as was done in the $h$-dif\/ference case in~\cite{Arinkin2006}. The $q$-dif\/ference version of this framework was the subject of a recent paper by Kinzel~\cite{Knizel2015}.

\begin{rem}
The elementary translations are those that multiply any collection of up to $m$ of the $a_i$'s by $q$ and multiply the same number of $\kappa_j$'s by $q^{-1}$. For the applications that follow, this formulation will be suf\/f\/icient, however, this is a slightly less general result than possible. One may generally f\/ind a rational matrix in which the $\theta_i$ values are shifted by $q$-powers in a~way that preserves \eqref{qgenprodcon}.
\end{rem}

\subsection{Dif\/ference equations and vector bundles}

The aim of this section is to present the theorems required for the existence of meromorphic solutions to \eqref{lineartau}, which we write as two cases:
\begin{gather*}
Y(-x-h) = A(x)Y(x)\qquad \text{and} \qquad Y(x) = Y(-x),
\end{gather*}
or
\begin{gather*}
Y(1/qx) = A(x)Y(x)\qquad \text{and} \qquad Y(x) = Y(1/x).
\end{gather*}
To prove the general existence of solutions with these symmetry properties, we turn to some general results concerning sheaves on compact Riemann surfaces (see~\cite{Forster1980} for example). For a~connected Riemann surface,~$\Sigma$, we may denote the sheaves of holomorphic and meromorphic functions on~$\Sigma$ by $\mathcal{O}_{\Sigma}$ and $\mathcal{M}_{\Sigma}$ respectively. A holomorphic or meromorphic vector bundle of rank~$n$ is a~sheaf of $\mathcal{O}_{\Sigma}$-modules or $\mathcal{M}_{\Sigma}$-modules which is locally isomorphic to~$\mathcal{O}_{\Sigma}^n$ or $\mathcal{M}_{\Sigma}^n$ respectively.

\begin{thm}[{\cite[Theorem~3]{Praagman:Solutions}}]\label{Praagman}
Let $G$ be a group of automorphisms of $\mathbb{P}_1$, $L$ is the limit set of~$G$ and~$U$ a component of $\mathbb{P}_1\setminus L$ such that $G(U) = U$. If there is a map, $G \to \mathrm{GL}_M(\mathcal{M}_U)$, $g \to A_g(x)$ satisfying
\begin{gather*}
A_{gh}(x) = A_g(h(z)) A_h(z),
\end{gather*}
then the system of equations
\begin{gather*}
Y(\gamma(z)) = A_{\gamma}(z) Y(z) ,\qquad \gamma \in G,
\end{gather*}
possesses a meromorphic solution.
\end{thm}

The two important examples in the context pertain to the case in which $G$ is a group of automorphisms of $\mathbb{P}_1$ admitting the presentation
\begin{gather*}
G = \big\langle \tau_1,\tau_2 \,|\, \tau_1^2 = \tau_2^2 = 1 \big\rangle,
\end{gather*}
which is often called the inf\/inite dihedral group. In particular, we are interested in the case in which the groups of automorphisms are
\begin{gather*}
G_h = \langle \tau_1,\tau_2 \,|\, \tau_1(x) = -h-x, \, \tau_2(x) = -x \rangle,\\
G_q = \left\langle \tau_1, \tau_2 \,|\, \tau_1(x) = \frac{1}{qx}, \, \tau_2(x) = \frac{1}{x} \right\rangle.
\end{gather*}
If we let $A_{\tau_2} = I$ in each case and $A_{\tau_1}(x)$ be some rational matrix, $A(x)^{-1}$, the commutation relation on $\tau_1$ and $\tau_2$ requires
\begin{gather*}
A(x) = A(-h-x)^{-1} \qquad \text{or}\qquad A(x) = A\left(\frac{1}{qx}\right)^{-1},
\end{gather*}
respectively.

\begin{lem}\label{lem:existencegenatu}
Let $\mathbb{L}/\mathbb{K}$ be a quadratic field extension and $A \in \mathrm{GL}_n(\mathbb{L})$ be a matrix such that $\bar{A}=A^{-1}$, where $\bar{A}$ is the conjugation of $A$ in $\mathbb{L}$ over $\mathbb{K}$. Then there exists a matrix $B \in \mathrm{GL}_n(\mathbb{L})$ such that $A = \bar{B}B^{-1}$ and $B$ is unique up to right-multiplication by $\mathrm{GL}_n(\mathbb{K})$.
\end{lem}

\begin{proof}
Given a vector $w \in \mathbb{L}^n$, it is easy to see that if $v = \bar{w} + A^{-1} w$ then $\bar{v} = w + A\bar{w} = A v$. Applying to any basis for $\mathbb{L}^n$ over $\mathbb{K}$ gives at least $n$ vectors satisfying $\bar{v} = Av$ that are linearly independent over $\mathbb{K}$, whose columns give a matrix $B$ such that
\begin{gather}\label{hil90proof}
\bar{B} = AB.
\end{gather}
For uniqueness we suppose two such matrices, $B_1$ and $B_2$, satisfy~\eqref{hil90proof}, then $C = B_1 B_2^{-1}$ satisf\/ies $C = \bar{C}$, in which case $C \in \mathrm{GL}_n(\mathbb{K})$.
\end{proof}

\begin{rem}
This lemma is a special case of what is often called ``Hilbert's theorem 90", which states that any $1$-cocycle of a Galois group with values in
$\mathrm{GL}_n$ is trivial. Hilbert dealt with the case in which $\mathrm{Gal}(\mathbb{L}/\mathbb{K})$ is cyclic, and $n=1$.
\end{rem}

Specializing to the function f\/ields $\mathbb{L}= \mathbb{C}(x)$ and $\mathbb{K}$ is the subf\/ield of rational functions invariant under $x \to -x-h$ or $x \to 1/x$ allows us to write $A(x)$ as one of two cases;
\begin{subequations}
\begin{gather}
\label{symmdiff}A(x) = B(-h-x) B(x)^{-1},\\
\label{symmqdiff}A(x) = B\left( \frac{1}{qx} \right) B(x)^{-1},
\end{gather}
\end{subequations}
where $B(x)$ is rational. This reduces the problem of determining the algebraic variety of all $n$-tuples of matrices with a symmetry condition to determining $n$-tuples of matrices with prescribed properties. In particular, it makes sense to let
\begin{gather*}
B(x) = B_0 + B_1 x + \cdots + B_nx^n,
\end{gather*}
where either
\begin{gather*}
(B_0,\ldots, B_{n-1}) \in \mathcal{M}_h(a_1, \ldots, a_{Mn}; d_1, \ldots, d_M;\rho_1, \ldots, \rho_M),\\
(B_0,\ldots, B_{n-1}) \in \mathcal{M}_q(a_1, \ldots, a_{Mn}; \kappa_1, \ldots, \kappa_M;\theta_1,\ldots, \theta_M).
\end{gather*}
In discussing the isomonodromic deformations, we specify two dif\/ferent types of discrete isomonodromic deformations; those that act on the left and those that act on the right, which are given as follows
\begin{subequations}\label{symlax}
\begin{gather}
\label{leftRB}\tilde{B}(x) = \lambda(x) R_l(x) B(x),\\
\label{rightRB}\tilde{B}(x) = \lambda(x) B(x) R_r(x),
\end{gather}
\end{subequations}
where $\lambda(x)$ is some rational scalar factor. This scalar factor only swaps poles and roots of the determinant and should be considered trivial from the perspective of integrability. These two equations should be thought of as the symmetric equivalent of~\eqref{comp}. We may rigidify the def\/initions of $R_l(x)$ or $R_r(x)$ by insisting that these matrices are proportional to identity matrices around $x=\infty$.

If we insist that $R_l(x)$ is invariant under $\tau_2$, i.e., we have the symmetry $R_l(x) = \tau_2 R_l(x)$, then it is clear that a transformation of the form~\eqref{leftRB} coincides with a transformation of the form~\eqref{comp}, hence, will be considered a discrete isomonodromic deformation. Furthermore, if we may f\/ind such a matrix, Theorem~\ref{hisolattice} or Theorem~\ref{qlattice}, depending on the case, tells us that this matrix and resulting transformation are unique, hence, the discrete isomonodromic deformation does preserve the required symmetry.

\subsection{Preserving the Galois group}\label{sec:Galois}

The main reason for passing from connection preserving deformations to the Galois theory of dif\/ference equations is that we have not shown that systems of the form \eqref{lineartau} possess connection matrices. While mechanically, we still have Lax pairs using \eqref{comp} or \eqref{symlax}, the implications of possessing a discrete Lax pair of any form are not generally known. We wish to show that \eqref{comp} and \eqref{symlax} preserve the associated Galois group.

This is an issue that is not conf\/ined to symmetric Lax pairs. Various Painlev\'e equations are known to arise as relations of the form of \eqref{comp} where the series part of the formal solutions at $x=\pm \infty$ or $x= 0$ are not convergent \cite{OrmerodqPV, Ormerodlattice}. From an integrable systems perspective, it is useful to know precisely what is preserved, and it turns out the associated dif\/ference module is always preserved under transformations of the form \eqref{Atildeev}. We require some of the formalism described in~\cite{VanderPut2003} to demonstrate this.

\begin{defn}
A dif\/ference ring is a commutative ring/f\/ield, $R$, with $1$, together with an automorphism $\sigma \colon R \to R$. The constants, denoted $C_R$ are the elements satisfying $\sigma(f) = f$. An ideal of a dif\/ference ring is an ideal,~$I$, such that $\phi(I) \subset I$. If the only dif\/ference ideals are~$0$ and~$R$ then the dif\/ference ring is called simple.
\end{defn}

This is a natural discrete analogue of a dif\/ferential f\/ield. In Picard--Vessiot theory, a Picard--Vessiot extension is formed by extending the f\/ield of constants by the solutions of a homogenous linear ordinary dif\/ferential equation \cite{VanderPut2003}. The analogue of this for dif\/ference equations is the following construction.

\begin{defn}
Let $\mathbb{K}$ be a dif\/ference f\/ield and \eqref{linearsigma} be a f\/irst-order system with $A(x) \in \mathrm{GL}_n(\mathbb{K})$. We call a $\mathbb{K}$-algebra, $R$, a Picard--Vessiot ring for \eqref{linearsigma} if:
\begin{enumerate}\itemsep=0pt
\item[1)] an extension of $\sigma$ to $R$ is given,
\item[2)] $R$ is a simple dif\/ference ring,
\item[3)] there exists a solution of \eqref{linearsigma} with coef\/f\/icients in $R$,
\item[4)] $R$ is minimal in the sense that no proper subalgebra satisf\/ies $1$, $2$ and $3$.
\end{enumerate}
\end{defn}

We are treating $\mathbb{C}(x)$ as a dif\/ference f\/ield where $\sigma_h$ and $\sigma_q$ are the relevant automorphisms. The f\/ield of constants contain $\mathbb{C}$ extended by the $\sigma$-periodic functions (e.g., $e^{2 i \pi x/h}$ and $\phi_{c,d} = e_{q,c}e_{q,d}/e_{q,cd}$). We may formally construct a Picard--Vessiot ring for \eqref{linearsigma} by considering a~mat\-rix of inderminants, $Y(x) = (y_{i,j}(x))$. We extend $\sigma$ to $\mathbb{K}(Y)$ via the entries of \eqref{linearsigma}. If $I$ is a maximal dif\/ference ideal, then we obtain a Picard--Vessiot ring for \eqref{linearsigma} by considering the quotient $\mathbb{K}(Y)/ I$. This quotient by a maximal dif\/ference ideal ensures the resulting construction is a simple dif\/ference ring.

This formal construction may be replaced by a fundamental system of meromorphic solutions of either \eqref{linearsigma} or \eqref{lineartau} specif\/ied by Theorem \ref{lem:existencegenatu}. For $q$-dif\/ference equations, in general (see~\cite{vanderPut}) the entries of any solution are elements of the f\/ield $\mathcal{M}(\mathbb{C})(l_q, (e_{q,c})_{c \in \mathbb{C}^*})$.

\begin{defn}
If $R$ is a Picard--Vessiot ring for~\eqref{linearsigma}, the Galois group, $G = \mathrm{Gal}(R/C_R)$ is the group of automorphisms of~$R$ commuting with $\sigma$.
\end{defn}

Let us brief\/ly describe the role of the connection matrix in this context. We have given conditions for there to exist two fundamental solutions, which we will call $Y_1(x)$ and $Y_2(x)$, which are distinguished by the regions of the complex plane in which they def\/ine meromorphic functions. If we adjoin the entries of $Y_1(x)$ or $Y_2(x)$ we describe two Picard--Vessiot extensions, denoted $R_1$ and $R_2$. We expect $R_1$ and $R_2$ to be isomorphic to the formal construction above, in particular, there exists an isomorphism between $R_1$ and $R_2$. The connection matrix, $P(x)$, relates solutions via
\begin{gather*}
Y_1(x) = P(x)Y_2(x),
\end{gather*}
which def\/ines such an isomorphism between $R_1$ and $R_2$. For any generic value of $x$ for which~$P(x)$ is def\/ined, $P(x)$ describes a connection map, which is an isomorphism of Picard--Vessiot extensions, hence, for generic values of $u$ and $v$ for which the connection matrix is def\/ined, the matrix $P(u)P(v)^{-1}$ def\/ines an automorphism of $R_1$. In the case of regular systems of $q$-dif\/ference equations, it is a result of Etingof that the Galois group is a linear algebraic group over $\mathbb{C}$ that is generated by matrices of the form $P(u)P(v)^{-1}$ for $u,v \in C$ where def\/ined~\cite{Etingof1995}. This mirrors dif\/ferential Galois theory where it is generally known that the dif\/ferential Galois group is generated by the monodromy matrices, the Stokes matrices and the exponential torus~\cite{Martinet1989}. More generally, this relation between values of the connection matrix and the Galois group has been the subject of works of a number of authors \cite{Sauloy, vanderPut}.

We may generalize the def\/inition of the Galois group from a category theoretic perspective. Given a dif\/ference f\/ield, $\mathbb{K}$ (e.g., $\mathbb{C}(x)$), with a dif\/ference operator $\sigma$, we can consider the ring of f\/inite sums of dif\/ference operators in a new operator,~$\phi$,
\begin{gather*}
k[\phi, \phi^{-1}] = \left\{ \sum_{n \in \mathbb{Z}} a_n \phi^n \right\},
\end{gather*}
where $\phi$ is def\/ined by the relation $\phi(\lambda) = \sigma(\lambda) \phi$ for $\lambda \in \mathbb{K}$. We can consider the category of left modules, $M$, over $\mathbb{K}$. Under a suitable basis, we may identify $M$ with $\mathbb{K}^m$. In this basis, the action of $\phi$ is identif\/ied with a matrix by
\begin{gather}\label{moduledef}
\phi Y = A \sigma Y.
\end{gather}
Conversely, given a dif\/ference equation of the form $\sigma Y = AY$, we may endow $\mathbb{K}^m$ with the structure of a dif\/ference module via~\eqref{moduledef}.

\begin{thm}\label{diffmods}
Two systems, $\sigma Y(x) = A(x)Y(x)$ and $\sigma \tilde{Y}(x) = \tilde{A}(x)\tilde{Y}(x)$ define isomorphic difference modules if and only if the matrices $A(x)$ and $\tilde{A}(x)$ are related by~\eqref{comp}.
\end{thm}

The object that is being preserved under these deformations is the local system/sheaf of solutions. We could also call these transformations isomodular since the dif\/ference module is preserved.

The advantage of this def\/inition is that the category of dif\/ference modules over a dif\/ference f\/ield is a rigid abelian tensor category. We may use the def\/initions of \cite{Deligne1981} to def\/ine the Galois group from a category theoretic perspective. While it is dif\/f\/icult to see a priori that a transformation of the form \eqref{Atildeev} necessarily preserves the Galois group, from the perspective of the category theory, isomorphic dif\/ference modules resulting from Theorem \ref{diffmods} yield isomorphic Galois groups.

\begin{cor}\label{corintegrability}
Two systems, $\sigma Y(x) = A(x)Y(x)$ and $\sigma \tilde{Y}(x) = \tilde{A}(x)\tilde{Y}(x)$, related by \eqref{Atildeev} defines a transformation that preserves the Galois group.
\end{cor}

These structures can be def\/ined without reference to a connection matrix, it only requires the existence of a linearly independent set of solutions specif\/ied by Theorem \ref{Praagman}. In particular, it specif\/ies that the birational maps of Theorems \ref{hisolattice} and \ref{qlattice} are integrable with respect to the preservation of a Galois group. What may be interesting from an integrable systems perspective is to consider the combinatorial data that specif\/ies the dif\/ference module. Such data would be the analogue of the characteristic constants involved in isomonodromic deformations, and the map from the given dif\/ference module to this data would constitute a discrete analogue of the Riemann--Hilbert map \cite{Boalch2005}.

\section{Discrete Garnier systems}\label{sec:Garnier}

We now turn to the parameterization of our discrete Garnier systems, which has drawn inspiration from a series of results concerning the description of various integrable autonomous mappings and discrete Painlev\'e equations in terms of reductions of partial dif\/ference equations~\cite{Ormerod2014a}. We have denoted the various cases of discrete Garnier systems by a~value~$m$ in a way that the case $m=1$ coincides with a discrete Garnier system that possesses the sixth Painlev\'e equation as a limit. With respect to the Garnier systems increasing~$m$ increases the number of poles of the matrix of the associated linear problem whereas increasing~$m$ by one in what we are calling the discrete Garnier systems increases the number of roots of the determinant of the matrix for the associated linear problem by two.

\subsection[The asymmetric $h$-dif\/ference Garnier system]{The asymmetric $\boldsymbol{h}$-dif\/ference Garnier system}

We start with \eqref{lineardiff} where $A(x)$ is specif\/ied by \eqref{prodform} for $N = 2m+4$ with each factor of \eqref{prodform} is taken to be of the form $L_i(x) = L(x,u_i,a_i)$ where
\begin{gather}\label{difffactor}
L(x,u,a) = \begin{pmatrix} u & 1 \\ x-a+u^2 & u \end{pmatrix}.
\end{gather}
The variable $a$ parameterizes the value of the spectral parameter, $x$, in which $L$ is singular. Some of the useful properties of these matrices are
\begin{subequations}\label{detL}
\begin{gather}
\det L(x,u,a) = a-x, \label{detrel}\\
L(x+\delta,u,a+\delta) = L(x,u,a),\label{hinvar} \\
L(x,u,a)^{-1} = \frac{1}{x-a} L(x,-u,a),\label{hinv}\\
\operatorname{Ker} L(a,u,a) = \left\langle \begin{pmatrix} 1 \\ - u \end{pmatrix} \right\rangle, \label{hker}\\
\operatorname{Im} L(a,u,a) = \left\langle \begin{pmatrix} 1 \\ u \end{pmatrix} \right\rangle, \label{hIm}
\end{gather}
\end{subequations}
hence we think of $u$ as the variable parameterizing the image and kernel vectors. The resulting matrix, $A(x)$, takes the general form
\begin{gather}\label{mtupleh}
A(x) = A_0 + A_1 x + \dots + A_{m+1} x^{m+1} + A_{m+2}x^{m+2}.
\end{gather}

\begin{prop}
Let $A(x)$ be the matrix specified by \eqref{prodform} where each factor is given by \eqref{difffactor} subject to the constraints
\begin{gather}
\label{constrainth}\sum_{i = 1}^{2m+4} u_i = 0,\\
\label{constrainth2} \sum_{k \textrm{ even}} \big(u_{k}^2-a_k\big) \neq \sum_{k \textrm{ odd}} (u_{k}^2-a_k),
\end{gather}
then $A(x)$ defines an $(m+2)$-tuple $(A_0,\ldots, A_{m+1})$ via \eqref{mtupleh} where $A_{m+2} = I$ with
\begin{gather*}
(A_0, \ldots, A_{m+1}) \in \mathcal{M}_h(a_1,\ldots, a_{2m+4};d_1, d_2;1,1),
\end{gather*}
where the values of $d_1$ and $d_2$ are
\begin{subequations}\label{dvals}
\begin{gather}
\label{d1}d_1 = \sum_{i=1}^{N}\sum_{j=1}^{i-1} u_i u_j + \sum_{k \textrm{ even}} \big(u_{k}^2-a_k\big), \\
\label{d2}d_2 = \sum_{i=1}^{N}\sum_{j=1}^{i-1} u_i u_j + \sum_{k \textrm{ odd}} \big(u_{k}^2-a_k\big).
\end{gather}
\end{subequations}
\end{prop}

\begin{proof}
The determinant of $A(x)$ is given by
\begin{gather}\label{deth}
\det A(x) = (x-a_1) \cdots (x-a_{N}),
\end{gather}
which follows from \eqref{detrel}. The f\/irst two terms in the asymptotic expansion around $x= \infty$ are
\begin{gather*}
A(x) = x^{m+2} \begin{pmatrix}
1 & 0 \\
r_{1,2} & 1
\end{pmatrix}
+ x^{m+1}
\begin{pmatrix} d_1 & r_{1,2} \\
r_{2,1} & d_2
\end{pmatrix}
+ O\big(x^{m-1}\big),
\end{gather*}
where $d_1$ and $d_2$ are given by (\ref{dvals}) and $r_{1,2}$ is given by the left-hand side of~\eqref{constrainth}, hence, $A_{m+2} = I$ when assuming the constraints. The value of $r_{2,1}$ is
\begin{gather}\label{r21h}
r_{2,1} = \sum _{i=1}^{m+2} \left(\big(a_{2 i-1}+u_{2 i-1}^2\big) \sum _{k=2 i}^{N} u_k+\big(a_{2 i}+u_{2 i}^2\big) \sum _{k=1}^{2 i-1} u_k\right) + \sum_{1\leq k < j < i \leq N} u_i u_j u_k,
\end{gather}
which \looseness=-1 may be used in a constant lower triangular gauge transformation that diagona\-lizes~$A_n\!$.~This naturally preserves $d_1$ and $d_2$, hence, def\/ines an element of $\mathcal{M}_h(a_1,\ldots, a_{2m+4};d_1, d_2;1,1)$.
\end{proof}

While it is a consequence of \eqref{deth}, \eqref{constrainth}, \eqref{d1} and \eqref{d2}, it should be noted that $d_1$ and $d_2$ satisfy
\begin{gather}\label{drel}
d_1 + d_2 + \sum_{i=1}^{N} a_i = 0,
\end{gather}
which is a constraint that is necessarily satisf\/ied by any element of $\mathcal{M}_h(a_1,\ldots, a_{2m+4};d_1, d_2;1,1)$.

Suppose we are given $A_{m+2} = I$, and an $(m+2)$-tuple
\begin{gather*}
(A_0,\ldots, A_{m+1}) \in \mathcal{M}_h(a_1,\ldots, a_N;d_1,d_2;1,1),
\end{gather*}
where $A_{m+1}$ has been diagonalized, we wish to know whether there is a corresponding matrix of the form~\eqref{prodform}. We claim that the subvariety of $(m+2)$-tuples arising from \eqref{prodform} is of the same dimension. If we f\/ix $A_{m+2} = I$ and $A_{m+1} = \operatorname{diag}(d_1,d_2)$ then each of the $4(m+1)$ entries of the $A_i$'s, for $i = 0, \ldots, m$, are considered free. We have $2m+3$ coef\/f\/icients of the determinant not automatically satisf\/ied. Conjugating by diagonal matrices may also be used to f\/ix one additional of\/f-diagonal entry, which also removes any gauge freedom, making a algebraic variety of dimension $2m$ (or $2m+1$ with a gauge freedom).

Similarly, a product of the form \eqref{prodform} is specif\/ied by $2m+4$ values, $u_i$ for $i = 1,\ldots, 2m+4$ subject to two constraints, namely \eqref{constrainth} and and \eqref{constrainth}, one gauge freedom and two constants related by~\eqref{drel}, giving a total of $2m+2$ free variables. Fixing $r_{2,1}$ removes another variable, as does conjugating by diagonal matrices, which gives an algebraic variety dimension $2m$ (or $2m+1$ with a gauge freedom), as above.

 We may also describe maps between $\mathcal{M}_h(a_1,\ldots, a_N;d_1,d_2;1,1)$ and matrices given by \eqref{prodform}. To obtain an element of $\mathcal{M}_h(a_1,\ldots, a_N;d_1,d_2;1,1)$, we expand the product and diagonalize. To obtain \eqref{prodform} we obtain left (or right) factors of $A(x)$ by observing the corresponding image (kernel) vectors at the points $x = a_1$ (or $x=a_n$).

The property we will use to parameterize the system of discrete isomonodromic deformations is given by the following observation.

\begin{lem}\label{com:diff}
The matrices of the form of \eqref{difffactor} satisfy the commutation relation
\begin{gather*}
L(x,u_i,a_i)L(x,u_j,a_j) = L(x,\tilde{u}_j,a_j)L(x,\tilde{u}_i,a_i),
\end{gather*}
where the map $(u_i,u_j ) \to \left(\tilde{u}_i,\tilde{u}_j\right)$ is given by
\begin{gather}\label{FV}
\tilde{u}_i = u_j + \frac{a_i - a_j}{u_i+u_j}, \qquad \tilde{u}_j = u_i - \frac{a_i - a_j}{u_i+u_j}.
\end{gather}
\end{lem}

This is a well known relation for these matrices \cite{Adler1993, KajiwaraSym, Suris2003}. This map is related to the discrete potential Korteweg--de Vries equation \cite{Papageorgiou2006}. If we let $R_{i,j}$ be the map
\begin{gather}\label{RijYB}
R_{i,j} \colon \ ( u_1, \ldots, u_i, \ldots, u_j, \ldots, u_n ) \to (u_1, \ldots, \tilde{u}_i, \ldots, \tilde{u}_j, \ldots, u_n),
\end{gather}
then this map satisf\/ies the relation
\begin{gather}\label{YangBaxter}
R_{23}R_{13}R_{12}(u,v,w) = R_{12}R_{13}R_{23}(u,v,w),
\end{gather}
which is known as the Yang--Baxter property for maps. This map appears as $\mathrm{F}_{\mathrm{V}}$ in the classif\/ication of quadrirational Yang--Baxter maps \cite{Adler2003}. A common pictorial representation of this property appears in Fig.~\ref{fig:YangBaxter}. More generally, it has been remarked upon in \cite{Borodin:connection} that the set of commuting transformations obtained by discrete isomonodromic deformations def\/ine solutions to the set-theoretic Yang--Baxter maps \cite{Veselov2003}.

\begin{figure}[!ht]\centering
\begin{tikzpicture}[scale=2]
\draw[thick] (0,0) -- (1,.5) -- (1,1.5) -- (0,2) -- (-1,1.5) -- (-1,.5)--cycle;
\draw (-1,.5) --(0,1) -- (0,2);
\draw (0,1) -- (1,.5);
\draw[dashed] (1,1) arc (60:120:2cm);
\draw[dashed] (-.5,.25) arc (-61:-7:2cm);
\draw[dashed] (.5,.25) arc (-120:-173:2cm);
\filldraw[fill=white,draw=black] (.5,.25) circle (.08cm);
\filldraw[fill=white,draw=black] (1,1) circle (.08cm);
\filldraw[fill=white,draw=black] (.5,1.75) circle (.08cm);
\filldraw[fill=black,draw=black] (-.5,.25) circle (.08cm);
\filldraw[fill=black,draw=black] (-1,1) circle (.08cm);
\filldraw[fill=black,draw=black] (-.5,1.75) circle (.08cm);
\node at (-.6,0) {$w$};
\node at (-.6,2) {$u$};
\node at (-1.3,1) {$v$};
\node at (.6,2) {$\tilde{w}$};
\node at (.6,0) {$\tilde{u}$};
\node at (1.3,1) {$\tilde{v}$};
\node at (-1.1,1.7) {$R_{12}$};
\node at (0,-.2) {$R_{13}$};
\node at (1.1,1.7) {$R_{23}$};
\begin{scope}[xshift=4cm]
\draw[thick] (0,0) -- (1,.5) -- (1,1.5) -- (0,2) -- (-1,1.5) -- (-1,.5)--cycle;
\draw (0,0) --(0,1) -- (1,1.5);
\draw (0,1) -- (-1,1.5);
\draw[dashed] (1,1) arc (-60:-120:2cm);
\draw[dashed] (-.5,.25) arc (-187:-237:2cm);
\draw[dashed] (.5,.25) arc (7:60:2cm);
\filldraw[fill=white,draw=black] (.5,.25) circle (.08cm);
\filldraw[fill=white,draw=black] (1,1) circle (.08cm);
\filldraw[fill=white,draw=black] (.5,1.75) circle (.08cm);
\filldraw[fill=black,draw=black] (-.5,.25) circle (.08cm);
\filldraw[fill=black,draw=black] (-1,1) circle (.08cm);
\filldraw[fill=black,draw=black] (-.5,1.75) circle (.08cm);
\node at (-.6,0) {$w$};
\node at (-.6,2) {$u$};
\node at (-1.3,1) {$v$};
\node at (.6,2) {$\tilde{w}$};
\node at (.6,0) {$\tilde{u}$};
\node at (1.3,1) {$\tilde{v}$};
\node at (-1.1,.3) {$R_{23}$};
\node at (0,2.2) {$R_{13}$};
\node at (1.1,.3) {$R_{12}$};
\end{scope}
\end{tikzpicture}
\caption{The quadralaterals labeled by $R_{i,j}$ denote the application of~\eqref{RijYB} to the triple $(u,v,w)$. The equivalence of the left and right pictures is the Yang--Baxter property.} \label{fig:YangBaxter}
\end{figure}
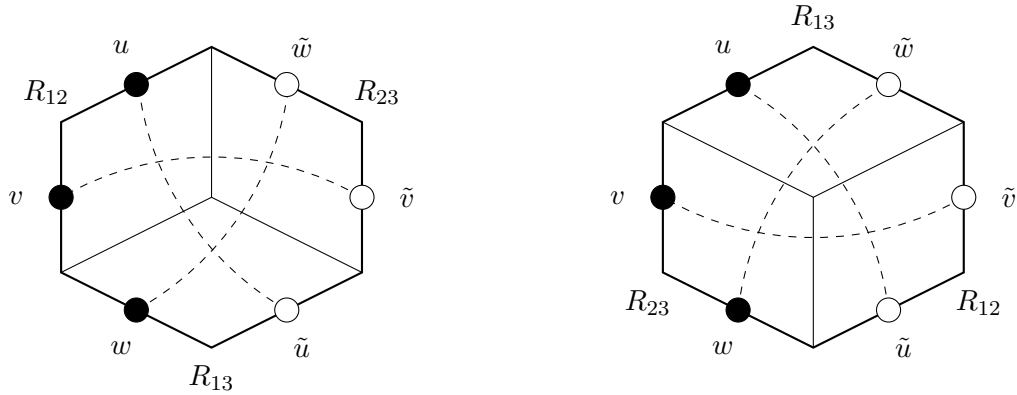

We may use Lemma \ref{com:diff} to def\/ine an action of $S_N$ on $A(x)$. Given a permutation, $\sigma \in S_N$, we denote the corresponding rational transformation of the $u_i$ and $a_i$ by $s_{\sigma} u_i$ and $s_{\sigma} a_i$ respectively. The group $S_N$ is generated by 2-cycles of the form $(i,i+1)$, whose action we denote by $s_i = s_{(i,i+1)}$ for $i = 1, \ldots, N-1$. Using Lemma~\ref{com:diff} these are given by
\begin{subequations}\label{permute}
\begin{gather}
s_i \colon \ u_i \to u_{i+1} + \frac{a_i - a_{i+1}}{u_i+u_{i+1}}, \qquad s_i \colon \ a_i = a_{i+1},\\
s_i \colon \ u_{i+1} \to u_i - \frac{a_i - a_{i+1}}{u_i+u_{i+1}}, \qquad s_i \colon \ a_{i+1} = a_i.
\end{gather}
\end{subequations}
By construction for any $\sigma \in S_N$, the ef\/fect of~$s_{\sigma}$ on~$A(x)$ is trivial. We may use the action of~$S_n$ to determine the image or kernel of $A(x)$ at $x=a_i$ by acting on $A(x)$ by a permutation that sends the factor that is singular at $x=a_i$ to either the f\/irst or last term of~\eqref{prodform} respectively.

We are now in a position to def\/ine an elementary collection of translations, $T_i$, whose ef\/fect on the parameters, $a_i$, is given by
\begin{gather*}
T_i \colon \ a_j \to \begin{cases}
a_i + h & \text{if $i = j$},\\
a_j & \text{if $i\neq j$},
\end{cases}
\end{gather*}
and whose action on the $u_i$ variables is the subject of the following proposition.

\begin{prop}
The matrix $R(x) = L(x-h,u_1,a_1)^{-1}$ in \eqref{Atildeev} defines a birational map between linear algebraic varieties
\begin{gather*}
T_1 \colon \ \mathcal{M}_h(a_1, a_2, a_3, \ldots, a_N; d_1, d_2;1,1) \to \mathcal{M}_h(a_1+h, a_2+h, a_3, \ldots, a_N;d_2-h, d_1;1,1).
\end{gather*}
The effect of $T_1$ on the $u_i$ variables is given by
\begin{gather*}
T_1 u_k = \begin{cases}
u_{1,k} + \dfrac{a_1+h-a_k}{u_{1,k}+u_k} & \text{for $k = 2, \ldots, N$},\\
u_{1,1} & \text{for $k = 1$},
\end{cases}
\end{gather*}
where
\begin{gather*}
u_{1,k-1} = \begin{cases}
u_{k-1} - \dfrac{a_1+h-a_k}{u_{1,k}+u_k} & \text{for $k = 2, \ldots, N-1$}, \\
u_1 & \text{for $k= N+1$}.
\end{cases}
\end{gather*}
\end{prop}

\begin{proof}
To ascertain the how this transformation acts on $A(x)$, we observe that a rearrangement of \eqref{Atildeev} is that
\begin{gather}\label{tildeprod}
\tilde{A}(x) = L(x,u_2,a_2) \cdots L(x,u_{2m+4},a_{2m+4}) L( x, u_1, a_1 + h),
\end{gather}
where we have used \eqref{hinvar}. It is convenient to leave it in this form and read of\/f the transformed values of $d_1$ and $d_2$ in the expansion of \eqref{tildeprod} to be given by
\begin{gather*}
T_1 d_1 = \sum_{i=1}^{N}\sum_{j=1}^{i-1} u_i u_j + \sum_{k \textrm{ odd}} \big(u_{k}^2-a_k\big) - h = d_2 - h,\\
T_1 d_2 =\sum_{i=1}^{N}\sum_{j=1}^{i-1} u_i u_j + \sum_{k \textrm{ even}} \big(u_{k}^2-a_k\big) = d_1,
\end{gather*}
which determines that $\tilde{A}(x)$ is an element of $\mathcal{M}_h(a_1+h, a_2+h, a_3, \ldots, a_{2m+4};d_2-h, d_1;1,1)$. We may inductively determine $T_1 u_k$ by observing that the kernel of~$\tilde{A}(a_{2m+2})$, giving us
\begin{gather*}
T_1 \colon \ u_{2m+4} = u_1 +\frac{a_1+h-a_{2m+4}}{u_{1}+u_{2m+4}},
\end{gather*}
by applying $s_{2m+3}$ and \eqref{hker}. Any subsequent kernels may be found inductively by examining the kernel of
\begin{gather*}
 L(a_k,u_2,a_2) \cdots L(x,u_{k},a_{k}) L( a_k, u_{1,k}, a_1 + h),
\end{gather*}
for $k > 1$ and where $u_{1,2m+2} = u_1$.
\end{proof}

Rather than computing compatibility relations explicitly, we have simply exploited the commutation relations between the $L_i$ factors. All the other elementary transformations may be obtained by conjugating by elements of~$S_N$. One of the issues with this type of transformation is that it is singular at $x =\infty$, which manifests itself in the way it has swapped the roles of~$d_1$ and~$d_2$. If we conjugate by the matrix with $1$'s on the of\/f diagonal, we can also swap the roles of~$d_1$ and~$d_2$, however, the ef\/fect this has on the~$u_i$ variables is not so clear, as it requires a~nontrivial refactorization into a product of the appropriate form. We may now present the ge\-ne\-ra\-tors for the discrete Garnier systems, which are compositions of the form $T_{i,j} = T_i \circ T_j$ where $i \neq j$.

\begin{prop}\label{tranha1a2}
The matrix $R(x) = L(x-h,u_2,a_2)^{-1}L(x-h,u_1,a_1)^{-1}$ in \eqref{Atildeev} defines a~birational map between linear algebraic varieties
\begin{gather*}
T_{1,2} \colon \ \mathcal{M}_h (a_1, a_2, a_3, \ldots, a_N;d_1,d_2;1,1) \\
\hphantom{T_{1,2} \colon}{} \ {}\to \mathcal{M}(a_1+h, a_2+h, a_3, \ldots, a_N;d_1-h, d_2-h;1,1).
\end{gather*}
The effect of $T_{1,2}$ on the $u_i$ variables is given by
\begin{gather}\label{conprevuh}
T_{1,2} \colon \ u_i = \begin{cases}
u_{1,2} & \text{for $i = 1$},\\
u_{2,2} & \text{for $i = 2$},\\
u_k + (u_{k,1} -u_{k-1,1})+ (u_{2,k}- u_{2,k-1}) & \text{for $k = 3,\ldots, N$},
\end{cases}
\end{gather}
where
\begin{gather*}
u_{1,k-1}= u_k + \frac{a_1+h-a_k}{u_{1,k}+u_k},\\
u_{2,k-1}= u_k - u_{1,k-1} + u_{1,k}+ \frac{a_2+h-a_k}{u_k + u_{2,k} + u_{1,k} - u_{1,k-1}},
\end{gather*}
for $k = 2, \ldots, N$, $u_{1,N} = u_1$ and $u_{2,N} = u_2$.
\end{prop}

\begin{proof}
As was the case in the previous proposition, using the identif\/ication of $\tilde{A}(x)$ with the action of $T_{1,2}$ we f\/ind that
\begin{gather*}
\tilde{A}(x) = L(x,u_3,a_3) \cdots L(x,u_{N}, a_{N}) L(x,u_1,a_1+ h)L(x,u_2,a_2+h),
\end{gather*}
whose expansion around $x=\infty$ reveals that $T_{1,2} d_i = d_i - h$ for $i = 1,2$, showing that the image of~$T_{1,2}$ is indeed in $\mathcal{M}_h(a_1+h, a_2+h, a_3, \ldots, a_N;d_1-h, d_2-h)$. To compute the action on the~$u_i$ variables, we inductively compute the kernel of
\begin{gather*}
 L(a_k,u_3,a_3) \cdots L(a_k,u_{k},a_{k}) L( a_k, u_{1,k}, a_1 + h)L( a_k, u_{2,k}, a_2 + h),
\end{gather*}
using the action of $S_N$, which gives \eqref{conprevuh} with an initial step where $u_{1,N} = u_1$ and $u_{2,N} = u_2$ as above.
\end{proof}

We may construct a generic element $T_{i,j}$, whose action on the space of parameters is
\begin{gather}
T_{i,j} \colon \ \mathcal{M}_h (a_1, \ldots, a_i, \ldots, a_j, \ldots, a_N;d_1,d_2;1,1) \nonumber\\
\hphantom{T_{i,j} \colon}{} \ {} \to \mathcal{M}_h(a_1, \ldots, a_i+h, \ldots, a_j+h, \ldots, a_N;d_1-h,d_2-h;1,1),\label{hGarnierAction}
\end{gather}
by conjugating by the element $\sigma_{(1i)(2j)}$. That is to say
\begin{gather*}
T_{i,j} = \sigma_{(1i)(2j)} \circ T_{1,2} \circ \sigma_{(1i)(2j)}.
\end{gather*}
The system of transformations of the form $T_{i,j}$ constitutes what we call the $h$-Garnier system. The simplest case, when $m=1$, is shown to coincide with the dif\/ference analogue of the sixth Painlev\'e equation in Section~\ref{dP6sec}.

As a consequence of Theorem \ref{hisolattice}, we have the following.

\begin{cor}
The set of transformations of the form $T_{i,j}$ satisfy the following
\begin{enumerate}\itemsep=0pt
\item[$1.$] The action is symmetric in $i$ and $j$, i.e.,
\begin{gather*}
T_{i,j} = T_{j,i}.
\end{gather*}

\item[$2.$] These actions commute, i.e.,
\begin{gather*}
T_{i_1, j_1} \circ T_{i_2, j_2} = T_{i_2, j_2} \circ T_{i_1, j_1}.
\end{gather*}
\end{enumerate}
\end{cor}

\subsection[The symmetric $h$-dif\/ference Garnier system]{The symmetric $\boldsymbol{h}$-dif\/ference Garnier system}

Let us consider dif\/ference equations whose solutions satisfy $Y(x) = Y(-x)$. The consistency of~\eqref{lineartau} requires that
\begin{gather*}
A(x)A(-h-x) =I.
\end{gather*}
Under these conditions we express $A(x)$ by \eqref{symmdiff}, in which
\begin{gather*}
B(x) = L_1(x) \cdots L_{N'}(x),
\end{gather*}
where $L_i(x) = L(x,u_i, a_i)$ given by \eqref{difffactor} and $N' = 2m+4$ as before. In this case, by using~\eqref{hinv} we may write
\begin{gather*}
A(x) = B(-h-x)^{-1}B(x)= \!\left[\prod_{k=1}^{N'} \frac{1}{x+a_k+h}\right]\! L(-x,-u_{N'}, a_{N'}+h){\cdots} L(-x,-u_1, a_1+h)\\
\hphantom{A(x) = B(-h-x)^{-1}B(x)=}{} \times L(x,u_1, a_1+h) \cdots L(x,-u_{N'}, a_{N'}+h).
\end{gather*}
This could be transformed via $\Gamma$ functions to a matrix of the form of~\eqref{prodform} for $N = 2N'$ and where the last $N$ factors take a slightly dif\/ferent form. If we were to apply Theorem~\ref{hisolattice}, it is not clear at this point that the solutions would preserve the symmetry.

Due to \eqref{hinv} and the invariance of~\eqref{permute} under changes to the spectral variable, it is easy to see that one may simultaneously act on~$B(x)$ and $B(-h-x)^{-1}$ by $S_n$ in the same way as~\eqref{permute}. As discussed previously, we expect to f\/ind transformations induced by multiplication on the left and the right. The left multiplication is expected to def\/ine a trivial transformation of~$A(x)$, but what is not expected is that the transformation is similar to the transformation specif\/ied by Lemma~\ref{com:diff}.

\begin{lem}
The rational matrix
\begin{gather*}
R_l(x) = \begin{pmatrix}(x-a_1)(x+a_1) & 0 \\ 0& (x-a_2)(x+a_2) \end{pmatrix}
 + (u_1+u_2)(a_1+ a_2) \begin{pmatrix} u_1 & -1 \\ u_1 E_{1,2} u_1 & -u_1 \end{pmatrix}
\end{gather*}
defines a birational transformation
\begin{gather*}
E_{1,2} \colon \ \mathcal{M}_h (a_1, \ldots, a_{N'}; d_1,d_2;1,1) \\
\hphantom{E_{1,2} \colon}{} \ {} \to \mathcal{M}_h(-a_1,-a_2, a_3, \ldots, a_{N'}; d_1+a_1+a_2,d_2+a_1+a_2;1,1),
\end{gather*}
via \eqref{leftRB} with $\lambda = (x-a_1)^{-1}(x-a_2)^{-1}$. The effect on the $u_i$ variables is given by
\begin{gather*}
E_{1,2} u_1 = u_1 - \frac{a_1-a_2}{u_1+ u_2},\qquad E_{1,2} u_2 = u_2 + \frac{a_1-a_2}{u_1+ u_2}.
\end{gather*}
\end{lem}

This is an elementary calculation that is easily verif\/ied. It is also seen that $R_l(x) = R_l(-x)$, as required, and that
\begin{gather*}
\det R_l(x) = (x-a_1)(x+a_1)(x-a_2)(x+a_2).
\end{gather*}
This def\/ines an involution on the parameter space.

\begin{lem}
The rational matrix
\begin{gather*}
R_r(x) = \begin{pmatrix}(x-a_{N'-1})(x+h+a_{N'-1}) & 0 \\ 0& (x-a_{N'})(x+h+a_{N'}) \end{pmatrix} \\
\hphantom{R_r(x) = }{} + (u_{N'}+u_{N'-1})(a_{N'}+ a_{N'-1}) \begin{pmatrix} u_{N'} & -1 \\ u_{N'} F_{N',N'-1}u_{N'} & -u_{N'} \end{pmatrix}
\end{gather*}
defines a birational transformation
\begin{gather*}
F_{N',N'-1}\colon \ \mathcal{M}_h(a_1,\ldots, a_{N'}; d_1,d_2;1,1) \\
\hphantom{F_{N',N'-1}\colon}{} \ {}\to \mathcal{M}_h(a_1, \ldots, a_{N'-2}, -a_{N'-1}-h,-a_{N'}-h;\\
\hphantom{F_{N',N'-1}\colon \ {}\to \mathcal{M}_h(}{} d_1+a_{N'}+a_{N'-1}+h,d_2+a_{N'}+a_{N'-1}+h;1,1),
\end{gather*}
via \eqref{rightRB} where $\lambda(x) = (x-a_{N'-1})(x-a_{N'})$ whose effect on the $u_i$ variables is given by
\begin{gather*}
F_{N',N'-1}u_{N'} = u_{N'} - \frac{a_{N'}-a_{N'-1}}{u_{N'}+ u_{N'-1}},\\
F_{N',N'-1}u_{N'-1} = u_{N'-1} + \frac{a_{N'}-a_{N'-1}}{u_{N'}+ u_{N'-1}}.
\end{gather*}
\end{lem}

It is easy to see $R_r(x) = R_r(-x-h)$ and
\begin{gather*}
\det R_r(x) = (x-a_1)(x+a_1+h)(x-a_2)(x+a_2+h).
\end{gather*}
These matrices are not of the same form as $L_i(x)$, yet the resulting transformation takes the form specif\/ied in Lemma~\ref{com:diff} where the roles of~$u_i$ and~$u_j$ have been swapped.

It is f\/itting that we def\/ine the generators of the symmetric dif\/ference Garnier system to be the maps~$E_{i,j}$ and $F_{i,j}$, which may be expressed as
\begin{subequations}\label{symgenh}
\begin{gather}
\label{heij}E_{i,j} = s_{(1i)(2j)} \circ E_{1,2} \circ s_{(1i)(2j)},\\
\label{hfij}F_{i,j} = s_{(1i)(2j)} \circ F_{1,2} \circ s_{(1i)(2j)}.
\end{gather}
\end{subequations}
The translations, $T_{i,j}$, are specif\/ied in terms of these generators as
\begin{gather*}
T_{i,j} = F_{i,j} \circ E_{i,j},
\end{gather*}
which form the generators for the system of translations in the $h$-dif\/ference Garnier system. While this bears some similarity with~\eqref{hGarnierAction}, the dif\/ference is that given an $A(x)$ with the appropriate symmetry, the resulting action is inequivalent since the resulting transformations of~\eqref{hGarnierAction} do not necessarily preserve the symmetry, whereas by acting upon~$B(x)$, the resulting matrix~$A(x)$ necessarily possesses the required symmetry. The key dif\/ference is not the moduli space itself, but the actions being considered on them.

\subsection[$q$-dif\/ference Garnier systems]{$\boldsymbol{q}$-dif\/ference Garnier systems}

As we did with the $h$-dif\/ference systems, we start with \eqref{linearsigma} where $\sigma = \sigma_q$ and where~$A(x)$ is specif\/ied by~\eqref{prodform}. Before def\/ining $L_i(x)$, we specify two matrices, $L(x,u,a)$ and a diagonal matrix which we call~$D$ given by
\begin{gather}
\label{qdifffactor} L(x,u,a) = \begin{pmatrix}
1 & u \\
 \dfrac{x}{au} & 1
\end{pmatrix},\qquad D = \begin{pmatrix}
\theta_1 & 0 \\
0& \theta_2
\end{pmatrix},
\end{gather}
which satisfy the commutation relation
\begin{gather}\label{commDL}
L(x,u,a)D = D L\left(x, \frac{\theta_2 u}{\theta_1},a\right) .
\end{gather}
Due to \eqref{commDL}, rather than letting each factor take the form $D L(x,u_i,a_i)$ it is suf\/f\/icient to letting only the f\/irst f\/irst factor take the form $DL(x,u_1,a_1)$ while all other factors are of the form $L(x,u_i,a_i)$, i.e., we let~$A(x)$ take the general form~\eqref{prodform} where
\begin{gather}\label{qfactor}
 L_i(x) = \begin{cases} D L(x,u_1,a_1) & \text{for $i = 1$}, \\
L(x,u_i,a_i) & \text{for $i \neq 1$}.
\end{cases}
\end{gather}
As in the previous section, some of the desirable properties of $L(x,u,a)$ are
\begin{subequations}
\begin{gather}
\label{detLq}\det L(x,u,a) = 1- \frac{x}{a},\\
\label{qident2}L(qx,u,qa) = L(x,u,a),\\
\label{qinv}L(x,u,a)^{-1} = \frac{a}{x-a} L(x,-u,a),\\
\label{kerq} \operatorname{Ker} L(a,u,a) = \left\langle \begin{pmatrix} -u \\ 1 \end{pmatrix} \right\rangle, \\
\label{Imq} \operatorname{Im} L(a,u,a) = \left\langle \begin{pmatrix} u \\ 1 \end{pmatrix} \right\rangle.
\end{gather}
\end{subequations}

By expanding \eqref{prodform} we have that $A(x)$ takes the general form
\begin{gather*}
A_0 + A_1 x + \cdots + A_{m+1} x^{m+1}.
\end{gather*}
As we did with the previous section, we f\/ind that the properties of $A(x)$ are given by the following proposition.

\begin{prop}\label{qalgvariety}
Given $A(x)$ specified by \eqref{prodform} where each factor is given by \eqref{qfactor}, with the constraints $\theta_1 \neq \theta_2$ and
\begin{gather}\label{con2q}
\theta_1\prod_{i \ \textrm{odd}} \big(a_iu_i^2\big) \neq \theta_2\prod_{j \ \textrm{odd}} \big(a_ju_j^2\big),
\end{gather}
defines an element,
\begin{gather*}
(A_0, \ldots, A_{m+1}) \in \mathcal{M}_q(a_1,\ldots, a_N; \kappa_1, \kappa_2; \theta_1, \theta_2),
\end{gather*}
where
\begin{subequations}\label{kappavals}
\begin{gather}
\kappa_1 = \theta_1\prod_{i \ \textrm{odd}} u_i \prod_{j \ \textrm{even}} (a_ju_j)^{-1},\\
\kappa_2 = \theta_2 \prod_{i \ \textrm{even}} u_i \prod_{j \ \textrm{odd}} (a_ju_j)^{-1},
\end{gather}
\end{subequations}
and $\theta_1$ and $\theta_2$ appear as they do in \eqref{qfactor}.
\end{prop}

\begin{proof}
The property \eqref{detLq} is suf\/f\/icient to tell us
\begin{gather}
\label{detq} \det A(x) = \theta_1\theta_2 \left( 1- \frac{x}{a_1} \right) \cdots \left( 1- \frac{x}{a_N} \right),
\end{gather}
where expansions around $x = \infty$ and $x = 0$ are
\begin{gather*}
A(x) = x^{m+1} \begin{pmatrix} \kappa_1 & 0 \\ r_{1,2} & \kappa_2 \end{pmatrix} + O\big(x^m\big),\\
A(x) = \begin{pmatrix} \theta_1 & \theta_1 \left( \sum\limits_{i = 0}^N u_i \right) \\
0 & \theta_2 \end{pmatrix} + O(x),
\end{gather*}
where the values of $\kappa_1$ and $\kappa_2$ are as above and
\begin{gather}\label{r12q}
r_{1,2} = \frac{\theta_2 \kappa_1}{\theta_1 u_1} \sum_{j = 1}^{N-1} \prod_{i=1}^{j} \left(\frac{a_{i+1}}{a_i}\right)^{\frac{1 - (-1)^i}{2}}\left( \frac{u_i}{u_{i+1}}\right)^{(-1)^i}.
\end{gather}
The constraints that both $\theta_1$ and $\theta_2$ and \eqref{con2q} are suf\/f\/icient (but not necessary) to ensure that~$A_0$ and~$A_{m+1}$ are semisimple.
\end{proof}

By a similar counting argument to the $h$-dif\/ference case, we may show that matrices taking the form given by~\eqref{prodform} and $\mathcal{M}_q(a_1,\ldots, a_N; \kappa_1, \kappa_2; \theta_1, \theta_2)$ both describe algebraic varieties of dimension~$2m$ with birational maps between the two. This justif\/ies that we may parameterize our discrete isomonodromic deformations in terms of actions on matrices of the form~\eqref{prodform}.

As we mentioned above, the conditions that $\theta_1= \theta_2$ or equality holds in \eqref{con2q} are not necessary for $A_0$ and $A_{m+1}$ to be semisimple, as requires that the matrix $A_0$ is diagonalizable, which amounts to requiring that $A_0$ is diagonal, which imposes the constraint
\begin{gather*}
\sum_{i=0}^N u_i = 0.
\end{gather*}
On the other hand, if equality holds in \eqref{con2q}, we require that $r_{1,2} = 0$. If $\theta_1 = \theta_2$ then the case of $m = 1$ has too many constraints to be interesting, hence, it is more natural to consider $m=2$ to be the f\/irst interesting case. Similarly, if both $\theta_1 = \theta_2$ and equality holds in \eqref{con2q}, then we have an additional constraint, making $m=3$ the f\/irst interesting case for similar reasons.

\begin{lem}\label{lem:qcom}
Matrices of the form of \eqref{qdifffactor} satisfy the commutation relation
\begin{gather*}
L(x,u_i,a_i)L(x,u_j,a_j) = L(x,\tilde{u}_j,a_j)L(x,\tilde{u}_i,a_i),
\end{gather*}
where the map $\left(u_i,u_j \right) \to \left(\tilde{u}_i,\tilde{u}_j\right)$ is given by
\begin{gather}\label{commD}
\tilde{u}_i = \frac{a_i u_i(u_i + u_j)}{a_i u_i + a_j u_j},\qquad \tilde{u}_j = \frac{a_j u_j(u_i + u_j)}{a_i u_i + a_j u_j} .
\end{gather}
\end{lem}

Once again, this map satisf\/ies the Yang--Baxter property in that if we def\/ine $R_{i,j}$ in the same manner as~\eqref{RijYB} then~\eqref{YangBaxter} holds. In the classif\/ication of quadrirational Yang--Baxter maps~\eqref{commD} appears as~$F_{\rm III}$~\cite{Adler2003}.

In the same manner as the previous section, it is useful to utilize Lemma~\ref{lem:qcom} to def\/ine the action of $S_N$ on $A(x)$. Given a permutation, $\sigma \in S_N$, we denote the action of $\sigma$ on the~$u_i$ and~$a_i$ by~$s_{\sigma} u_i$ and~$s_{\sigma} a_i$. Following the notation from the previous section, we specify the action of the generators is computed using~\eqref{lem:qcom} to be
\begin{alignat*}{3}
& s_i \colon \ u_i \to \frac{a_{i+1}u_{i+1}(u_i + u_{i+1})}{a_i u_i + a_{i+1} u_{i+1}}, \qquad && s_i \colon \ a_i = a_{i+1},&\\
& s_i \colon \ u_{i+1} \to \frac{a_{i+1}u_{i+1}(u_i + u_{i+1})}{a_i u_i + a_{i+1} u_{i+1}}, \qquad && s_i \colon \ a_{i+1} = a_i.&
\end{alignat*}
Once again, the ef\/fect of $s_\sigma$ is trivial on $A(x)$. This action and \eqref{commDL} will be suf\/f\/icient to express the discrete isomonodromic deformations. We wish to specify the transformation whose action on the parameters is
\begin{gather*}
T_i \colon \ a_j \to \begin{cases}
qa_i & \text{if $i = j$},\\
a_j & \text{if $i\neq j$},
\end{cases}
\end{gather*}
and whose action on the $u_i$ is to be specif\/ied. Once again, the matrices that def\/ine the elementary Schlesinger transformations are of the form of~$L(x,u,a)$. The most basic transformation is specif\/ied in terms of the left-most factor.

\begin{prop}
The matrix $R(x) = L(x/q,u_1,a_1)^{-1} D^{-1}$ in \eqref{Atildeev} defines a birational map between algebraic varieties
\begin{gather*}
T_1 \colon \ \mathcal{M}_q(a_1,\ldots, a_N; \kappa_1, \kappa_2 ; \theta_1, \theta_2) \to \mathcal{M}_q(a_1,\ldots, a_N; \kappa_2, \kappa_1/q; \theta_1, \theta_2).
\end{gather*}
The effect of $T_1$ on the $u_i$ variables is given by
\begin{gather}\label{qtrans1}
T_1 u_k = \begin{cases}
\dfrac{qa_1u_{1,k}(u_k \theta_2 + \theta_1 u_{1,k})}{a_k u_k\theta_2 + q a_1 \theta_1 u_{1,k}} & \text{for $k = 2, \ldots, N$},\\
u_{1,1} & \text{for $k = 1$},
\end{cases}
\end{gather}
where
\begin{gather*}
u_{1,k-1} =
\frac{a_ku_k\theta_2(u_k \theta_2 + \theta_1 u_{1,k})}{\theta_1(a_k u_k\theta_2 + q a_1 \theta_1 u_{1,k})}
u_1 \qquad \text{for $k= N+1$}.
\end{gather*}
\end{prop}

\begin{proof}
This proposition follows in a similar manner, in that we identify $\tilde{A}(x)$ with
\begin{gather*}
\tilde{A}(x) = L(x,u_2,a_2) \cdots L(x,u_{N},a_{N})DL(x,u_1,q a_1),
\end{gather*}
using \eqref{comp} and \eqref{qident2}, which allows us to compute the determinant. Secondly, we note that we may use \eqref{commDL} to show
\begin{gather*}
\tilde{A}(x) = DL\left(x,\frac{\theta_2u_2}{\theta_1},a_2\right) \cdots L\left(x,\frac{\theta_2u_N}{\theta_1},a_{N}\right) L(x,u_1,q a_1),
\end{gather*}
which is of the form in Proposition \ref{qalgvariety}, which shows
\begin{gather*}
T_1 \kappa_1 = \theta_1 (qa_1 u_1)^{-1} \prod_{i\ \mathrm{ even}} \left( \frac{\theta_2u_i}{\theta_1} \right)\prod_{i\ \mathrm{ odd}, \; i \neq 1} \left(a \frac{\theta_2u_i}{\theta_1} \right)^{-1}
 = \frac{\theta_2}{q} \prod_{i\ \mathrm{ even}} u_i\prod_{i\ \mathrm{ odd}} \left(a u_i \right)^{-1} = \frac{\kappa_2}{q},\\
T_1 \kappa_2 = \kappa_1.
\end{gather*}
This shows that the image of $T_1$ is indeed in $\mathcal{M}_q(a_1,\ldots, a_N; \kappa_2, \kappa_1/q; \theta_1, \theta_2)$. To determine the ef\/fect on the $u_i$ variables, the only dif\/ference in the inductive step is that we need to use~\eqref{commD} in combination with Lemma~\ref{lem:qcom}. We compute the kernel of
\begin{gather*}
 L(x,u_2,a_2) \cdots L(x,u_{k},a_{k}) D L(x,u_{1,k},q a_1),
\end{gather*}
using the action of $S_N$ and \eqref{commD}, which inductively provides us with~\eqref{qtrans1} with $u_{1,N} = u_1$.
\end{proof}

The $T_i$ transformations may be obtain through conjugation by the action of $S_N$. This transformation is also not a transformation of the form specif\/ied in Theorem~\ref{qlattice}, since it swaps the role of $\kappa_1$ and $\kappa_2$. To def\/ine the generators of the $q$-Garnier system, we compute $T_1 \circ T_2$, which may be used to compute~$T_{i,j}$.

\begin{prop}
The matrix $R(x) = L(x/q,u_2,a_2)^{-1} L(x/q,u_1,a_1)^{-1}D^{-1}$ in \eqref{comp} defines a~birational map between algebraic varieties
\begin{gather*}
T_{1,2} \colon \ \mathcal{M}_q(a_1, \ldots, a_N; \kappa_1, \kappa_2; \theta_1, \theta_2) \to \mathcal{M}_q(qa_1, qa_2, \ldots, a_N; \kappa_1/q, \kappa_2/q; \theta_1, \theta_2).
\end{gather*}
The effect of $T_{1,2}$ on the $u_k$ coordinates is given by
\begin{gather}
\label{qT12} T_{1,2} u_k = \begin{cases}
u_{1,2} & \text{for $k = 1$},\\
u_{2,2} & \text{for $k = 2$},\\
\dfrac{a_2u_ku_{2,k-1}u_{2,k}\theta_2}{a_1u_{1,k}u_{1,k-1}\theta_1} & \text{for $k = 3, \ldots, N$},
\end{cases}
\end{gather}
where
\begin{gather*}
u_{1,k-1} = \frac{\theta _2 a_k u_k \left(\theta _1 u_{1,k}+\theta _2 u_k\right)}{\theta _1 \left(qa_1 \theta _1 u_{1,k}+\theta _2 a_k u_k\right)},\\
u_{2,k-1} =\frac{a_1 u_{1,k-1} u_{1,k} \left(\theta _1 \left(u_{1,k}-u_{1,k-1}+u_{2,k}\right)+\theta _2
 u_k\right)}{a_1 \theta _1 u_{1,k-1} u_{1,k}+a_2 \theta _2 u_k u_{2,k}}.
\end{gather*}
\end{prop}

The induction follows in the same way as it did for Proposition~\ref{tranha1a2}. In a similar way, we specify that that what we call the $q$-Garnier system is the system of transformations whose action on the parameters is specif\/ied by
\begin{gather}
T_{i,j} \colon \ \mathcal{M}_q (a_1, \ldots, a_i, \ldots, a_j, \ldots, a_N;\kappa_1,\kappa_2;\theta_1,\theta_2)\nonumber\\
\hphantom{T_{i,j} \colon}{} \ {}\to \mathcal{M}_q(a_1, \ldots, qa_i, \ldots, qa_j, \ldots, a_N;\kappa_1/q,\kappa_2/q;\theta_1,\theta_2),\label{qGarnierAction}
\end{gather}
which is given by
\begin{gather*}
T_{i,j} = \sigma_{(1i)(2j)} \circ T_{1,2} \sigma_{(1i)(2j)}.
\end{gather*}
The simplest case, when $m=1$, is the $q$-analogue of the sixth Painlev\'e equation.

The following arises as a consequence of Theorem~\ref{qlattice}.

\begin{cor}
The set of transformations of the form $T_{i,j}$ satisfy the following
\begin{enumerate}\itemsep=0pt
\item[$1.$] The action is symmetric in $i$ and $j$, i.e.,
\begin{gather*}
T_{i,j} = T_{j,i}.
\end{gather*}

\item[$2.$] These actions commute, i.e.,
\begin{gather*}
T_{i_1, j_1} \circ T_{i_2, j_2} = T_{i_2, j_2} \circ T_{i_1, j_1}.
\end{gather*}
\end{enumerate}
\end{cor}

\subsection[Symmetric $q$-Garnier system]{Symmetric $\boldsymbol{q}$-Garnier system}

We now impose the symmetry constraint that the solutions satisfy $Y(x) = Y(1/x)$. The consistency of~\eqref{lineartau} requires that
\begin{gather*}
A(x)A(1/(qx))=I.
\end{gather*}
We assume that $A(x)$ takes the form of \eqref{symmqdiff} where $B(x)$ is given by the product of $L$-matrices as
\begin{gather*}
B(x) = L_1(x) \cdots L_{N'}(x),
\end{gather*}
where $L_i$ is given by \eqref{qfactor} where the diagonal entry cancels, hence, without loss of generality, we may choose $D = I$. Using~\eqref{qinv}, we may write $A(x)$ as
\begin{gather*}
A(x) = \left[\prod_{i=1}^{N'} \frac{1}{1-x/a_i}\right] L\left(\frac{1}{x},-u_{N'},qa_{N'}\right) \cdots L\left(\frac{1}{x},-u_{1},qa_{1}\right) \\
\hphantom{A(x) =}{} \times
 L(x,u_{1},a_{1}) \cdots L(x,u_{N'},a_{N'}) ,
\end{gather*}
which def\/ines an matrix in terms of a product of $L$-matrices.

\begin{prop}
The rational matrix
\begin{gather*}
R_l(x) = \begin{pmatrix} x+\dfrac{1}{x} - \dfrac{1}{a_1} - \dfrac{1}{a_2}\!\! & 0 \\ 0 & x+\dfrac{1}{x} -a_1 - a_2 \end{pmatrix}
 + \frac{(1-a_1a_2)}{a_2u_2}\! \begin{pmatrix} -u_1 & u_1(u_1 + u_2) \vspace{1mm}\\ -1- \dfrac{a_2u_2}{a_1 u_1}\!\! & u_1 \end{pmatrix},
\end{gather*}
defines a birational transformation
\begin{gather*}
E_{1,2}\colon \ \mathcal{M}_q(a_1,a_2, a_3, \ldots, a_{N'}; \kappa_1, \kappa_2; 1, 1) \\
\hphantom{E_{1,2}\colon}{} \ {} \to \mathcal{M}_q\left(\frac{1}{a_1},\frac{1}{a_2}, a_3, \ldots, a_{N'}; a_1a_2\kappa_1, a_1a_2 \kappa_2; 1, 1\right) ,
\end{gather*}
via \eqref{leftRB} where $\lambda(x) = (x-a_1)^{-1}(x-a_2)^{-1}$, whose effect on the $u_i$ variables is given by
\begin{gather*}
E_{1,2} u_1 = \frac{a_1 u_1 (u_1 + u_2)}{a_1 u_1 + a_2 u_2},\qquad
E_{1,2} u_2 = \frac{a_2 u_2 (u_1 + u_2)}{a_1 u_1 + a_2 u_2}.
\end{gather*}
\end{prop}

This is easy to verify directly. Furthermore, we have that $R_l(x)= R_l(1/x)$ and that
\begin{gather*}
\det R_l(x) = (x-a_1)(x-a_2)\big(x-a_1^{-1}\big)\big(x-a_2^{-1}\big).
\end{gather*}
The transformation induced by right multiplication is given by the following proposition.

\begin{prop}The rational matrix
\begin{gather*}
R_r(x) = \begin{pmatrix} x+\dfrac{1}{qx} - \dfrac{1}{qa_{N'-1}} - \dfrac{1}{qa_{N'}} & 0 \\ 0 & x+\dfrac{1}{x} -a_{N'-1} - a_{N'} \end{pmatrix} \\
\hphantom{R_r(x) =}{} + \frac{(1-qa_{N'-1}a_{N'})}{qa_{N'-1}u_{N'-1}} \begin{pmatrix} -u_{N'} & -u_{N'}(u_{N'-1} + u_{N'}) \vspace{1mm}\\ 1 + \dfrac{a_{N'-1}u_{N'-1}}{a_{N'} u_{N'}} & u_{N'} \end{pmatrix}
\end{gather*}
defines a birational transformation
\begin{gather*}
F_{N',N'-1}\colon \ \mathcal{M}_q(a_1, \ldots,a_{N'-2},a_{n-2}, a_n; \kappa_1, \kappa_2; 1, 1) \\
\hphantom{F_{N',N'-1}\colon}{} \ {} \to \mathcal{M}_q\left(a_1, \ldots, a_{n-2}, 1/(qa_{n-1}), 1/q a_n; \kappa_1, \kappa_2; 1, 1\right),
\end{gather*}
via \eqref{rightRB} where $\lambda = x(1-x/(qa_{N'-1}))^{-1}(1-x/(qa_{N'}))^{-1}$, whose effect on the $u_i$ variables is given by
\begin{gather*}
F_{N',N'-1} u_{N'-1} = \frac{a_{N'-1} u_{N'-1} (u_{N'} + u_{N'-1})}{a_{N'-1} u_{N'-1} + a_{N'} u_{N'}},\\
F_{N',N'-1} u_{N'} = \frac{a_{N'} u_{N'} (u_{N'} + u_{N'-1})}{a_{N'-1} u_{N'-1} + a_{N'} u_{N'}}.
\end{gather*}
\end{prop}

This is also easy to verify, as is the property that $R_r(x)= R_r(1/qx)$ and
\begin{gather*}
\det R_r(x) = (x-a_1)(x-a_2)\big(x-(qa_1)^{-1}\big)\big(x-(qa_2)^{-1}\big).
\end{gather*}
In the same way as the $h$-dif\/ference case, we def\/ine the $q$-dif\/ference Garnier system to be generated by maps $E_{i,j}$ and $F_{i,j}$, which may be expressed as
\begin{subequations}\label{symgenh2}
\begin{gather}
\label{qeij}E_{i,j} = s_{(1i)(2j)} \circ E_{1,2} \circ s_{(1i)(2j)},\\
\label{qfij}F_{i,j} = s_{(1i)(2j)} \circ F_{1,2} \circ s_{(1i)(2j)}.
\end{gather}
\end{subequations}
The translations, $T_{i,j}$, also specif\/ied by
\begin{gather*}
T_{i,j} = F_{i,j} \circ E_{i,j},
\end{gather*}
generate the translational portion of the symmetric $q$-Garnier system.

\section{Reparameterization}\label{sec:reparam}

The aim of this section is to express the above systems in terms of variables that have been chosen to make a correspondence between our $q$-Garnier systems and the $q$-Garnier system specif\/ied in the work of Sakai~\cite{Sakai:Garnier}. This choice makes sense in both the $h$-dif\/ference and $q$-dif\/ference setting.

\subsection[$h$-dif\/ference Garnier systems]{$\boldsymbol{h}$-dif\/ference Garnier systems}

Let us consider the $h$-dif\/ference Garnier system def\/ined by~\eqref{linearsigma} where $\sigma = \sigma_h$, where $A(x)$ is specif\/ied by~\eqref{prodform} for $N = 2m+2$ and $L_i$ is given by~\eqref{difffactor} subject to the constraint \eqref{constrainth}. As~\eqref{constrainth} implies that the leading coef\/f\/icient of $A(x)$ is proportional to the identity matrix, provided $d_1 \neq d_2$, which is specif\/ied~\eqref{d1} and~\eqref{d2}, we may gauge by a constant lower triangular matrix so that the next leading coef\/f\/icient is diagonal. Under these conditions, we specify a new set of variables, $y_i$, $z_i$ and~$w_i$, related to~$A(x)$ by
\begin{gather*}
A(a_i) = y_i \begin{pmatrix} 1 \vspace{1mm}\\ \dfrac{z_i}{w} \end{pmatrix} \begin{pmatrix} w_i & w \end{pmatrix} = \begin{pmatrix} w_i y_i & w y_i \vspace{1mm}\\
 \dfrac{w_i y_i z_i}{w} & y_i z_i \end{pmatrix},
\end{gather*}
for $i=1, \ldots, N$. This choice is inspired by many works on the matter, such as the work on the $q$-Garnier systems \cite{Sakai:Garnier}, and various works on the Lagrangian approaches to dif\/ference equa\-tions~\cite{Dzhamay2007, Dzhamay2013}. This def\/ines $3N$ parameters, many of which are redundant. After diagonalizing, with $A(x) = (a_{i,j}(x))$, we have that each $a_{i,j}(x)$ is a polynomial with the following properties specifying their coef\/f\/icients:
\begin{itemize}\itemsep=0pt
\item $a_{i,i}(x) = x^{m+2} + d_i x^{m+1} + O(x^m)$ with $a_{1,1}(a_k) = y_kw_k$ and $a_{2,2}(a_k) = y_kz_k$ for $k=1,\ldots, m+1$,
\item $a_{1,2}(a_k) = wy_k$ and $a_{2,1} = w_ky_kz_k/w$ for $k=1, \ldots, m+1$.
\end{itemize}
We use a form of Lagrangian interpolation in the following way: If we let
\begin{gather*}
D(x) = \prod_{i=1}^{m+1} (x-a_i),
\end{gather*}
then the polynomial $D(x)/(x-a_k)$ satisf\/ies
\begin{gather}\label{deltapoly}
\frac{D(x)}{(x-a_k) D'(a_k)} =
\begin{cases}
0 & \text{if $x = a_j$ for $j \neq k$},\\
1 & \text{if $x = a_k$}.
\end{cases}
\end{gather}
This allows us to express the entries of $A(x)$ as
\begin{subequations}\label{largrangeinth}
\begin{gather}
a_{1,1}(x) = D(x)\left( x + d_1 + \sum_{i=1}^{m+1} a_i+ \sum_{i=1}^{m+1} \frac{y_iw_i}{(x-a_i)D'(a_i)}\right),\\
a_{2,2}(x) = D(x)\left( x + d_2 + \sum_{i=1}^{m+1} a_i+ \sum_{i=1}^{m+1} \frac{y_iz_i}{(x-a_i)D'(a_i)}\right),\\
a_{1,2}(x) = w D(x) \sum_{i=1}^{m+1} \frac{y_i}{D'(a_i)(x-a_i)},\\
a_{2,1}(x) = \frac{D(x)}{w} \sum_{i=1}^{m+1} \frac{y_iz_iw_i}{D'(a_i)(x-a_i)}.
\end{gather}
\end{subequations}
It is convenient to write the expressions for each of the $y_k$, $z_k$ and $w_k$ as
\begin{gather*}
y_k := \frac{a_{1,2}(a_k)}{w}, \qquad w_k = \frac{a_{1,1}(a_k)}{y_i}, \qquad z_k = \frac{a_{2,2}(a_k)}{y_i},
\end{gather*}
which is trivially true for $k = 1,\ldots, m+1$ and def\/ines an expression for $y_k$, $z_k$ and $w_k$ in terms of the f\/irst $m+1$ values for $k = m+2, \ldots, N$. This also produces expressions for each of the new variables in terms of the $u_i$. Naturally, this does not take into account any constants with respect to $T_{i,j}$. After diagonalizing the leading coef\/f\/icient in the polynomial expansion in $x$, it is easy to see that the matrix inducing $T_{i,j}$ takes the form
\begin{gather}\label{Rhform}
R(x) = \frac{xI + R_0}{(x-a_i-h)(x-a_j-h)},
\end{gather}
hence, we may calculate the equivalent of \eqref{hGarnierAction} on the $y_k$, $z_k$ and $w_k$ variables.

\begin{thm}
The system \eqref{hGarnierAction} is equivalent to the following action on the variables $y_k$, $z_k$ and $w_k$
\begin{subequations}\label{hGarnierBirrational}
\begin{gather}
\label{tijrhy} \frac{(T_{i,j} wy_k)}{w}= y_k \frac{z_i (a_k-a_j-h) - z_j(a_k-a_i-h) + w_k(a_i-a_j)}{z_i(a_k-a_i)-z_j(a_k-a_j)+ w (T_{i,j}z_k)(a_i-a_j)},\\
T_{i,j}z_k = \left( \frac{T_{i,j}w}{w} \right)\frac{z_iz_j(a_i-a_j) + z_iz_k(a_k-a_i) + z_jz_k(a_j-a_k)}{z_i(a_k-a_j) + z_j(a_i-a_k) + z_k(a_j-a_i)},\\
\label{whrew}T_{i,j}w_k = \left( \frac{T_{i,j}w}{w} \right)\frac{z_iz_j(a_i-a_j) + z_iw_k(a_i+h-a_k) - z_jw_k(a_j+h-a_k)}{z_i(a_j+h-a_k) - z_j(a_i+h-a_k) - w_k(a_i-a_j)},
\end{gather}
for $k \neq i, j$
\begin{gather}\allowdisplaybreaks
\frac{T_{i,j}y_i}{a_i - a_j} = \frac{w^2 D(a_i+h)-(w a_{1,1}(a_i+h)+z_i
 a_{1,2}(a_i+h)) (w a_{2,2}(a_i+h)-z_i
 a_{1,2}(a_i+h))}{h a_{1,2}(a_i+h) (T_{i,j}w)
(z_i-z_j){}^2}\nonumber\\
{} +\frac{(w a_{1,1}(a_i+h)+z_i
 a_{1,2} (a_i+h ) ) (w a_{2,2} (a_i+h )-z_j
 a_{1,2} (a_i+h ) )-w^2 P (a_i+h )}{a_{1,2} (a_i+h ) (T_{i,j}w)
 (z_i-z_j ){}^2 (a_i-a_j+h )},\!\!\!\\
\frac{T_{i,j} z_i}{a_i-a_j} = \frac{w^2 a_{2,1} (a_i+h )+w z_i a_{2,2} (a_i+h )+z_i z_j T_{i,j}(w y_i)}{w
 (T_{i,j}y_i) (z_i (a_i-a_j+h )-h z_j )},\\
w (T_{i,j}w_i) = -z_j T_{i,j}w, \qquad \frac{T_{i,j}w}{w} = 1 + \frac{(a_i-a_j)(d_1 -d_2+h)}{z_i-z_j},
\end{gather}
\end{subequations}
whereas for $k = j$ we swap the roles of $i$ and $j$ above.
\end{thm}

\begin{proof}
We temporarily use the notation $\tilde{u} = T_{i,j}u$. Given \eqref{Rhform}, we multiply the left and right-hand sides of \eqref{comp} by $(x-a_i)(x-a_j)(x-a_i-h)(x-a_j-h)$ whereby evaluating the resulting expression at $x = a_i$ gives us
\begin{gather*}
((a_i+h)I + R_0)\tilde{y}_i \begin{pmatrix} 1 \vspace{1mm}\\ \dfrac{z_i}{w} \end{pmatrix} \begin{pmatrix} w_i & w \end{pmatrix},
\end{gather*}
which specif\/ies that the rows of $((a_i+h)I + R_0)$ are annihilated by the image of $A(a_i)$. Imposing the same condition for $x=a_j$ uniquely specif\/ies $R_0$ by
\begin{gather}\label{R0h1}
R_0 = \frac{1}{z_i-z_j} \begin{pmatrix} (a_i+h)z_j - (a_j+h)z_i & (a_j-a_i)w \vspace{1mm}\\ \dfrac{(a_i-a_j)z_iz_j}{w} & (a_j+h)z_j - (a_i+h)z_i \end{pmatrix}.
\end{gather}
Using the values $x= a_i + h$ and $x=a_j+h$ gives us
\begin{gather*}
R_0 = \frac{1}{\tilde{w}_i - \tilde{w}_j} \begin{pmatrix} (a_j+h)\tilde{w}_j - (a_i+h)\tilde{w} & (a_j-a_i)\tilde{w} \vspace{1mm}\\
\dfrac{(a_i-a_k)\tilde{w}_i\tilde{w}_j}{\tilde{w}} & (a_i+h)\tilde{w} \end{pmatrix},
\end{gather*}
whose equivalence with \eqref{R0h1} gives the f\/irst part of \eqref{whrew}. Using \eqref{R0h1} with \eqref{comp} at $x= a_k$ gives us
\begin{gather*}
y_kR(a_k+h) \begin{pmatrix} 1 \\ \frac{z_k}{w} \end{pmatrix} \begin{pmatrix} w_k & w \end{pmatrix} = \tilde{y}_k \begin{pmatrix} 1 \\ \frac{\tilde{z}_k}{\tilde{w}} \end{pmatrix} \begin{pmatrix} \tilde{w}_k & \tilde{w} \end{pmatrix} R(a_k),
\end{gather*}
which is equivalent to (\ref{tijrhy})--(\ref{whrew}). The remaining parts may be calculated from evaluating
\begin{gather*}
\tilde{A}(x) = R(x+h)A(x)R(x)^{-1},
\end{gather*}
which is equivalent to \eqref{comp} using \eqref{R0h1} at $x = a_i +h$. The symmetry and uniqueness of $R(x)$ determines that the corresponding formula for $k=j$ may be obtained by swapping the roles of~$i$ and~$j$.
\end{proof}

While we have chosen to express the system in this way, this is not to be considered a~$3(m+1)$-dimensional map since it has enough constants with respect to $T_{i,j}$ to be considered a~$(N - 2)$-dimensional system in terms of the~$u_i$.

The symmetric version may be treated in the same way by considering transformations of~$B(x)$ instead of $A(x)$. We take $A(x)$ to be given by~\eqref{AviaB} where~$B(x)$ is given by~\eqref{prodform}, in which case we may parameterize~$B(x)$ in the same way by introducing the variables~$y_i$,~$z_i$ and~$w_i$ by
\begin{gather*}
B(a_i) = y_i \begin{pmatrix} 1 \vspace{1mm}\\ \dfrac{z_i}{w} \end{pmatrix} \begin{pmatrix} w_i & w \end{pmatrix} = \begin{pmatrix} w_i y_i & w y_i \vspace{1mm}\\
 \dfrac{w_i y_i z_i}{w} & y_i z_i \end{pmatrix},
\end{gather*}
for $i = 1,\ldots, N'$. The Lagrangian interpolation is the same as it was for $A(x)$ above, hence the entries of $B(x) = (b_{i,j}(x))$ are also given by \eqref{largrangeinth}. We may calculate the ef\/fect of~$E_{i,j}$ and $F_{i,j}$ on these new variables.

\begin{prop}\label{eh}
The system~\eqref{heij} is equivalent to the following action on the variables~$y_k$,~$z_k$ and $w_k$
\begin{subequations}
\begin{gather}
\label{hEij1}E_{i,j}y_k = \frac{w y_k \big(a_i^2 (z_k-z_j )+a_k^2(z_j-z_i)+a_j^2(z_i-z_k)\big)}{(a_i-a_k) (a_k-a_j)
 w(h+t) (z_i-z_j)},\\
E_{i,j} z_k = \frac{(E_{i,j}w) \big(a_i^2 z_i (z_k-z_j)+a_j^2 z_j
(z_i-z_k)+a_k^2 z_k (z_j-z_i)\big)}{w \big(a_i^2 (z_k-z_j )+a_k^2 (z_j-z_i )+a_j^2 (z_i-z_k )\big)},\\
\label{hEij3} wE_{i,j} w_k = w_k E_{i,j}w, \qquad \frac{E_{i,j} w}{w} = 1 - \frac{a_i^2 - a_j^2}{z_i - z_j},
\end{gather}
for $k \neq i,j$ and for $k = i$
\begin{gather}
(E_{i,j}y_iw) = \frac{(a_i-a_j)(wa_{2,2}(-a_i) - a_{1,2}(-a_i)z_i)}{2a_i(z_j-z_i)},\\
wE_{i,j}z_i = z_j E_{i,j}w, \qquad E_{i,j} w_i = E_{i,j}w \frac{z_i a_{1,1}(-a_i) - w a_{2,1}(-a_i)}{z_i a_{1,2}(-a_i) - w a_{2,2}(-a_i)},
\end{gather}
\end{subequations}
whereas for $k=j$ we swap the roles of $i$ and $j$ above.
\end{prop}

\begin{prop}\label{fh}
The system \eqref{hfij} is equivalent to the following action on the variables~$y_k$,~$z_k$ and $w_k$
\begin{subequations}
\begin{gather}
\label{hFij1}\frac{F_{i,j}wy_k}{w y_k} =\frac{w_i(a_k-a_j(a_j+h))- w_j(a_k-a_i(a_i+h)) - w_k(a_i-a_j)(a_i+a_j+h) }{(a_i-a_k) (a_k-a_j)
(w_i-w_j) (a_i-a_k+h)(-a_j+a_k-h)}, \\
\frac{wF_{i,j} w_k}{F_{i,j}w} = -\frac{w_iw_j(a_j\!-\!a_i)(a_i\!+\!a_j\!+h)\! +\! w_i w_k (a_i(a_i\!+\!h)\! -\! a_k)\! -\! w_j w_k(a_j(a_j\!+\!h)\! - \! a_k)}{w_k(a_j\!-\!a_i)(a_i\!+\!a_j\! +\! h)\! +\! w_j(a_i(a_i\!+\!h) \!-\! a_k)\! -\! w_i(a_j(a_j\!+h)\!-a_k)},\!\!\!\!\\
\label{hFij3}wF_{i,j}z_k =z_k F_{i,j} w, \qquad F_{i,j} w = 1- \frac{(a_i-a_j)(a_i+a_j+h)}{w_i-w_j},
\end{gather}
for $k \neq i,j$ and for $k = i$
\begin{gather}
(F_{i,j}y_iw) = \frac{(a_i-a_j)(wa_{1,1}(-a_i-h) - a_{1,2}(-a_i-h)w_i)}{(2a_i+h)(w_j-w_i)},\\
w F_{i,j} w_i = w_j F_{i,j} w, \qquad \frac{F_{i,j}}{F_{i,j} w} = \frac{w_i a_{2,2}(-h-a_i) - w a_{2,1}(-h-a_i)}{a_{1,2}(-h-a_i)- w a_{1,1}(-h-a_i)},
\end{gather}
\end{subequations}
whereas for $k=j$ we swap the roles of $i$ and $j$ above.
\end{prop}

\begin{proof}[Proof of Propositions~\ref{eh} and~\ref{fh}]
We note that for $B(x)$ to be of the same form we require that $R_l(x)$ and $R_r(x)$ from \eqref{leftRB} and \eqref{rightRB} take the forms
\begin{alignat*}{3}
& R_l(x) = x^2 I + R_0, \qquad && \det R_l(x) = (x-a_i)(x-a_j)(x+a_i)(x+a_j),& \\
& R_r(x) = x(x+h) + R_1, \qquad && \det R_r(x) = (x-a_i)(x-a_j)(x+a_i+h)(x+a_j+h),&
\end{alignat*}
with $\lambda = (x-a_i)^{-1}(x-a_j)^{-1}$. We may multiply the left and right-hand sides of~\eqref{leftRB} and~\eqref{rightRB} by $(x-a_i)(x-a_j)$ to see that $R_0$ and $R_1$ satisfy
\begin{gather*}
\big(x^2+R_0\big)A(x) = (x-a_i)(x-a_j)\tilde{A}(x), \qquad \big(x^2+R_1\big)A(x) = (x-a_i)(x-a_j)\hat{A}(x),
\end{gather*}
where we use the notation $E_{i,j}u = \tilde{u}$ and $F_{i,j}u = \hat{u}$ for the parameters of $A(x)$ and $A(x)$ itself. Evaluating at $x = a_i$ and $x = a_j$ gives us the two matrices
\begin{gather*}
R_l(x) = \begin{pmatrix} (x-a_i)(x+a_i) & 0 \\ 0 & (x-a_j)(x+a_j) \end{pmatrix}
+ \frac{\big(a_i^2-a_j^2\big)}{z_i-z_j} \begin{pmatrix} z_i & -w \vspace{1mm}\\ \dfrac{z_iz_j}{w} & -z_i \end{pmatrix}, \\
R_r(x) = \begin{pmatrix} (x-a_i)(x+a_i+h) & 0 \\ 0 & (x-a_j)(x+a_j+h) \end{pmatrix}\\
\hphantom{R_r(x) =}{}
+ \frac{(a_i-a_j)(h+a_i+a_j)}{w_i-w_j} \begin{pmatrix}- w_j & w \vspace{1mm}\\ \dfrac{w_iw_j}{w} & w_j \end{pmatrix},
\end{gather*}
from which using these values in \eqref{leftRB} and \eqref{rightRB} evaluated at $x= a_k$ give (\ref{hEij1})--(\ref{hEij3}) and (\ref{hFij1})--(\ref{hFij3}) easily follow.
\end{proof}

\subsection[$q$-dif\/ference Garnier systems]{$\boldsymbol{q}$-dif\/ference Garnier systems}

Let us consider the $q$-dif\/ference Garnier system, def\/ined by \eqref{linearsigma} where $\sigma = \sigma_q$, where $A(x)$ is specif\/ied by~\eqref{prodform} for $N = 2m+2$ and $L_i$ is given by~\eqref{qdifffactor}. We may diagonalize the leading coef\/f\/icient matrices around $x = 0$ and $x=\infty$ provided $\theta_1 \neq \theta_2$ and $\kappa_1 \neq \kappa_2$ using a lower diagonal constant matrix. From this matrix, we def\/ine a new set of variables, $y_i$, $z_i$ and $w_i$, for
\begin{gather*}
A(a_i) = y_i \begin{pmatrix} 1 \vspace{1mm}\\ \dfrac{z_i}{w} \end{pmatrix} \begin{pmatrix} w_i & w \end{pmatrix} = \begin{pmatrix} w_i y_i & w y_i \vspace{1mm}\\
 \dfrac{w_i y_i z_i}{w} & y_i z_i \end{pmatrix},
\end{gather*}
for $i=1, \ldots, N$. This specif\/ication in terms of the image and kernel of $A(a_i)$ means that we may use~\eqref{kerq} and/or~\eqref{Imq} and the action of $S_n$ to determine the values of $z_i/w$ and $w_i/w$. This def\/ines $3N$ parameters, many of which are redundant. However, if we choose the f\/irst $N$ (or any collection), we may reconstruct~$A(x)$ using Lagrangian interpolation using any collection of~$m+1$ values with the following data:
\begin{itemize}\itemsep=0pt
\item $a_{i,i}(x) = \kappa_i x^{m+1} + O(x^m)$ with $a_{1,1}(a_k) = y_kw_k$ and $a_{2,2}(a_k) = y_kz_k$ for $k=1,\ldots, m+1$,
\item $a_{1,2}(a_k) = wy_k$ and $a_{2,1} = w_ky_kz_k/w$ for $k=1, \ldots, m+1$.
\end{itemize}
If we let this collection be the f\/irst $m$ values, and let $D(x)$ satisfy \eqref{deltapoly}. We use this to express the entries of $A(x)$ as
\begin{subequations}\label{LagrangianIntq}\allowdisplaybreaks
\begin{gather}
a_{1,1}(x) = \kappa_1 D(x) \left[1+ \sum_{i=1}^{m} \frac{w_iy_i}{D'(a_i)(x-a_i)} \right], \\
a_{1,2}(x) = \kappa_2 w D(x) \left[\sum_{i=1}^{m} \frac{y_i}{D'(a_i)(x-a_i)} \right],\\
a_{2,1}(x) = \frac{\kappa_1 D(x)}{w} \left[\sum_{i=1}^{m} \frac{w_iz_iy_i}{D'(a_i)(x-a_i)} \right],\\
a_{2,2}(x) = \kappa_2 D(x) \left[1+\sum_{i=1}^{m} \frac{z_iy_i}{D'(a_i)(x-a_i)} \right].
\end{gather}
\end{subequations}
After diagonalizing the leading coef\/f\/icient in the polynomial expansion in $x$, it is easy to see that the matrix inducing $T_{i,j}$ takes the form
\begin{gather}\label{Rqformywz}
R(x) = \frac{xI + R_0}{(x-qa_i)(x-qa_j)},
\end{gather}
from which we may calculate the equivalent action on the variables $y_i$, $z_i$, $w_i$ and $w$.

\begin{prop}
The action of $T_{i,j}$ specified by the action of $S_n$ and \eqref{conprevuh} is equivalent to the following action on the variables $y_k$, $z_k$ and $w_k$:
\begin{subequations}\label{qGarnier}
\begin{gather}
\frac{(T_{i,j} wy_k)}{w}= y_k\frac{(qa_j - a_k)(qa_i-a_k)}{(z_i-z_j)^2} \left(\frac{w_k+z_j}{qa_j-a_k} - \frac{w_k+z_i}{qa_i-a_k} \right)\left( \frac{z_k-z_j}{a_k-a_j} - \frac{z_k-z_i}{a_k-a_i}\right),\!\!\!\!\\
\frac{T_{i,j}z_k}{z_iz_j} \left( \frac{z_k-z_j}{a_k-a_j} - \frac{z_k-z_i}{a_k-a_i} \right) = \frac{T_{i,j}w}{w}\left(\frac{1}{z_j} \frac{z_k-z_j}{a_k-a_j} - \frac{1}{z_i}\frac{z_k-z_i}{a_k-a_i}\right),\\
\label{qwk} (T_{i,j}w_k)\left(\frac{w_k+z_i}{a_k-qa_i}-\frac{w_k+z_j}{a_k-qa_j}\right) = z_iz_j \frac{T_{i,j}w}{w} \left(\frac{1}{z_j} \frac{z_k-z_j}{a_k-a_j} - \frac{1}{z_i}\frac{z_k-z_i}{a_k-a_i}\right),\\
\label{qw} \frac{T_{i,j}w}{w} = \frac{T_{i,j}w_i - T_{i,j}w_j}{z_i - z_j} = 1+\frac{(\kappa _1 q/\kappa_2-1) (a_i-a_j)}{z_i-z_j},
\end{gather}
for $k \neq i,j$ where for $k=i$ we have $wT_{i,j} w_i = z_iT_{i,j}w$ and
\begin{gather}
\frac{T_{i,j} y_i}{a_i-a_j} = \frac{\big(w^2 \kappa_1\kappa_2D(qa_i)- (w a_{1,1} (q a_i )+z_i a_{1,2} (q a_i ) ) (w a_{2,2} (q a_i )-z_i a_{1,2} (q a_i ) )\big)}{(q-1) a_i a_{1,2} (q a_i ) (z_i-z_j ){}^2 (T_{i,j}w)}\nonumber\\
\hphantom{\frac{T_{i,j} y_i}{a_i-a_j} =}{} -\frac{\big(\kappa_1\kappa_2w^2 D(qa_i)- (w a_{1,1} (q a_i )+z_i a_{1,2} (q a_i ) ) (w a_{2,2} (q a_i )-z_j a_{1,2} (q a_i ) )\big)}{a_{12}(q a_i) (z_i-z_j){}^2 (q a_i-a_j)(T_{i,j}w)},\!\!\!\!\label{tijyi}\\
\label{tijzi}\frac{T_{i,j} z_i}{a_i-a_j}= \frac{\big(w a_{2,2} (q a_i ) z_i+w^2 a_{2,1} (q a_i )+z_i z_j (T_{i,j}w y_i)\big)}{w (T_{i,j}y_i) (z_i (q a_i-a_j )-(q-1) a_i z_j )}.
\end{gather}
\end{subequations}
Swapping $i$ and $j$ gives the case for $k = j$.
\end{prop}

\begin{proof}
For a parameter or matrix, $u$, we use the notation $\tilde{u} = T_{i,j}u$. After establishing \eqref{Rqformywz}, we may multiply \eqref{comp} by $(x-a_i)(x-a_j)(x-qa_i)(x-qa_j)$, whereby cancelling the denominators and evaluating at $x = qa_i$ shows that
\begin{gather*}
\tilde{y}_i \begin{pmatrix} 1 \vspace{1mm}\\ \dfrac{\tilde{z}_i}{\tilde{w}} \end{pmatrix} \begin{pmatrix} \tilde{w}_i & \tilde{w} \end{pmatrix}(qa_i I + R_0) = 0,
\end{gather*}
which specif\/ies that the columns are in the kernel of $\tilde{A}(qa_i)$, whereas evaluating \eqref{comp} at $x=qa_j$ gives a similar equation which is enough to uniquely specif\/ies $R_0$, which can be written explicitly as
\begin{gather*}
\frac{1}{\tilde{w}_i-\tilde{w}_j} \begin{pmatrix} qa_j\tilde{w}_j - qa_i \tilde{w}_i & q(a_j-a_i)\tilde{w} \\ q(a_i-a_j)\tilde{w}_i\tilde{w}_j/\tilde{w} & qa_i\tilde{w}_j - qa_j \tilde{w}_i \end{pmatrix}.
\end{gather*}
Evaluating \eqref{comp} at $x= a_i$ gives that
\begin{gather*}
(qa_iI + R_0) y_i \begin{pmatrix} 1 \vspace{1mm}\\ \dfrac{z_i}{w} \end{pmatrix} \begin{pmatrix} w_i & w \end{pmatrix},
\end{gather*}
which specif\/ies that the rows are in the kernel, which means $R_0$ may be computing in terms of~$z_i$ and~$z_j$, which is explicitly given by
\begin{gather*}
R_0 =\frac{1}{z_i-z_j} \begin{pmatrix} q a_i z_j-q a_j z_i & q w (a_j-a_i )\vspace{1mm} \\
 \dfrac{q \left(a_i-a_j\right) z_i z_j}{w} & q a_j z_j-q a_i z_i
\end{pmatrix}.
\end{gather*}
The comparison of these values specif\/ies $wT_{i,j} w_i = z_iT_{i,j}w$ which implies \eqref{qw}. The second part of~\eqref{qw} is specif\/ied by looking at the leading order expansion of~\eqref{comp} in the top right-hand entry. The remaining values of are easily and uniquely determined by evaluating \eqref{comp} at $x= a_k$.

We need only determine the action on $y_i$ and $z_i$, which can be achieved by evaluating
\begin{gather*}
\tilde{A}(x) = R(qx)A(x)R(x)^{-1},
\end{gather*}
at $x = qa_i$, whereby using the value of $R_0$ above gives \eqref{tijyi} and \eqref{tijzi}. By Proposition \ref{qlattice}, the uniqueness of $R(x)$ shows $T_{i,j} = T_{j,i}$, and the symmetry of $A(x)$ with respect to swapping $i$ and $j$ implies the action on $y_j$ and $z_j$ are obtained by swapping $i$ and $j$ in \eqref{tijyi} and~\eqref{tijzi}.
\end{proof}

The resulting form of the evolution was called the birational form of the $q$-Garnier system in~\cite{Sakai:Garnier}.

\begin{rem}
The author of \cite{Sakai:Garnier} also produces another parameterization in which every root of the polynomial $a_{1,2}(x)$ is a parameter say $y_1, \ldots, y_m$, while the other parameter are the values of $z_i = a_{1,1}(y_i)$ for $i=1,\ldots, n$. This may be considered a natural extension of known parameterizations of Lax pairs for Painlev\'e equations and discrete Painlev\'e equations. The issue in def\/ining a collection of variables in this way is that we can only formally distinguish the roots of $a_{1,2}(x)$. A discrete isomonodromic will produce $\tilde{a}_{1,2}(x)$, whose roots are $\tilde{y}_1, \ldots, \tilde{y}_n$, yet there is no way of ordering the $y_i$ and $\tilde{y}_i$ in a way that makes the mapping $y_i \to \tilde{y}_i$. The space formed by considering set of roots of monic polynomials of degree $n$ is a construction for the $n$-th symmetric power of $\mathbb{C}$, which may be consider the correct setting for such a parameterization. In the continuous setting, this parameterization makes more sense as the variables change continuously.
\end{rem}

Let us now start with a matrix satisfying $A(x)A(1/qx) = I$, then we take $A(x)$ to be given by \eqref{AviaB} where $B(x)$ is given by \eqref{prodform2}. We def\/ine variables $y_i$, $z_i$ and $w_i$ by
\begin{gather*}
B(a_i) = y_i \begin{pmatrix} 1 \vspace{1mm}\\ \dfrac{z_i}{w} \end{pmatrix} \begin{pmatrix} w_i & w \end{pmatrix} = \begin{pmatrix} w_i y_i & w y_i\vspace{1mm}\\
 \dfrac{w_i y_i z_i}{w} & y_i z_i \end{pmatrix},
\end{gather*}
for $i = 1,\ldots, N'$. The Lagrangian interpolation is equivalent to the formulation for~$A(x)$ above, hence the entries of $B(x) = (b_{i,j}(x))$ are also given by~\eqref{LagrangianIntq}. We may calculate the ef\/fect of~$E_{i,j}$ and~$F_{i,j}$ on these new variables.

\begin{prop}\label{eq}
The system \eqref{qeij} is equivalent to the following action on the variables~$y_k$,~$z_k$ and $w_k$
\begin{subequations}\allowdisplaybreaks
\begin{gather}
\label{qEij1}\frac{(E_{i,j}y_k) \tilde{w}}{wy_k} = \frac{a_j (a_i a_k-1)(z_j-z_k)}{(a_j-a_k) (z_i-z_j)}-\frac{a_i (a_j a_k-1)(z_i-z_k)}{(a_i-a_k) (z_i-z_j)},\\
\frac{E_{i,j}z_ky_k}{w_ky_k} - \frac{a_ia_jz_k}{w} =\frac{a_i z_j (a_i a_j-1) (z_i-z_k)}{w(a_k-a_i) (z_i-z_j)}+\frac{a_j z_i (a_i a_j-1) (z_j-z_k)}{w
(a_j-a_k) (z_i-z_j)},\\
\label{qEij3}wE_{i,j}w_k = w_kE_{i,j}w, \qquad
\frac{E_{i,j} w}{w} = 1 + \frac{(a_j-a_i)(a_ia_j-1)}{a_ia_j(z_i-z_j)},
\end{gather}
for $k \neq i,j$ and for $k=i$ with
\begin{gather}
\label{qEij4}\frac{(E_{i,j}wy_i)}{a_i -a_j} = \frac{b_{1,2}\left(\frac{1}{a_i}\right) z_i-b_{2,2}\left(\frac{1}{a_i}\right) w}{\left(a_i-\frac{1}{a_i}\right) (z_i-z_j)},\\
\label{qEij5} w E_{i,j} z_i = z_j E_{i,j}w, \qquad E_{i,j}\frac{w_i}{w} = \frac{wb_{2,1}\left(\frac{1}{a_i}\right) - b_{1,1}\left(\frac{1}{a_i}\right)z_i}{wb_{2,2}\left(\frac{1}{a_i}\right) - b_{1,2}\left(\frac{1}{a_i}\right)z_i},
\end{gather}
\end{subequations}
with the equivalent form for $k=j$ obtained by interchanging $i$ and $j$.
\end{prop}

\begin{prop}\label{fq}
The system \eqref{qfij} is equivalent to the following action on the variables~$y_k$,~$z_k$ and $w_k$
\begin{subequations}
\begin{gather}
\label{qFij1}\frac{(F_{i,j}y_kw)-qa_ia_jwy_k}{(1-qa_ia_j)wy_k} = \frac{a_i (w_i-w_k)}{(a_i-a_k) (w_i-w_j)}+\frac{a_j (w_j-w_k)}{(a_k-a_j) (w_i-w_j)},\\
F_{i,j} z_k - \frac{w_jz_kF_{i,j}w}{w^2} = \frac{a_j y_k z_k (w_j-w_k) (q a_i a_k-1)}{w (a_j-a_k) (F_{i,j}y_k)},\\
\label{qFij3} z_k F_{i,j}w_k = w F_{i,j} z_k, \qquad \frac{F_{i,j} w}{w} = 1 + \frac{(a_i-a_j)(1-qa_ia_j)\kappa_1}{qa_ia_j(w_i-w_j)\kappa_2},
\end{gather}
for $k \neq i,j$ whereas for $k =i$ we have
\begin{gather}
(F_{i,j} wy_i) = \frac{(1-q)qa_ia_ja_{1,2}\left( \frac{1}{qa_i} \right)}{(1-qa_i^2)(1-qa_ia_j)}\nonumber \\
\hphantom{(F_{i,j} wy_i) =}{} - \frac{q^2 a_i \left(a_i-a_j\right) (a_i a_j-1) \left(z_j a_{12}\left(\frac{1}{q a_i}\right)+w a_{1,1}\left(\frac{1}{q a_i}\right)\right)}{\big(q a_i^2-1\big)(z_i-z_j) (q a_i a_j-1)},\label{qFij4}\\
\label{qFij5}\frac{F_{i,j} w_i y_i}{z_i} + \frac{F_{i,j}w y_i}{w} = \frac{(1-q)qa_ia_j a_{1,2}\left(\frac{1}{qa_i}\right)}{\big(1-qa_i^2\big)(1-qa_ia_j)w} +\frac{(1-q)qa_ia_j a_{1,1}\left(\frac{1}{qa_i}\right)}{\big(1-qa_i^2\big)(1-qa_ia_j)z_i},
\end{gather}
\end{subequations}
with the equivalent form for $k=j$ obtained by interchanging $i$ and $j$.
\end{prop}

\begin{proof}[Proof of Propositions \ref{eq} and \ref{fq}]
We wish to take a dif\/ferent approach from the proofs of Propositions~\ref{eh} and~\ref{fh} by deducing $R_l(x)$ and $R_r(x)$ in terms of $\tilde{A}(x)$ and $\hat{A}(x)$ where $\tilde{u} = E_{i,j} u$ and $\hat{u} = F_{i,j}u$ respectively. Since we know the determinant of $R_l(x)$ must include and factor of $(x-a_i)(x-a_i)$ and is symmetric with respect to the action $x \to 1/x$, we have that~$R_l(x)$ takes the form
\begin{gather*}
R_l(x) = I\left(x+ \frac{1}{x}\right) + R_0, \qquad \det R_l(x) = (x-a_i)(x-a_j)(xa_i-1)(xa_j-1)/x^2,
\end{gather*}
and $\lambda(x) = (x-a_i)^{-1}(x-a_j)^{-1}$, whereas $R_r$ is symmetric with respect to $x \to 1/qx$, hence~$R_r(x)$ takes the form
\begin{gather*}
R_r(x) = I\left(x+ \frac{1}{qx}\right) + R_1, \qquad \det R_r(x) = (x-a_i)(x-a_j)(qxa_i-1)(qxa_j-1)/q^2x^2,
\end{gather*}
with the same $\lambda(x)$. Due to the involutive nature of the transformation, it is natural that the~$R_0$ and~$R_1$ satisfy
\begin{gather*}
(xa_i-1)(xa_j-1) A(x) = \left( \left( x + \frac{1}{x} \right)I + R_0^*\right) \tilde{A}(x),\\
(qxa_i-1)(qxa_j-1) A(x) =\hat{A}(x) \left( \left( x + \frac{1}{x} \right)I + R_1^*\right) ,
\end{gather*}
where $R_i^*$ is the cofactor matrix for $R_i$ for $i = 0,1$. Since $\tilde{\tilde{u}} = \hat{\hat{u}} = u$, these two equations are equivalent to~\eqref{leftRB} and~\eqref{rightRB} applied to the transformed values of~$A(x)$. This gives
\begin{gather*}
R_l(x) = \begin{pmatrix} x + \dfrac{1}{x} - a_j - \dfrac{1}{a_j} & 0 \\ 0 & x + \dfrac{1}{x} - a_i - \dfrac{1}{a_i} \end{pmatrix}
 + \frac{(a_i-a_j)(1-a_ia_j)}{a_ia_j(\tilde{z}_i-\tilde{z}_j)} \begin{pmatrix} \tilde{z}_j & -\tilde{w} \vspace{1mm}\\ \dfrac{\tilde{z}_i\tilde{z}_j}{\tilde{w}} & -\tilde{z}_i \end{pmatrix},\nonumber\\
R_r(x) = \begin{pmatrix} x + \dfrac{1}{qx} - a_j - \dfrac{1}{qa_j} & 0 \\ 0 & x + \dfrac{1}{qx} - a_i - \dfrac{1}{qa_i} \end{pmatrix}\\
\hphantom{R_r(x) =}{}
 + \frac{(a_i-a_j)(1-qa_ia_j)}{qa_ia_j(\hat{w}_i-\hat{w}_j)} \begin{pmatrix} \hat{w}_j & \hat{w} \vspace{1mm}\\ -\dfrac{\hat{w}_i\hat{w}_j}{\tilde{w}} & -\hat{w}_j \end{pmatrix},
\end{gather*}
from which we may calculate and equivalent form of (\ref{qEij1})--(\ref{qEij3}) and (\ref{qFij1})--(\ref{qFij3}) in terms of $\tilde{w}_i$'s and $\hat{w}_i$ respectively. Comparing entries of \eqref{leftRB} and \eqref{rightRB} using these values at $x = 1/a_i$ and $x= 1/qa_i$ gives the remaining values and brings gives (\ref{qEij4})--(\ref{qEij5}) and (\ref{qFij4})--(\ref{qFij5}). The f\/irst parts of \eqref{qEij5} and \eqref{qFij5} bring (\ref{qEij1})--(\ref{qEij3}) and (\ref{qFij1})--(\ref{qFij3}) into their presented form, similarly with $x= 1/a_j$ and $x=1/qa_j$.
\end{proof}

\section{Special cases}\label{sec:special}

We wish to demonstrate that the simplest cases of the $h$-dif\/ference and $q$-dif\/ference Garnier systems are known to coincide with discrete versions of the sixth Painlev\'e equation. Specializing the higher cases coincide with discrete Painlev\'e equations that appear higher in Sakai's hierarchy. We summarize the results in Table~\ref{sumtable}. To avoid confusion, we have used the value of $N$ and since we have used the notation $r_{1,2}$ and $r_{2,1}$ in both sections we state that the value of $r_{2,1}$ in Table~\ref{sumtable} is specif\/ied by~\eqref{r21h} and the value of $r_{1,2}$ is given by~\eqref{r12q}.

\begin{table}[!ht]\centering
\begin{tabular}{|p{4cm} |p{1cm}| p{5cm} |p{3cm}|} \hline
& $N$& conditions & Painlev\'e equation \\ \hline \hline
$h$-Garnier & $6$ & $d_1 \neq d_2$ & $d$-$\mathrm{P}\big(A_2^{(1)}\big)$ \\
 & $8$ & $d_1 = d_2$, $r_{2,1} = 0$ & $d$-$\mathrm{P}\big(A_1^{(1)}\big)$ \\ \hline
 symmetric $h$-Garnier & $8$ & $d_1 = d_2$, $r_{2,1} = 0$ & $d$-$\mathrm{P}\big(A_1^{(1)}\big)$ \\ \hline
 $q$-Garnier & $4$ & $\kappa_1\neq \kappa_2$, $\theta_1\neq \theta_2$ & $q$-$\mathrm{P}\big(A_3^{(1)}\big)$ \\
 & $6$ & $\kappa_1 = \kappa_2$, $\theta_1=\theta_2$ & $q$-$\mathrm{P}\big(A_2^{(1)}\big)$ \\
 & $8$ & $\kappa_1 = \kappa_2$, $\theta_1=\theta_2$, $r_{1,2} = 0$ & $q$-$\mathrm{P}\big(A_1^{(1)}\big)$ \\ \hline
symmetric $q$-Garnier & $8$ & $\kappa_1 = \kappa_2$, $\theta_1 = \theta_2$, $r_{1,2} = 0$ & $q$-$\mathrm{P}\big(A_0^{(1)}\big)$ \\ \hline
\end{tabular}
\caption{A summary of the special cases of discrete Garnier systems whose evolution coincides with dis\-cre\-te Painlev\'e equations.\label{sumtable}}
\end{table}

We remark that scalar Lax pairs for the $q$-dif\/ference cases of discrete Painlev\'e equations we present have also been presented in~\cite{Yamada:LaxqEs} and more recently scalar Lax pairs for the $h$-dif\/ference cases appeared in~\cite{Kajiwara2015}. A correspondence between the scalar Lax pairs and matrix Lax pairs for the $q$-$\mathrm{P}\big(A_2^{(1)}\big)$ case that appears here was constructed in~\cite{Ormerod:qE6}. Such correspondences are almost sure to exist for the other cases, however, we do not pursue these lengthy correspondences here. We do however remark that the characteristic properties of the Lax pairs presented in~\cite{Yamada:LaxqEs} and~\cite{Kajiwara2015} and scalar versions of the Lax pairs we present here seem to coincide up to some nontrivial transformations.

\subsection[The twisted $m=1$ asymmetric $q$-dif\/ference Garnier system]{The twisted $\boldsymbol{m=1}$ asymmetric $\boldsymbol{q}$-dif\/ference Garnier system}

The f\/irst system we present as a special case is the $q$-analogue of the sixth Painlev\'e equation, which we write as
\begin{subequations}\label{qP6}
\begin{gather}
z(q t)z(t) = \frac{b_3 b_4(y(t) - a_1t)(y(t)-a_2t)}{(y-a_3)(y-a_4)},\\
y(q t)y(t) = \frac{a_3a_4(z(q t) - b_1 t)(z(q t) - b_2 t)}{(z(q t) - b_3)(z(q t) - b_4)},
\end{gather}
\end{subequations}
where
\begin{gather*}
q = \frac{a_1a_2b_3b_4}{b_1 b_2 a_3a_4},
\end{gather*}
which was f\/irst presented by Jimbo and Sakai~\cite{Sakai:qP6}. We consider an associated linear problem of the from \eqref{linearqdiff}, where
\begin{gather}\label{ALqP6}
A(x) = \begin{pmatrix} \theta_1 & 0 \\ 0 & \theta_2 \end{pmatrix} \begin{pmatrix}
1 & u_1 \vspace{1mm}\\
 \dfrac{x}{a_1u_1} & 1
\end{pmatrix}\begin{pmatrix}
1 & u_2 \vspace{1mm}\\
 \dfrac{x}{a_2u_2} & 1
\end{pmatrix} \begin{pmatrix}
1 & u_3 \vspace{1mm}\\
 \dfrac{x}{a_3u_3} & 1
\end{pmatrix} \begin{pmatrix}
1 & u_4 \vspace{1mm}\\
 \dfrac{x}{a_4u_4} & 1
\end{pmatrix}.
\end{gather}
This matrix is of the form
\begin{gather*}
A(x) = A_0 + A_1 x + A_2 x^2,
\end{gather*}
where $A_0$ is upper triangular with diagonal entries $\theta_1$ and $\theta_2$, while $A_2$ is lower triangular with diagonal entries
\begin{gather*}
\kappa_1 = \frac{\theta_1u_1 u_3}{a_2a_4u_2u_4}, \qquad \kappa_1 = \frac{\theta_2u_2 u_4}{a_1a_3u_1u_3}.
\end{gather*}
The two natural consequences that
\begin{gather}
\label{detqP6}\det A(x) = \kappa_1\kappa_2 (x-a_1)(x-a_2)(x-a_3) (x-a_4),\\
\theta_1 \theta_2 = \kappa_1 \kappa_2 a_1 a_2 a_3 a_4,\nonumber
\end{gather}
which means that by diagonalizing the constant coef\/f\/icient, we may let $A_0 = \operatorname{diag}(\theta_1,\theta_2)$ and have a pair
\begin{gather*}
(A_1,A_2) \in \mathcal{M}(a_1, \ldots, a_4; \kappa_1, \kappa_2;\theta_1, \theta_2).
\end{gather*}
We may diagonalize $A_2$ in order to bring this Lax pair into the form of Jimbo and Sakai~\cite{Sakai:qP6}. We propose a slightly dif\/ferent form in which $A_0$ and $A_2$ are upper and lower triangular respectively. This gives us a simple alternative parameterization, which takes the general form
\begin{gather}\label{quasisakaiqP6}
A(x,t) = \begin{pmatrix}
\kappa_2 x^2 + \alpha x + \theta_1 & w (x-y) \vspace{1mm}\\
\dfrac{x^2 \gamma + \delta x}{w} & \kappa_2 x^2 + \beta x + \theta_2
\end{pmatrix}.
\end{gather}
We satisfy \eqref{detqP6} when $x=y$ by letting
\begin{alignat*}{3}
&a_{1,1}(x) = \kappa_1z_1, \qquad &&a_{11}(x) = \kappa_2z_2,& \\
&z_1 = \frac{(y-a_1)(y-a_2)}{z}, \qquad &&z_2 = (y-a_3)(y-a_4)z.&
\end{alignat*}
We may solve for $\alpha$, $\beta$, $\gamma$ and $\delta$ in terms of $y$, $z$ and $w$ to show
\begin{gather*}\allowdisplaybreaks
\alpha = \frac{z_1 - \kappa_1 y^2 - \theta_1}{y}, \qquad \beta = \frac{z_1 - \kappa_1 y^2 - \theta_1}{y},\\
\gamma = \kappa_1\kappa_2(a_1 + a_2 + a_3 + a_4) + \frac{\alpha}{\kappa_1} + \frac{\beta}{\kappa_2}, \\
\delta= \frac{\kappa_1\kappa_2(a_1a_2a_3 + a_1a_2a_4 + a_1a_3a_4+ a_2a_3a_4)}{y}-\frac{\theta_1\alpha + \theta_2 \beta}{y}.
\end{gather*}
The only minor dif\/ference in the theory presented above is that the constant coef\/f\/icient in the series part of the solution, $Y_{\infty}(x)$, is lower triangular, rather than the identity, as is the leading term in the discrete isomonodromic deformation.

As above, we wish to we have four variables, $u_1, \ldots,u_4$, with one constant with respect to~$T_{1,2}$, which we wish to identify with the variables $y$, $z$ and $w$. Equating the various coef\/f\/icients of~\eqref{quasisakaiqP6} with the corresponding expressions in~\eqref{ALqP6} gives the following expressions for $y$ and $z$
\begin{subequations}\label{qP6yz}
\begin{gather}
y = -\frac{a_2 a_3 u_2 u_3 (u_1+u_2+u_3+u_4)}{a_2 u_2 (u_1+u_2) u_4+a_3 u_1 u_3 (u_3+u_4)},\\
z = -\frac{a_3a_4(y-a_1)(y-a_2)(u_3+u_4)}{(y-a_3)(y-a_4)(u_1+u_2)\theta_1}.
\end{gather}
\end{subequations}
Conversely, we notice that since the right-most factor of $A(a_4)$, $L(a_4,u_4,a_4)$, has a $0$ eigenvector of the form $(u_4,-1)$, we may iteratively def\/ine $u_i$ by determining the $0$-eigenvector at $x=a_i$ for $i = 1,\ldots, 4$. For example, using $x= a_4$ we see
\begin{gather*}
u_4 = -a_{1,2}(a_4)/a_{1,1}(a_4),
\end{gather*}
which is given in terms of $y$, $w$ and $z$ above. This gives a right factor which we may remove to iteratively proceed for $x= a_3$ and so on and so forth. This gives us a one-to-one correspondence between $u_1, \ldots, u_4$ and $y$, $z$ and $w$ with $\kappa_1$ and $\kappa_2$ specif\/ied, with constraint, in terms of $u_1, \ldots, u_4$.

\begin{prop}\label{qP6ident}
The birational transformation of algebraic varieties
\begin{gather*}
T_{1,2} \colon \ \mathcal{M}_q(a_1,a_2,a_3, a_4; \kappa_1, \kappa_2;\theta_1, \theta_2) \to \mathcal{M}_q(qa_1,qa_2,a_3, a_4; \kappa_1/q, \kappa_2/q;\theta_1, \theta_2)
\end{gather*}
is equivalent to the mapping $t \to qt$ in \eqref{qP6} where the values of~$b_i$ are given by
\begin{gather}\label{qP6bvals}
b_1 = \frac{q^2 a_1a_2}{\theta_1}, \qquad b_2 = \frac{q^2 a_1 a_2}{\theta_2}, \qquad b_3 = \frac{q}{\kappa_1}, \qquad b_4 = \frac{q^2}{\kappa_2}.
\end{gather}
\end{prop}

\begin{proof}
Using \eqref{Imq}, we see that the image of $A(a_1)$ and $A(a_2)$ gives $u_1$ and $u_2$, which are explicitly given by
\begin{gather*}
u_1 = \frac{\theta _2 w (a_1-y)}{\theta _1 \big(a_1 \beta +a_1^2 \kappa _2+\theta _2\big)}, \\
u_2 = -\frac{a_1 \theta _2 w (a_1-y) (\kappa _2 (a_1 (y-a_2)+a_2 y)+\theta _2+\beta y)}{\theta _1 (a_1
(a_1 \kappa _2+\beta )+\theta _2) (a_1 (a_2 (\beta +\kappa _2 y)+\theta _2)+\theta _2
(a_2-y))}.
\end{gather*}
This determines the an $R(x)$ inducing $T_{1,2}$ in terms of $y$, $w$ and $z$. This is used into be used in~\eqref{comp}. If we temporarily introduce the notation $T_{1,2} f = \tilde{f}$ then these calculations reveal
\begin{gather*}
\tilde{w}= w\frac{q^2-\tilde{z}\kappa_1}{q-\tilde{z}\kappa_2},\qquad
\tilde{z}z = \frac{q^2}{\kappa_1\kappa_2} \frac{(y-a_1)(y-a_2)}{(y-a_3)(y-a_4)},\\
\tilde{y}y = \frac{(\theta_1\tilde{z}-q^2a_1a_2)(\kappa_1\kappa_2a_3a_4 \tilde{z}-q^2 \theta_1)}{(\kappa\tilde{z}-q)(\kappa_2 \tilde{z}-q^2)},
\end{gather*}
 which coincides with \eqref{qP6} when the $b_i$ are specif\/ied by \eqref{qP6bvals}.
\end{proof}

Alternatively, we may simply use \eqref{qP6yz} and~\eqref{qT12} and the expressions for the~$u_i$ in terms of~$y$ and~$z$.

\subsection[A special case of the $m=1$ $h$-dif\/ference Garnier system]{A special case of the $\boldsymbol{m=1}$ $\boldsymbol{h}$-dif\/ference Garnier system}\label{dP6sec}

The second system we present is the case of the dif\/ference analogue of the sixth Painlev\'e equation, which we present as
\begin{subequations}\label{hP6}
\begin{gather}
(y(t) + z(t))(y(t+h)+z(t)) = \frac{(z(t)+a_3)(z(t)+a_4)(z(t)+a_5)(z(t)+a_6)}{(z(t)+a_7+t)(z(t)+a_8+t)},\\
(y(t+h) + z(t))(y(t+h)+z(t+h))\nonumber\\
\qquad{}= \frac{(y(t+h)-a_3)(y(t+h)-a_4)(z(t+h)-a_5)(y(t+h)-a_6)}{(y(t+h)-a_1-t-h)(y(t+h)-a_2-t-h)},
\end{gather}
\end{subequations}
where
\begin{gather*}
h = a_3 + a_4 + a_5 + a_6 - a_1 - a_2 - a_7 - a_8.
\end{gather*}
We consider an associated linear problem of the from \eqref{linearqdiff}, where
\begin{gather}
A(x) = \begin{pmatrix} u_1 & 1 \\ x-a_1+u_1^2 & u_1 \end{pmatrix}\begin{pmatrix} u_2 & 1 \\ x-a_2+u_2^2 & u_2 \end{pmatrix}\begin{pmatrix} u_3 & 1 \\ x-a_3+u_3^2 & u_3 \end{pmatrix}\nonumber\\
\hphantom{A(x) =}{} \times\begin{pmatrix} u_4 & 1 \\ x-a_4+u_4^2 & u_4 \end{pmatrix}\begin{pmatrix} u_5 & 1 \\ x-a_5+u_5^2 & u_5 \end{pmatrix}\begin{pmatrix} u_6 & 1 \\ x-a_6+u_6^2 & u_6 \end{pmatrix},\label{ALhP6}
\end{gather}
where we impose the constraint
\begin{gather*}
\sum_{i=1}^{8} u_i = 0.
\end{gather*}
This product takes the general form
\begin{gather*}
A(x) = A_0 + A_1 x + A_2 x^2 + A_3x^3,
\end{gather*}
and may be expressed in the general form
\begin{gather*}
A(x) = x^3I + \begin{pmatrix} d_1 ((x-\alpha)(x-y) + z_1) & w(x-y) \\ a_{2,1}(x) & d_2((x-\beta)(x-y)+z_2) \end{pmatrix},
\end{gather*}
where $a_{2,1}(x)$ is a polynomial of degree $2$ before we diagonalize $A_2$. After diagonalizing $A_2$ it becomes a linear function in $x$, which we write as
\begin{gather*}
a_{2,1}(x) = \frac{\gamma x + \delta}{w}.
\end{gather*}
The values of $\alpha$, $\beta$, $\gamma$ and $\delta$ are uniquely determined by \eqref{deth}. The values of $z_1$ and $z_2$ are satisfy
\begin{gather*}
\big( y^3 + d_1 z_1 \big)\big( y^3 + d_2 z_2 \big) = (y-a_1)(y-a_2)(y-a_3)(y-a_4)(y-a_5)(y-a_6).
\end{gather*}
This relation is solved by introducing a variable, $z$, via
\begin{gather*}
 y^3 + d_1 z_1 = \frac{(y-a_3)(y-a_4)(y-a_5)(y-a_6)}{y+z},\\
 y^3 + d_2 z_2 = (y-a_1)(y-a_2)(y+z).
\end{gather*}
We also have that the variables $d_1$ and $d_2$ are specif\/ied by
\begin{gather*}
d_1 = a_1 + a_3 + a_5 + u_1^2 + u_3^2 + u_5^2 + \sum_{i=1}^6 \sum_{j=1}^{i-1} u_iu_j,\\
d_2 = a_2 + a_4 + a_5 + u_2^2 + u_4^2 + u_6^2 + \sum_{i=1}^6 \sum_{j=1}^{i-1} u_iu_j,
\end{gather*}
which are known to be constant with repect to $T_{i,j}$. Using the determinantal relations, and the correspondence between~$r_i$ and~$d_i$, we have by setting $A_3 = I$, we have the $3$-tuple
\begin{gather*}
(A_0, A_1, A_2) \in \mathcal{M}(a_1, \ldots, a_6; d_1,d_2;1,1).
\end{gather*}

\begin{thm}\label{hP6ident}
The action of the translation $T_{1,2}$ is equivalent to \eqref{hP6} where $a_7$ and $a_8$ are given by
\begin{gather}\label{hP6achoice}
a_7 = -h - a_1 - a_2 - d_1,\qquad a_8= a_3 + a_4 + a_5 + a_6 + d_1.
\end{gather}
\end{thm}

\begin{proof}
This follows much the same way as Proposition~\ref{qP6ident}, however, there is an added dif\/f\/iculty in that the diagonalization of $A_2$ introduces a non-trivial correspondence between the mat\-rix~\eqref{ALhP6} and its corresponding~$R(x)$, denoted~$R'(x)$, to be used in~\eqref{comp}. The resulting mat\-rix,~$R'(x)$, can be shown to be of the form
\begin{gather*}
R'(x) = \frac{x I + R_0}{(x- a_1-h)(x- a_2-h)},
\end{gather*}
for some constant matrix $R_0$, which can be calculated using \eqref{comp}. The uniqueness of $R(x)$ ensures this calculation coincides with~$T_{1,2}$, as def\/ined by~\eqref{conprevuh}. Using this same over determined relation, namely~\eqref{comp}, we may determine that the mapping in terms of the variables~$y$ and~$z$ are specif\/ied by
\begin{gather*}
(\tilde{y}+z)(\tilde{z}+\tilde{y}) = \frac{(\tilde{y}-a_3)(\tilde{y}-a_4)(\tilde{y}-a_5)(\tilde{y}-a_6)}{(\tilde{y}-a_1-h)(\tilde{y}-a_2-h)},\\
(\tilde{y}+z)(y+z) = \frac{(z+a_3)(z+a_4)(z+a_5)(z+a_6)}{(z+a_3+a_4+a_5+a_6+d_1)(z-d_1-a_1-a_2-h)},
\end{gather*}
where $\tilde{y}$ and $\tilde{z}$ are identif\/ied with $T_{1,2} y$ and $T_{1,2}z$ respectively. This coincides with \eqref{hP6} with~$a_7$ and~$a_8$ specif\/ied by~\eqref{hP6achoice}.
\end{proof}

\subsection[An extra special case of the $m=3$ asymmetric and symmetric $h$-dif\/ference Garnier system]{An extra special case of the $\boldsymbol{m=3}$ asymmetric\\ and symmetric $\boldsymbol{h}$-dif\/ference Garnier system}

We have one more special case to consider in the symmetric case, when we allow $A(x)$ to be given by the product
\begin{gather}
A(x) =\begin{pmatrix} u_1 & 1 \\ x-a_1+u_1^2 & u_1 \end{pmatrix}
\begin{pmatrix} u_2 & 1 \\ x-a_2+u_2^2 & u_2 \end{pmatrix}
\begin{pmatrix} u_3 & 1 \\ x-a_3+u_3^2 & u_3 \end{pmatrix}\nonumber\\
\hphantom{A(x) =}{} \times \begin{pmatrix} u_4 & 1 \\ x-a_4+u_4^2 & u_4 \end{pmatrix}
\begin{pmatrix} u_5 & 1 \\ x-a_5+u_5^2 & u_5 \end{pmatrix}
\begin{pmatrix} u_6 & 1 \\ x-a_6+u_6^2 & u_6 \end{pmatrix}\nonumber\\
\hphantom{A(x) =}{} \times \begin{pmatrix} u_7 & 1 \\ x-a_7+u_7^2 & u_7 \end{pmatrix}
\begin{pmatrix} u_8 & 1 \\ x-a_8+u_8^2 & u_8 \end{pmatrix},\label{hAprodform8}
\end{gather}
where we impose the constraint that
\begin{gather*}
u_1 + u_2 + u_3 + u_4 + u_5 +u_6 + u_7 + u_8 = 0.
\end{gather*}
Under this constraint, the coef\/f\/icient of $x^3$, denoted $A_3$, takes the form
\begin{gather*}
A_3 = \begin{pmatrix}
d_1 & 0 \\
d_{2,1} & d_2
\end{pmatrix}.
\end{gather*}
We introduce one more constraint that $d_1 = d_2$, where $d_1$ and $d_2$ are def\/ined by \eqref{dvals}. It is clear from above that $T_{1,2} d_i = d_i + h$, however, it is also easy to show that
\begin{gather*}
(T_{1,2} - I) d_{2,1} = (d_1-d_2)(u_1 + u_2),
\end{gather*}
hence, $d_{1,2}$ is constant with respect to $T_{1,2}$ if $d_1 = d_2$. We impose the constraint that the expression for $d_{1,2}$ is identically $0$ for $A(x)$ to def\/ine a regular system, so that $A_3 = d_1 I = d_2I$, where the equality implies that
\begin{gather*}
d_1 = d_2 = \frac{1}{2}\sum_{i=1}^{8} a_i.
\end{gather*}
As $d_1$ and $d_2$ are def\/ined in terms of the $u_i$ by
\begin{gather*}
d_1 = a_1 + a_3 + a_5 +a_7+ u_1^2 + u_3^2 + u_5^2 +u_7^2+ \sum_{i=1}^8 \sum_{j=1}^{i-1} u_iu_j,\\
d_2 = a_2 + a_4 + a_5 +a_8+ u_2^2 + u_4^2 + u_6^2 + u_8^2+ \sum_{i=1}^8 \sum_{j=1}^{i-1} u_iu_j,
\end{gather*}
this is considered an extra constraint on the $u_i$. The map resulting from~$T_{1,2}$ is two-dimensional, which an additional dif\/ference equations satisf\/ied by one additional gauge freedom. The result is a matrix of the general form
\begin{gather}
A(x) = (x-a_1)(x-a_2) \begin{pmatrix} x^2 + \alpha_1 x + \alpha_2 & w \vspace{1mm}\\ \dfrac{\gamma}{w} & x^2 + \beta_1 x + \beta_2 \end{pmatrix}\nonumber\\
\hphantom{A(x) =}{} + \frac{x-a_2}{a_1-a_2} y_1 \begin{pmatrix} 1 \vspace{1mm}\\ \dfrac{z_1}{w} \end{pmatrix} \begin{pmatrix} w_1 & w \end{pmatrix} + \frac{x-a_1}{a_2-a_1} y_2 \begin{pmatrix} 1 \vspace{1mm}\\ \dfrac{z_2}{w} \end{pmatrix} \begin{pmatrix} w_2 & w \end{pmatrix},\label{A0form}
\end{gather}
where
\begin{gather*}
\alpha_1 = \beta_1 = \frac{a_1}{2} + \frac{a_2}{2}-\sum_{i=3}^{8} \frac{a_i}{2}.
\end{gather*}
The determinant at $x = a_1$ and $x=a_2$ are automatically $0$ by construction. This is also a~polynomial of degree six with six nontrivial
conditions to satisfy, which are suf\/f\/icient to write down expressions for~$\alpha_2$,~$\beta_2$ and~$\gamma$.

\begin{prop}
The map $T_{1,2}$ on the variables $(z_1,z_2,w_1,w_2,y_1,y_2)$ is given by
\begin{gather*}\allowdisplaybreaks
T_{1,2} z_1 =
\big\{\tilde{w} \big(w^2 \bar{a}_{21} (z_2 (a_2-a_1-h)+h z_1)+w z_1 (h z_1 \bar{a}_{22}+z_2 (\bar{a}_{22} (a_2-a_1-h)\\
\hphantom{T_{1,2} z_1 = }{} -(a_2-a_1)
 \bar{a}_{11}))+(a_1-a_2) z_1^2 z_2 \bar{a}_{12}\big)\big\}\big/\big\{w (z_1 (w (\bar{a}_{11} (a_1-a_2+h)+(a_2-a_1) \bar{a}_{22})\\
\hphantom{T_{1,2} z_1 = }{}
 -h z_2 \bar{a}_{12})+w ((a_2-a_1) w \bar{a}_{21}-h z_2 \bar{a}_{11})+z_1^2 \bar{a}_{12} (a_1-a_2+h))\big\}, \\
T_{1,2} z_2 =
\big\{\tilde{w} \big(\hat{a}_{21} w^2 (z_1 (a_1-a_2-h)+h z_2)+w z_2 (\hat{a}_{22} h z_2+z_1 (\hat{a}_{22} (a_1-a_2-h)\\
\hphantom{T_{1,2} z_2 =}{} -(a_1-a_2)
 \hat{a}_{11}))+(a_2-a_1) \hat{a}_{12} z_1 z_2^2\big)\big\}\big/\big\{w (z_2 (w (\hat{a}_{11} (a_2-a_1+h)+(a_1-a_2) \hat{a}_{22})\\
\hphantom{T_{1,2} z_2 =}{}
 -\hat{a}_{12} h z_1)+w ((a_1-a_2) \hat{a}_{21} w-\hat{a}_{11} h z_1)+\hat{a}_{12} z_2^2 (-a_1+a_2+h))\big\}, \\
 T_{1,2} y_1 = \frac{(a_1-a_2) (w \bar{a}_{11}+z_1 \bar{a}_{12})}{(z_1-z_2) \tilde{w} (a_1-a_2+h)}-\frac{(a_1-a_2){}^2 (w^2 \bar{a}_{21}+w z_1
 (\bar{a}_{22}-\bar{a}_{11})-z_1^2 \bar{a}_{12})}{h (z_1-z_2){}^2 \tilde{w} (a_1-a_2+h)},\\
T_{1,2} y_2 = \frac{(a_2-a_1) (\hat{a}_{11} w+\hat{a}_{12} z_2)}{(z_2-z_1) \tilde{w} (-a_1+a_2+h)}-\frac{(a_2-a_1){}^2 (\hat{a}_{21} w^2+(\hat{a}_{22}-\hat{a}_{11}) w
 z_2-\hat{a}_{12} z_2^2)}{h (z_2-z_1){}^2 \tilde{w} (-a_1+a_2+h)},\\
T_{1,2} w_1 = -z_2 \frac{T_{1,2} w}{w}, \qquad T_{1,2} w_2 = -z_1 \frac{T_{1,2} w}{w}, \qquad T_{1,2}w = w+ \frac{(a_1-a_2) h w}{z_1-z_2},
\end{gather*}
where we use the notation $\bar{a}_{ij} = a_{ij} (a_1 +h)$ and $\hat{a}_{ij} = a_{ij}(a_2+h)$.
\end{prop}

While this has been written as a 6-dimensional map, the action and constraints in the $u$ variables tells us there are $4$ invariants. For example, we could determine expressions for $A(x)$ in terms of the two values $z_1$ and $z_2$ and $w$, specif\/ied by
\begin{gather*}
z_1 = w u_1, \qquad z_2 = w \left( u_1 + \frac{a_1-a_2}{u_1 + u_2} \right),
\end{gather*}
where $w$ is determined by the coef\/f\/icient of $x^2$ in the top right entry of $A(x)$. The resulting mapping is a dif\/ference equation that sits above the $d$-$\mathrm{P}\big(A_2^{(1)}\big)$ and under some rigidif\/ication, it is clearer that the compactifying the moduli space of linear dif\/ference equations is indeed $\mathbb{P}_2$ blown up at $9$ points, which has been the subject of one of the authors work~\cite{Rains2013}.

Another two-dimensional mapping may be obtained by allowing $A(x)$ to symmetric with respect to the change $x \to -x$, in which case we may allow $A(x)$ to be given by~\eqref{symmdiff} and~$B(x)$ to be given by~\eqref{hAprodform8} subject to the same constraints. It is easy to show that, under the conditions, that
\begin{gather*}
(E_{1,2} - I) d_{2,1} = 0, \qquad (F_{1,2} - I) d_{2,1} = 0,
\end{gather*}
indicating that \eqref{A0form} is a valid parameterization of $B(x)$ that is invariant under the actions $E_{i,j}$ and $F_{i,j}$.

\begin{prop}
The maps $E_{1,2}$ and $F_{1,2}$ on the variables $(z_1,z_2,w_1,w_2,y_1,y_2)$ are given by
\begin{gather*}\allowdisplaybreaks
E_{1,2} z_1 =z_2 \frac{E_{1,2} w}{w},\qquad E_{1,2} w_1 = E_{1,2} w \frac{z_1a_{11}(-h-a_1) - wa_{21}(-h-a_1)}{z_1a_{12}(-h-a_1) - wa_{22}(-h-a_1)},\\
E_{1,2} z_2 =z_1 \frac{E_{1,2} w}{w},\qquad E_{1,2} w_2 = E_{1,2} w \frac{z_2a_{11}(-h-a_2) - wa_{21}(-h-a_2)}{z_2a_{12}(-h-a_2) - wa_{22}(-h-a_2)},\\
E_{1,2} y_1 = \frac{(a_1-a_2) (w a_{22}(-a_1-h)-z_1 a_{12}(-a_1-h))}{(z_2-z_1) \tilde{w} (2 a_1+h)},\\
E_{1,2} y_2 = \frac{(a_2-a_1)(w a_{22}(-a_2-h)-z_2 a_{12}(-a_2-h))}{(z_1-z_2) \tilde{w} (2 a_2+h)},
\end{gather*}
and
\begin{gather*}\allowdisplaybreaks
F_{1,2} z_1 = \frac{F_{1,2}w (a_{21}(-a_1) w-a_{22}(-a_1) w_1)}{a_{11}(-a_1) w-a_{12}(-a_1) w_1},\qquad F_{1,2} w_1 = w_2 \frac{F_{1,2} w}{w},\\
F_{1,2} z_2 = \frac{F_{1,2}w (a_{21}(-a_2) w-a_{22}(-a_2) w_2)}{a_{11}(-a_2) w-a_{12}(-a_2) w_2}, \qquad F_{1,2} w_2 = w_1 \frac{E_{1,2} w}{w},\\
F_{1,2} y_1 = \frac{(a_1-a_2) (w a_{22}(-a_1-h)-z_1 a_{12}(-a_1-h))}{(z_2-z_1) \tilde{w} (2 a_1+h)},\\
F_{1,2} y_2 = \frac{(a_2-a_1) (w a_{22}(-a_2-h)-z_1 a_{12}(-a_1-h))}{(z_2-z_1) \tilde{w} (2 a_1+h)}.
\end{gather*}
\end{prop}

The map $E_{1,2} \circ F_{1,2}$ is once again 2-dimensional and specializes to $T_{1,2}$. This map is also acting on a surface obtained by blowing up $\mathbb{P}_2$ at~$9$ points, hence, this map coincides with $d$-$\mathrm{P}\big(A_0^{(1)}\big)$. We seek to establish a more explicit correspondence with well established versions of $d$-$\mathrm{P}\big(A_0^{(1)}\big)$ the future.

\subsection[An extra special case of the $m=3$ asymmetric and symmetric $q$-dif\/ference Garnier system]{An extra special case of the $\boldsymbol{m=3}$ asymmetric\\ and symmetric $\boldsymbol{q}$-dif\/ference Garnier system}

Let us consider the multiplicative version of the previous section, where $A(x)$ is given by the product
\begin{gather}
A(x) = \begin{pmatrix}
1 & u_1 \vspace{1mm}\\
 \dfrac{x}{a_1^2u_1} & 1
\end{pmatrix}\begin{pmatrix}
1 & u_2 \vspace{1mm}\\
 \dfrac{x}{a_2^2u_2} & 1
\end{pmatrix} \begin{pmatrix}
1 & u_3 \vspace{1mm}\\
 \dfrac{x}{a_3^2u_3} & 1
\end{pmatrix} \begin{pmatrix}
1 & u_4 \vspace{1mm}\\
 \dfrac{x}{a_4^2u_4} & 1
\end{pmatrix}\nonumber\\
\hphantom{A(x) =}{} \times \begin{pmatrix}
1 & u_1 \vspace{1mm}\\
 \dfrac{x}{a_5^2u_5} & 1
\end{pmatrix}\begin{pmatrix}
1 & u_6 \vspace{1mm}\\
 \dfrac{x}{a_6^2u_6} & 1
\end{pmatrix} \begin{pmatrix}
1 & u_7 \vspace{1mm}\\
 \dfrac{x}{a_7^2u_7} & 1
\end{pmatrix} \begin{pmatrix}
1 & u_8 \vspace{1mm}\\
 \dfrac{x}{a_8^2u_8} & 1
\end{pmatrix}.\label{Aprod8}
\end{gather}
As discussed above, if the $u_i$ variables satisfy
\begin{gather*}
u_1 + u_2 + u_3 + u_4 + u_5 + u_6 + u_7 + u_8,
\end{gather*}
then $A_0 = I$. With $\kappa_1$ and $\kappa_2$ specif\/ied by \eqref{kappavals}, then
\begin{gather*}
A_4 = \begin{pmatrix} \kappa_1 & 0 \\ \kappa_{2,1} & \kappa_2 \end{pmatrix}.
\end{gather*}
If $\kappa_1 = \kappa_2$ then we f\/ind that
\begin{gather*}
T_{1,2} \kappa_{2,1} = \frac{a_1 u_1^2}{qa_2u_2^2} \kappa_{2,1},
\end{gather*}
which means that if $\kappa_{2,1} = 0$ then $T_{1,2} \kappa_{2,1} = 0$. For similar reasons as the previous section, this def\/ines a two-dimensional mapping. Using a similar approach as the previous section, we may specify that the matrix $A(x)$ takes the general form
\begin{gather}
A(x) = \frac{(x-a_1)(x-a_2)}{a_1a_2} \begin{pmatrix} \alpha_2 x^2 + \alpha_1 x + 1 & x w\vspace{1mm}\\
\dfrac{\gamma}{w} & \beta_2x^2 + \beta_1 x + 1 \end{pmatrix}\nonumber\\
\hphantom{A(x) =}{} + \frac{x(x-a_2)}{a_1(a_1 - a_2)} y_1 \begin{pmatrix} 1 \vspace{1mm}\\ \dfrac{z_1}{w} \end{pmatrix} \begin{pmatrix} w_1 & w \end{pmatrix} + \frac{x(x-a_1)}{a_2(a_2 - a_1)} y_2 \begin{pmatrix} 1 \vspace{1mm}\\ \dfrac{z_2}{w} \end{pmatrix} \begin{pmatrix} w_2 & w \end{pmatrix}, \label{qBformA0A1}
\end{gather}
with $\kappa_1 = \kappa_2$ and the \eqref{detq} implying that
\begin{gather*}
\alpha_2 = \beta_2 = \sqrt{\frac{a_1a_2}{a_3a_4a_5a_6a_7}},
\end{gather*}
in which the mapping $T_{1,2}$ is may be computed accordingly.

\begin{prop}
The map $T_{1,2}$ on the variables $(z_1,z_2,w_1,w_2,y_1,y_2)$ is given by
\begin{gather*}\allowdisplaybreaks
T_{1,2} z_1 =\frac{z_2 T_{1,2}w}{w}+
\big\{a_1 (q-1) (z_2-z_1) T_{1,2}w (w (\bar{a}_{21} w+\bar{a}_{22} z_1)\\
\hphantom{T_{1,2} z_1 =}{} -z_2 (\bar{a}_{11} w+\bar{a}_{12} z_1))\big\}\big/\big\{w \big(a_1 \big(w z_1 (\bar{a}_{22}-\bar{a}_{11}
 q)+(q-1) z_2 (\bar{a}_{11} w+\bar{a}_{12} z_1)\\
\hphantom{T_{1,2} z_1 =}{}
 -\bar{a}_{12} q z_1^2+\bar{a}_{21} w^2\big)+a_2 \big((\bar{a}_{11}-\bar{a}_{22}) w z_1+\bar{a}_{12}
 z_1^2-\bar{a}_{21} w^2\big)\big)\big\}, \\
T_{1,2} z_2 =\frac{z_1 T_{1,2}w}{w}+
\big\{a_2 (q-1) (z_1-z_2) T_{1,2}w (w (\hat{a}_{21} w+\hat{a}_{22} z_2)\\
\hphantom{T_{1,2} z_2 =}{}
-z_1 (\hat{a}_{11} w+\hat{a}_{12} z_2))\big\}\big/\big\{w \big(a_2 \big(w z_2 (\hat{a}_{22}-\hat{a}_{11}
 q)+(q-1) z_1 (\hat{a}_{11} w+\hat{a}_{12} z_2)\\
\hphantom{T_{1,2} z_2 =}{}
 -\hat{a}_{12} q z_2^2+\hat{a}_{21} w^2\big)+a_1 \big((\hat{a}_{11}-\hat{a}_{22}) w z_2+\hat{a}_{12}
 z_2^2-\hat{a}_{21} w^2\big)\big)\big\},\\
T_{1,2}w_1 = z_2 \frac{T_{1,2}w}{w}, \qquad\! T_{1,2}w_2 = z_1 \frac{T_{1,2}w}{w},\qquad\! T_{1,2} w =\frac{q w ((a_1-a_2) (\alpha _2 q-\alpha _2)+z_1-z_2)}{z_1-z_2},\\
T_{1,2}y_1 = \frac{(a_2-a_1){}^2 ((\bar{a_{11}}-\bar{a_{22}}) w z_1+\bar{a_{12}} z_1^2-\bar{a_{21}} w^2)}{a_1 (q-1) (z_2-z_1){}^2 \tilde{w} (a_1
 q-a_2)}-\frac{(a_2-a_1) (\bar{a_{11}} w+\bar{a_{12}} z_1)}{(z_2-z_1) \tilde{w} (a_2-a_1 q)}, \\
T_{1,2}y_2 = \frac{(a_1-a_2){}^2 ((\hat{a_{11}}-\hat{a_{22}}) w z_2+\hat{a_{12}} z_2^2-\hat{a_{21}} w^2)}{a_2 (q-1) (z_1-z_2){}^2 \tilde{w} (a_2
 q-a_1)}-\frac{(a_1-a_2) (\hat{a_{11}} w+\hat{a_{12}} z_2)}{(z_1-z_2) \tilde{w} (a_1-a_2 q)},
\end{gather*}
where $\bar{a}_{ij} = a_{ij}(q a_1)$ and $\hat{a}_{ij} = a_{ij}(q a_2)$.
\end{prop}

This map induces a map on the two-dimensional moduli space of linear systems of $q$-dif\/ference equations and also sits above the $q$-$\mathrm{P}\big(A_2^{(1)}\big)$ case, hence, we naturally expect this to coincide with the $q$-$\mathrm{P}\big(A_1^{(1)}\big)$ case.

We obtain a distinct two-dimensional mapping when we allow the matrix $Y(x)$ to be symmetric with respect to the change $x \to 1/x$, in which case we have~$A(x)$ given by~\eqref{symmqdiff} and~$B(x)$ given by~\eqref{Aprod8}. The f\/irst thing to check is that if $\kappa_1 = \kappa_2$, then $E_{1,2} \kappa_1 = E_{1,2} \kappa_2$ and~$F_{1,2} \kappa_1 = F_{1,2} \kappa_2$, which is easily done. It is also easy to check that if $\kappa_{2,1} = 0$ then $E_{1,2} \kappa_{2,1} = 0$ and $F_{1,2} \kappa_{2,1} = 0$. This ensures that the mappings $E_{1,2}$ and $F_{1,2}$ may be applied to a~mat\-rix~$B(x)$ that takes the form~\eqref{qBformA0A1}.

\begin{prop}
The maps $E_{1,2}$ and $F_{1,2}$ on the variables $(z_1,z_2,w_1,w_2,y_1,y_2)$ are given by
\begin{gather*}\allowdisplaybreaks
E_{1,2} z_1 = z_2 \frac{E_{1,2}w}{w}, \qquad E_{1,2}w_1 = \frac{w \left(a_{21}\left(\frac{1}{a_1}\right) w-a_{11}\left(\frac{1}{a_1}\right) z_1\right)}{a_{22}\left(\frac{1}{a_1}\right) w-a_{12}\left(\frac{1}{a_1}\right) z_1},\\
E_{1,2} z_2 = z_1 \frac{E_{1,2}w}{w}, \qquad E_{1,2}w_1 = \frac{w \left(a_{21}\left(\frac{1}{a_1}\right) w-a_{11}\left(\frac{1}{a_1}\right) z_1\right)}{a_{22}\left(\frac{1}{a_1}\right) w-a_{12}\left(\frac{1}{a_1}\right) z_1},\\
E_{1,2}wy_1 = \frac{a_1 a_2 a_{22}\left(\frac{1}{a_1}\right) w \left(1-a_1^2 (q-1)\right)}{\big(a_1^2-1\big) (a_1 a_2-1) z_2 }
-\frac{a_1^2 \left(a_1-a_2\right) a_2 q \left(a_{22}\left(\frac{1}{a_1}\right)
 w-a_{12}\left(\frac{1}{a_1}\right) z_1\right)}{\big(a_1^2-1\big) (a_1 a_2-1 ) (z_1-z_2 ) }, \\
 E_{1,2}wy_2 = \frac{a_1 a_2 a_{22}\left(\frac{1}{a_2}\right) w \big(1-a_2^2 (q-1)\big)}{\big(a_2^2-1\big) (a_1 a_2-1) z_1 }
 -\frac{a_2^2 (a_2-a_1) a_1 q \left(a_{22}\left(\frac{1}{a_2}\right)
 w-a_{12}\left(\frac{1}{a_2}\right) z_2\right)}{\big(a_2^2-1\big) (a_1 a_2-1) (z_2-z_1) },
\end{gather*}
and
\begin{gather*}\allowdisplaybreaks
F_{1,2} w_1 = w_2 \frac{F_{1,2}w}{w}, \qquad F_{1,2} z_1 = F_{1,2}w \frac{w a_{21}\left(\frac{1}{a_1 q}\right)-w_1 a_{22}\left(\frac{1}{a_1 q}\right)}{w a_{11}\left(\frac{1}{a_1 q}\right)-w_1 a_{12}\left(\frac{1}{a_1 q}\right)},\\
F_{1,2} w_2 = w_1 \frac{F_{1,2}w}{w}, \qquad F_{1,2} z_2 = F_{1,2}w \frac{w a_{21}\left(\frac{1}{a_2 q}\right)-w_2 a_{22}\left(\frac{1}{a_2 q}\right)}{w a_{11}\left(\frac{1}{a_2 q}\right)-w_2 a_{12}\left(\frac{1}{a_2 q}\right)},\\
F_{1,2}wy_1 = \frac{a_1 a_2 q \big(1-a_1^2 (q-1) q\big) a_{12}\left(\frac{1}{a_1 q}\right)}{\big(a_1^2 q-1\big) (a_1 a_2 q-1)} \\
\hphantom{F_{1,2}wy_1 =}{}
-\frac{a_1^2 (a_1-a_2) a_2 q^3 \left(w a_{11}\left(\frac{1}{a_1
 q}\right)+z_2 a_{12}\left(\frac{1}{a_1 q}\right)\right)}{(z_1-z_2) \big(a_1^2 q-1\big) (a_1 a_2 q-1)},\\
 F_{1,2}wy_2 = \frac{a_1 a_2 q \left(1-a_2^2 (q-1) q\right) a_{12}\left(\frac{1}{a_2 q}\right)}{\big(a_2^2 q-1\big) (a_1 a_2 q-1)} \\
 \hphantom{F_{1,2}wy_2 =}{}
 -\frac{a_2^2 (a_2-a_1) a_1 q^3 \left(w a_{11}\left(\frac{1}{a_2
 q}\right)+z_1 a_{12}\left(\frac{1}{a_2 q}\right)\right)}{(z_2-z_1) \big(a_2^2 q-1\big) (a_1 a_2 q-1)}.
\end{gather*}
\end{prop}

By construction, this is a map that sits above the case of $q$-$\mathrm{P}\big(A_1^{(1)}\big)$, hence, it should be $q$-$\mathrm{P}\big(A_0^{(1)}\big)$. We seek to establish a more explicit correspondence with well established versions of $q$-$\mathrm{P}\big(A_0^{(1)}\big)$ in future works.

\section{Reductions of partial dif\/ference equations}\label{sec:reductions}

One of the consequences of this work is that we will able to show that the $q$-Garnier system of \cite{Sakai:Garnier} arises as a set of reduction the lattice Schwarzian Korteweg--de Vries equation. By specializing the $h$-Garnier systems, this also means that $q$-$\mathrm{P}\big(A_1^{(1)}\big)$ and $d$-$\mathrm{P}\big(A_1^{(1)}\big)$ also arise as reductions of the lattice Schwarzian Korteweg--de Vries equation and lattice potential Korteweg--de Vries equation.

The general setting for reductions of partial dif\/ference equations on the quad(rilateral) is that we take a function $w \colon \mathbb{Z}^2 \to \mathbb{C}$, whose values are denoted $w_{l,m}$, in which for every $(l,m)\in \mathbb{Z}^2$ we impose the constraint
\begin{gather}\label{quad}
Q(w_{l,m}, w_{l+1,m}, w_{l,m+1}, w_{l+1,m+1}; \alpha_l, \beta_m )=0,
\end{gather}
where $Q$ is linear in each of the variables. Given a staircase of initial conditions, we are able to determine each value on $\mathbb{Z}^2$, hence, we require an inf\/inite number of initial conditions to specify a solution \cite{VanderKamp:IVPs}. For this reason, these systems are commonly referred to as inf\/inite-dimensional systems, and are considered to be the discrete analogues of partial dif\/ferential equations.

The twisted $(n_1,n_2)$-reduction is the system of solutions that satisfy an additional relation of the form
\begin{gather}\label{twist}
w_{l+n_1, m+n_2} = T(w_{l,m}),
\end{gather}
where the function, $T$, is called the twist \cite{Ormerod2014b}. We require that \eqref{quad} is invariant under $T$, i.e., we require that
\begin{gather*}
Q(w_{l,m}, w_{l+1,m}, w_{l,m+1}, w_{l+1,m+1}; \alpha_l, \beta_m )=0 \\
\qquad {} \iff \ Q(Tw_{l,m}, Tw_{l+1,m}, Tw_{l,m+1}, Tw_{l+1,m+1}; \alpha_l, \beta_m )=0.
\end{gather*}
The staircase of initial conditions for a twisted reduction consists of the $n_1 + n_2$ initial conditions in some f\/inite staircase extended inf\/initely in both directions using \eqref{twist}. Secondly, we require that the parameters change in a way that if two points are related by~\eqref{twist}, then the points calculated using those $n_1 + n_2$ initial conditions are also related by~\eqref{twist}. This means we require
\begin{gather}
Q(w_{l,m}, w_{l+1,m}, w_{l,m+1}, w_{l+1,m+1}; \alpha_l, \beta_m )=0\nonumber\\
\qquad \iff \ Q(w_{l,m}, w_{l+1,m}, w_{l,m+1}, w_{l+1,m+1}; \alpha_{l+n_1}, \beta_{m+n_2})=0. \label{parchange}
\end{gather}
The resulting system may be described by a $(n_1 + n_2)$-dimensional map, which we call a twisted reduction of~\eqref{quad}~\cite{Ormerod2014a}. If $\alpha_l = \alpha_{l+n_1}$ and $\beta_m = \beta_{m+n_2}$ then the resulting ordinary dif\/ference equation is necessarily autonomous, otherwise, the system is non-autonomous. In the special case that $T$ is the identity, we call the reduction a periodic reduction.

\begin{figure}[!ht]\centering
\begin{tikzpicture}[scale=1]
\draw[black!50] (-.3,-.3) grid (8.3,4.4);
\draw[very thick,red] (1,1) -- (4,1) -- (4,2)--(6,2) --(6,3)-- (7,3);
\draw[very thick,red!40] (1,1) -- (0,1)--(0,0) -- (-.3,0);
\draw[very thick,red!40] (7,3) -- (8.3,3);
\filldraw[fill=blue,draw=black] (1,1) circle (.08cm);
\filldraw[fill=blue,draw=black] (7,3) circle (.08cm);
\filldraw[fill=red,draw=black] (2,1) circle (.08cm);
\filldraw[fill=red,draw=black] (3,1) circle (.08cm);
\filldraw[fill=red,draw=black] (4,1) circle (.08cm);
\filldraw[fill=red,draw=black] (4,2) circle (.08cm);
\filldraw[fill=red,draw=black] (5,2) circle (.08cm);
\filldraw[fill=red,draw=black] (6,2) circle (.08cm);
\filldraw[fill=red,draw=black] (6,3) circle (.08cm);
\filldraw[fill=red!20,draw=black] (0,1) circle (.08cm);
\filldraw[fill=red!20,draw=black] (0,0) circle (.08cm);
\filldraw[fill=red!20,draw=black] (8,3) circle (.08cm);
\end{tikzpicture}
\caption{The intial conditions specif\/ied for a $(6,2)$-reduction, where the points in blue are related via~\eqref{twist}.}
\end{figure}
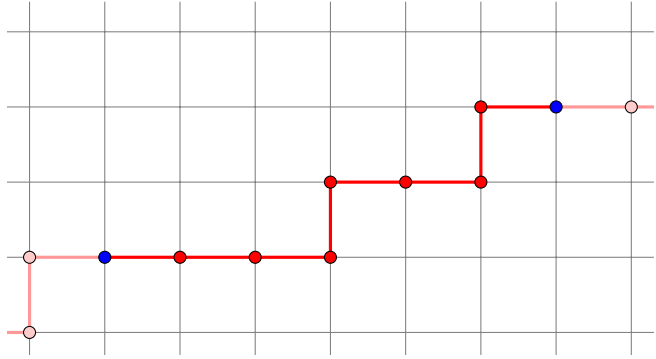

One def\/inition of integrability for systems of the form \eqref{quad} is 3-dimensional consistency. If we impose a constraint of the form \eqref{quad} on each of the faces of a cube, then 3-dimensional consistency requires that each way of determining the values on the vertices of the cube agree~\cite{Nijhoff:CAC}. A~classif\/ication of 3-dimensionally consistent multilinear equations of the form of \eqref{quad} was the subject of the classif\/ication of Adler et al.~\cite{ABS:ListI, ABS:ListII}.

We consider two equations of the form \eqref{quad}; the lattice potential Korteweg--de Vries equa\-tion~\cite{Nijhoff:lkdvreview},
\begin{gather}\label{lpkdv}
(w_{l,m} - w_{l+1,m+1})(w_{l+1,m} - w_{l,m+1}) = \alpha_l - \beta_m,
\end{gather}
which is also known as H1 in~\cite{ABS:ListI, ABS:ListII} and the lattice Schwarzian Korteweg--de Vries equation \cite{Nijhoff:dSKdV},
\begin{gather}
\alpha _l \left(\frac{1}{w_{l,m+1}-w_{l+1,m+1}}+\frac{1}{w_{l+1,m}-w_{l,m}}\right)\nonumber\\
\qquad {} = \beta _m \left(\frac{1}{w_{l+1,m}-w_{l+1,m+1}}+\frac{1}{w_{l,m+1}-w_{l,m}}\right), \label{dSKdV}
\end{gather}
which is also known as $\mathrm{Q1}_{\delta= 0}$ in \cite{ABS:ListI, ABS:ListII}.

It is easy to see that~\eqref{lpkdv} is invariant under translational twists, i.e., those of the form $T(u) = u + \lambda$ whereas~\eqref{dSKdV} is invariant under any twist in the full group of invertible M\"obius transformations. Secondly, we see that~\eqref{parchange} holds for~\eqref{lpkdv} if
\begin{gather}\label{parchangeh}
\alpha_{l + n_1} = \alpha_l + h, \qquad \beta_{m+n_2} = \beta_m + h,
\end{gather}
and \eqref{parchange} holds for \eqref{dSKdV} if
\begin{gather*}
\alpha_{l + n_1} = q\alpha_l, \qquad \beta_{m+n_2} = q\beta_m,
\end{gather*}

A Lax pair in the context of equations of the form of \eqref{quad} is a pair of equations of the form
\begin{gather*}
\Phi(l+1,m) = L_{l,m} \Phi(l,m),\qquad \Phi(l,m+1) = M_{l,m} \Phi(l,m),
\end{gather*}
whose compatibility reads
\begin{gather*}
L_{l,m+1}M_{l,m} - M_{l+1,m} L_{l,m} = 0.
\end{gather*}
The matrices, $L_{l,m}$ and $M_{l,m}$, for the lattice potential Korteweg--de Vries equation are
\begin{subequations}\label{LaxdpKdV}
\begin{gather}
L_{l,m} = \begin{pmatrix} w_{l+1,m} - w_{l,m} & 1 \\ \gamma -\alpha_l + (w_{l+1,m} - w_{l,m})^2 & w_{l+1,m} - w_{l,m} \end{pmatrix},\\
M_{l,m} = \begin{pmatrix} w_{l,m+1} - w_{l,m} & 1 \\ \gamma -\beta_m + (w_{l,m+1} - w_{l,m})^2 & w_{l,m+1} - w_{l,m} \end{pmatrix},
\end{gather}
\end{subequations}
which are both of the form \eqref{difffactor} for $u = w_{l+1,m} - w_{l,m}$ and $u = w_{l,m+1} - w_{l,m}$. The matrices, $L_{l,m}$ and $M_{l,m}$, for the lattice Schwarzian Korteweg--de Vries equation are
\begin{subequations}\label{LaxdSKdV}
\begin{gather}
L_{l,m} = \begin{pmatrix} 1&w_{l+1,m} - w_{l,m} \vspace{1mm}\\ \dfrac{\gamma}{\alpha_l(w_{l+1,m} - w_{l,m})}& 1 \end{pmatrix},\\
M_{l,m} = \begin{pmatrix} 1&w_{l,m+1} - w_{l,m} \vspace{1mm}\\ \dfrac{\gamma}{\beta_m(w_{l,m+1} - w_{l,m})} & 1 \end{pmatrix},
\end{gather}
\end{subequations}
which are also both of the form \eqref{qdifffactor} for $u = w_{l+1,m} - w_{l,m}$ and $u = w_{l,m+1} - w_{l,m}$. This determines a well known relation between integrable equations of the form~\eqref{quad} and Yang--Baxter maps \cite{Papageorgiou2006}. The general framework for determining Lax pairs for ordinary dif\/ference equations arising as twisted reductions of partial dif\/ference was recently outlined in \cite{Ormerod2014b}.

From the point of view of symmetries of reductions~\cite{Ormerod2014b}, it is slightly more conducive to regard a twisted $(n_1,n_2)$-reduction as a reduction on an $(n_1 + n_2)$-dimensional hypercube~\cite{Joshi2014}. The symmetries of the reductions arise from dif\/ferent paths on this hypercube from the points connected via~\eqref{twist}.

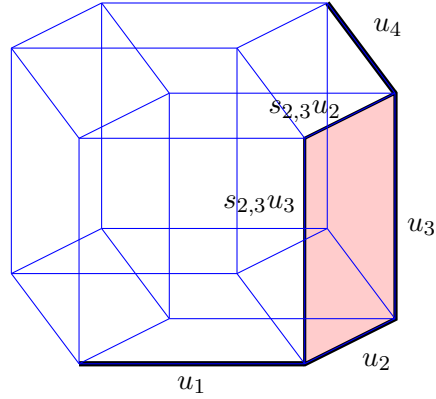
\begin{figure}[t]\centering
\begin{tikzpicture}[scale=3]
\draw[line width=2pt,black] (0,0) -- (1,0) node[below, midway] {$u_1$} -- (1.4,.2) node[below right, midway] {$u_2$} -- (1.4,1.2) node[below right, midway] {$u_3$} -- (1.1,1.6) node[above right, midway] {$u_4$} ;
\draw [fill=red!20,very thick] (1,0)--(1.4,.2) -- (1.4,1.2) -- (1,1) -- cycle;
\draw[blue] (0.,0.)--(-0.3,0.4);
\draw[blue] (0.,0.)--(0.4,0.2);
\draw[blue] (0.,0.)--(1.,0.);
\draw[blue] (0.,0.)--(0.,1.);
\draw[blue] (-0.3,0.4)--(0.1,0.6);
\draw[blue] (-0.3,0.4)--(0.7,0.4);
\draw[blue] (-0.3,0.4)--(-0.3,1.4);
\draw[blue] (0.4,0.2)--(0.1,0.6);
\draw[blue] (0.4,0.2)--(1.4,0.2);
\draw[blue] (0.4,0.2)--(0.4,1.2);
\draw[blue] (0.1,0.6)--(1.1,0.6);
\draw[blue] (0.1,0.6)--(0.1,1.6);
\draw[blue] (1.,0.)--(0.7,0.4);
\draw[blue] (1.,0.)--(1.4,0.2);
\draw[blue] (1.,0.)--(1.,1.);
\draw[blue] (0.7,0.4)--(1.1,0.6);
\draw[blue] (0.7,0.4)--(0.7,1.4);
\draw[blue] (1.4,0.2)--(1.1,0.6);
\draw[blue] (1.4,0.2)--(1.4,1.2);
\draw[blue] (1.1,0.6)--(1.1,1.6);
\draw[blue] (0.,1.)--(-0.3,1.4);
\draw[blue] (0.,1.)--(0.4,1.2);
\draw[blue] (0.,1.)--(1.,1.);
\draw[blue] (-0.3,1.4)--(0.1,1.6);
\draw[blue] (-0.3,1.4)--(0.7,1.4);
\draw[blue] (0.4,1.2)--(0.1,1.6);
\draw[blue] (0.4,1.2)--(1.4,1.2);
\draw[blue] (0.1,1.6)--(1.1,1.6);
\draw[blue] (1.,1.)--(0.7,1.4);
\draw[blue] (1.,1.)--(1.4,1.2);
\draw[blue] (0.7,1.4)--(1.1,1.6);
\draw[blue] (1.4,1.2)--(1.1,1.6);
\node at (.8,.7) {$s_{2,3} u_3$};
\node at (1,1.12) {$s_{2,3} u_2$};
\end{tikzpicture}
\caption{The reduction associated with $q$-$\mathrm{P}_{\rm VI}$ visualized on a hypercubic lattice. The relations between the values of $s_{2,3} u_2$ and $s_{2,3} u_3$ are def\/ined in terms of the highlighted face on the cube.}
\end{figure}

The key to constructing Lax pairs for periodic reductions is that we have two parameters, $\alpha_l$ and $\beta_m$, whereas the reductions of~\eqref{lpkdv} and~\eqref{dSKdV} depend upon a single variable $t= \alpha_l - \beta_m$ and $t = \alpha_l/\beta_m$ respectively, which is constant with respect to shifts $(l,m) \to (l+n_1,m + n_2)$. We simply need to choose a spectral variable that is not constant with respect to the shift $(l,m) \to (l+n_1,m + n_2)$.

Let us f\/irst treat the $h$-dif\/ference case. The correspondence between the discrete Garnier systems is made simple by taking periodic reductions of~\eqref{lpkdv} with
\begin{gather}\label{parchoiceh}
x = \alpha_l + a_l, \qquad t = \alpha_l - \beta_m + b_m,
\end{gather}
where $a_l$ and $b_m$ are $n_1$-periodic and $n_2$-periodic functions of $l$ and $m$ respectively. Note that an operator that shifts $(l,m) \to (l+n_1,m + n_2)$ by~\eqref{parchangeh} has the ef\/fect of f\/ixing $t$, and has the ef\/fect of shifting $x \to x+h$. The operator that shifts $(l,m) \to (l,m - n_2)$ f\/ixes $x$ and shifts $t \to t+h$. A matrix inducing the shift in $x$ may be written as
\begin{gather*}
A(x,t) = M_{l+n_1,m+n_2-1}\cdots M_{l+n_1,m} L_{l,n_1-1} \cdots L_{l,m},
\end{gather*}
where $L_{l,m}$ and $M_{l,m}$ are given by~\eqref{LaxdpKdV} and we have assumed the correspondence between $\alpha_l$ and $\beta_m$ and $x$ and $t$ is given by~\eqref{parchoiceh}. By writing~$A(x,t)$ in this way, we have chosen a path of initial conditions represents an L-shaped path. For the other operator, we have
\begin{gather*}
R(x,t) = M_{l,m-n_2}^{-1} \cdots M_{l,m-1}^{-1},
\end{gather*}
which brings us to the following result.

\begin{thm}
The $h$-Garnier system, as defined by the mapping~\eqref{hGarnierAction}, arises as a periodic $(N-2,2)$-reduction of~\eqref{lpkdv}.
\end{thm}

\begin{proof}
We start by showing that $T_{1,2}$ arises as a periodic reduction. Up to relabeling for some f\/ixed $l$, $m$, we let
\begin{gather*}
b_{m+1} = a_1,\qquad b_m = a_2,\qquad a_{l+N-3} = a_3,\qquad \ldots, \qquad a_l = a_N,
\end{gather*}
which is extended periodically with periods $2$ and $N-2$. For that same f\/ixed $l$, $m$ we let
\begin{gather*}
u_1 = w_{l+ N-2, m+2} - w_{l+ N-2, m+ 1}, \qquad u_2 = w_{l+ N-2,m+ 2} - w_{l+ N-2, m+ 1},\\
u_3 = w_{l+ N-2, m} - w_{l+ N-1, m}, \qquad \ldots, \qquad u_N = w_{l+ 1, m} - w_{l, m},
\end{gather*}
which, due to \eqref{twist} in the case that $T$ is the identity (i.e., the periodic case), then these values are also extended periodically in the lattice. Furthermore, the constraint
\begin{gather*}
u_1 + \dots + u_N = w_{l+N-2,m+2} - w_{l,m} = 0,
\end{gather*}
by \eqref{twist}. Under this labelling of the initial conditions the operator $A(x,t)$ is of the form~\eqref{prodform} where each factor is of the form \eqref{difffactor}. Furthermore, due to the periodicity, we have that
\begin{gather*}
R(x,t) = L(x,u_2,a_2+t+h)^{-1} L(x,u_1,a_1+t+h)^{-1} = L_2(x-h)^{-1} L_1(x-h)^{-1}.
\end{gather*}
The compatibility, given by \eqref{comp} where $\tilde{A}(x,t) = A(x,t+h)$ coincides with the computations in Proposition~\ref{tranha1a2}. To complete the correspondence, one notices that the action of~$S_n$ def\/ined by~\eqref{FV} is equivalent to using~\eqref{lpkdv} to def\/ine a dif\/ferent path of initial conditions.
\end{proof}

Due to the equivalence between \eqref{hGarnierAction} and the birational form, namely \eqref{hGarnierBirrational}, also arises as a~reduction, as do any of the special cases that have arisen in Section~\ref{sec:special}.

The correspondence between the $q$-Garnier systems is made simple by taking twisted reductions of \eqref{dSKdV} with
\begin{gather}\label{parchoiceq}
x = \alpha_l, \qquad t = \alpha_l/\beta_m.
\end{gather}
We may construct the matrix connecting points connected via \eqref{twist}, which is given by
\begin{gather*}
A'(x,t) = M_{l+n_1,m+n_2-1}\cdots M_{l+n_1,m} L_{l,n_1-1} \cdots L_{l,m},
\end{gather*}
where the matrices, $L_{l,m}$ and $M_{l,m}$ are given by~\eqref{LaxdSKdV} and we have assumed the correspondence between $x$ and $t$ and $\alpha_l$ and $\beta_m$ are given by~\eqref{parchoiceq}. As noted in~\cite{Ormerod2014a}, this has a non-trivial ef\/fect on the solutions of the linear problem. This introduces a twist matrix whose ef\/fect is given by
\begin{gather*}
T Y(x,t) = S Y(x,t) ,
\end{gather*}
where $S$ is independent of the spectral parameter, $x$. It follows that the corresponding associated linear problem takes the form
\begin{gather*}
Y(qx,t) = A(x,t) Y(x,t),
\end{gather*}
where
\begin{gather*}
A(x,t) = S^{-1} A'(x,t),
\end{gather*}
where $A'(x,t)$ is as above. This is the same path, but for dif\/ferent matrices. Also, the matrix inducing the transformation~$T_{1,2}$ is also changed by the twist to
\begin{gather*}
R(x,t) = S M_{l,m-n_2}^{-1} \cdots M_{l,m-1}^{-1}
\end{gather*}
and is also in terms of the corresponding $L_{l,m}$ and $M_{l,m}$. As in~\cite{Ormerod2014a}, we may calculate the twist via
\begin{gather}\label{twistmatcalc}
T(R(x,t)) S = (T_{1,2} S) R(x,t),
\end{gather}
which only requires $R(x,t)$ to obtain.

\begin{thm}
The $q$-Garnier system, as defined by the mapping~\eqref{qGarnierAction}, arises as a twisted $(N-2,2)$-reduction of~\eqref{LaxdSKdV} where the twist is an affine linear transformation,
\begin{gather}\label{affine}
T(u) = \theta_2/\theta_1 u + b.
\end{gather}
\end{thm}

\begin{proof}
The twist matrix in the case of \eqref{affine} is given by the diagonal matrix
\begin{gather*}
S = \begin{pmatrix} \theta_1^{-1} & 0 \\ 0 & \theta_2^{-1} \end{pmatrix},
\end{gather*}
which is a solution to \eqref{twistmatcalc}, as guaranteed by the relation~\eqref{commDL}. The remainder of this proof follows in a similar manner to before; up to relabeling for some f\/ixed~$l$,~$m$, we let
\begin{gather*}
b_{m+1}= a_1,\qquad b_m = a_2,\qquad a_{l+N-3} = a_3,\qquad \ldots, \qquad a_l = a_N,
\end{gather*}
which is extended periodically with periods $2$ and $N-2$. For that same f\/ixed~$l$, $m$ we let
\begin{gather*}
u_1 = w_{l+ N-2, m+2} - w_{l+ N-2, m+ 1}, \qquad u_2 = w_{l+ N-2,m+ 2} - w_{l+ N-2, m+ 1},\\
u_3 = w_{l+ N-2, m} - w_{l+ N-1, m}, \qquad \ldots, \qquad u_N = w_{l+ 1, m} - w_{l, m}.
\end{gather*}
in which case the form of $A(x,t)$ is precisely given by \eqref{prodform} in which case the discrete isomonodromic deformations follow.
\end{proof}

\begin{cor}
The discrete Painlev\'e equations of types $q$-$\mathrm{P}\big(A_3^{(1)}\big)$, $q$-$\mathrm{P}\big(A_2^{(1)}\big)$ and $q$-$\mathrm{P}\big(A_1^{(1)}\big)$ arise as special cases of twisted reductions of \eqref{dSKdV}. Similarly, $d$-$\mathrm{P}\big(A_2^{(1)}\big)$ and $d$-$\mathrm{P}\big(A_1^{(1)}\big)$ arise as special cases of periodic reductions of~\eqref{lpkdv}.
\end{cor}

In the cases of discrete Painlev\'e equations of types $q$-$\mathrm{P}\big(A_3^{(1)}\big)$, $q$-$\mathrm{P}\big(A_2^{(1)}\big)$ and $d$-$\mathrm{P}\big(A_2^{(1)}\big)$, by replacing $u$ with a dif\/ference in $w_{l,m}$-values, we have obtained explicit correspondences between the variables parameterizing the lattice equation and the Painlev\'e variables. However, in the cases of $d$-$\mathrm{P}\big(A_1^{(1)}\big)$ and $q$-$\mathrm{P}\big(A_1^{(1)}\big)$, we do not have this. We only know that the moduli spaces of the relevant dif\/ference equations have anticanonical divisors with two irreducible components~\cite{Rains2013}.

At this point, we are unsure as to how the symmetric cases may arise as twisted or periodic reductions. The problem with these cases is that the reductions described require that the spectral parameter enters via~\eqref{prodform} in the same way for each factor, whereas in the symmetric cases, the product form requires some involution of the spectral parameter. It may be the case that this is more natural from the perspective of reductions of cases higher than~\eqref{lpkdv} or~\eqref{dSKdV}.

\section{Discussion}

By far the most interesting feature of the above Lax pairs is the existence of a symmetry, which manifests itself in the elliptic case quite naturally, which we will present separately. There are two interesting consequences we may derive from this work; the existence of symmetric Lax pairs for the lower discrete Painlev\'e equations and the existence of discrete symmetric Lax pairs possessing continuous isomonodromic deformations, which we shall present in another paper.

From the above work, we have shown that the $d$-$\mathrm{P}\big(A_1^{(1)}\big)$ and $q$-$\mathrm{P}\big(A_1^{(1)}\big)$ arise as reductions of partial dif\/ference equations, however, it is unclear at this point how to make the full correspondence between the $u$-variables and the Painlev\'e variables. We have also shown that the $q$-Garnier system def\/ined by Sakai in~\cite{Sakai:Garnier} arises as a reduction, but the problem of f\/inding reductions for the symmetric Garnier systems is not clear from this work.

\subsection*{Acknowledgements}

The work of EMR was partially supported by the National Science Foundation under the grant DMS-1500806.

\LastPageEnding

\end{document}